\documentclass[pra,groupedaddress,superscriptaddress,nofootinbib,floatfix,preprintnumbers,longbibliography,notitlepage,tightenlines,11pt]{revtex4-1}
\usepackage[paperwidth=210mm,paperheight=297mm,centering,hmargin=2.3cm,vmargin=2.3cm]{geometry}

\usepackage{amssymb,amsmath,graphicx,color,dsfont, mathrsfs 
,tikz-cd
} 
\usepackage{bm}
\usepackage[export]{adjustbox}
\usepackage{braket}
\usepackage[colorlinks=true,citecolor=blue,linkcolor=blue
]{hyperref}
\usepackage{amsfonts}
\usepackage{dsfont}
\usepackage{amsthm}
\usepackage{braket}
\usepackage{enumerate}
\usepackage{adjustbox}
\usepackage{geometry}
\usepackage{array}
\usepackage{xfrac}
\usepackage{braket,subcaption}
\usepackage{tikz}
\usepackage{quantikz}
\usepackage{url}
\usepackage{xcolor}
\usepackage{multirow}
\usepackage{rotating,booktabs}
\usepackage[verbose]{placeins}
\usepackage{ragged2e}
\usepackage[titletoc]{appendix}
\usepackage{soul}

\newcommand{\kett}[1]{\left. \left| #1 \right\rangle \right\rangle}
\newcommand{\bbra}[1]{\left\langle \left\langle #1 \right| \right.}
\newcommand{\bbrakett}[1]{\left\langle \braket{#1} \right\rangle}
\newcommand{\vecpq}{\overset{\longrightarrow}{p\, q}}

\newcommand\JS[1]{{\color{black}{#1}}}

\newtheorem{theorem}{Theorem}[section]

\newtheorem{lemma}[theorem]{Lemma}

\newtheorem{corollary}[theorem]{Corollary}

\begin{document}

\dimen\footins=5\baselineskip\relax

\preprint{IQuS@UW-21-086}

\title{Loop-string-hadron approach to SU(3) lattice Yang-Mills theory:\\ I. Hilbert space of a trivalent vertex}

\author{Saurabh V. Kadam}
\email{ksaurabh@uw.edu}
\affiliation{
InQubator for Quantum Simulation (IQuS), Department of Physics, University of Washington, Seattle, Washington 98195, USA}

\author{Aahiri Naskar}
 \email{p20230046@goa.bits-pilani.ac.in}
\affiliation{Department of Physics, BITS-Pilani,
K K Birla Goa Campus, Zuarinagar, Goa 403726, India.}

\author{Indrakshi Raychowdhury}
 \email{indrakshir@goa.bits-pilani.ac.in}
\affiliation{Department of Physics, BITS-Pilani,
K K Birla Goa Campus, Zuarinagar, Goa 403726, India.}
\affiliation{Center for Research in Quantum Information and Technology, Birla Institute of Technology and Science Pilani, K K Birla Goa Campus,  403726, India}

\author{Jesse R.~Stryker}
\email{jstryker@lbl.gov}
\thanks{leading contributor to analytic developments in this manuscript.}
\affiliation{
Physics Division, Lawrence Berkeley National Laboratory, Berkeley, California 94720, USA.}
\affiliation{
Maryland Center for Fundamental Physics and Department of Physics, 
University of Maryland, College Park, Maryland 20742, USA.}

\date{\today}

\begin{abstract}
The construction of gauge invariant states of SU(3) lattice gauge theories has garnered new interest in recent years, but implementing them is complicated by the need for SU(3) Clebsch-Gordon coefficients. In the loop-string-hadron (LSH) approach to lattice gauge theories, the elementary excitations are strictly gauge invariant, and constructing the basis requires no knowledge of Clebsch-Gordon coefficients. Originally developed for SU(2), the LSH formulation was recently generalized to SU(3), but limited to one spatial dimension. In this work, we generalize the LSH approach to constructing the basis of SU(3) gauge invariant states at a trivalent vertex -- the essential building block to multidimensional space. A direct generalization from the SU(2) vertex yields a legitimate basis; however, in certain sectors of the Hilbert space, the naive LSH basis vectors so defined suffer from being nonorthogonal. The issues with orthogonality are directly related to the “missing label” or “outer multiplicity” problem associated with SU(3) tensor products, and may also be phrased in terms of Littlewood-Richardson coefficients or the need for a ``seventh Casimir'' operator.
The states that are unaffected by the problem are orthonormalized in closed form.
For the sectors that are afflicted, we discuss the nonorthogonal bases and their orthogonalization. A few candidates for seventh Casimir operators are readily constructed from the suite of LSH gauge-singlet operators. The diagonalization of a seventh Casimir represents one prescriptive solution toward obtaining a complete orthonormal basis, but a closed-form general solution remains to be found.

\end{abstract}

\maketitle

\section{Introduction}
\noindent
\label{sec:int}
Gauge theories have been the subject of intense focus in fundamental physics for many decades now, yet certain avenues of exploration have remained stuck in nascent stages despite the advent of computers, supercomputers, and most recently exascale computing.
Quantum chromodynamics (QCD), describing the strong force, is a gauge theory of direct relevance to the physics of our universe -- and also the backdrop for a variety of interesting physics scenarios that are computationally intractable ~\cite{Gross:2022hyw}. The development of quantum technology is initiating a paradigm shift in studying gauge theories~\cite{Bauer:2022hpo, DiMeglio:2023nsa, Bauer:2023qgm, Banuls:2019rao, Meurice:2020pxc,Dalmonte:2016alw, Preskill:2018fag, Banuls:2019bmf} where novel quantum simulation and computation algorithms are being developed and are being implemented~\cite{Banerjee:2012pg, Banerjee:2012xg, Huffman:2021gsi, Klco:2018kyo, Klco:2019evd, Ciavarella:2021nmj, Stetina:2020abi, Ciavarella:2021lel, Farrell:2022wyt, Turro:2024pxu, Ciavarella:2024fzw, Farrell:2022vyh, Farrell:2024fit, Farrell:2024mgu, Haase:2020kaj, Dasgupta:2020itb, Bennewitz:2024ixi, Davoudi:2024wyv, Belyansky:2023rgh, Davoudi:2022xmb, Davoudi:2019bhy, Mueller:2022xbg, Davoudi:2022uzo, Nguyen:2021hyk, Atas:2022dqm, Atas:2021ext, Martinez:2016yna, Kasper:2016mzj,  Kane:2022ejm,  Mil:2019pbt, Muschik:2016tws, Meth:2023wzd, Zhang:2021bjq, Paulson:2020zjd, Guo:2024tnb, Crippa:2024cqr, Angelides:2023noe,Chakraborty:2020uhf, Paulson:2020zjd, Halimeh:2023lid, Angelides:2023noe, deJong:2021wsd, Lamm:2019bik, Shaw:2020udc, Kan:2021xfc, Cohen:2021imf, Murairi:2022zdg, Charles:2023zbl, Mildenberger:2022jqr, Halimeh:2022pkw, Zhang:2023hzr, Kane:2023jdo, Yang:2020yer, Zohar:2012ts, Zohar:2012xf, Hauke:2013jga, Wiese:2013uua, Marcos:2013aya, Zohar:2015hwa, Yang:2016hjn, Bender:2018rdp, Luo:2019vmi, Notarnicola:2019wzb, Surace:2019dtp, Surace:2020ycc, Kasper:2020akk, Aidelsburger:2021mia, Surace:2023qwo, Zache:2023cfj, Popov:2023xft, Gonzalez-Cuadra:2022hxt, Illa:2024kmf, Zache:2018jbt, Armon:2021uqr,Byrnes:2005qx, Gonzalez-Cuadra:2017lvz,Tagliacozzo:2012df,Tagliacozzo:2012vg,Zhou:2021kdl, Yamamoto:2022jnn, Zohar:2011cw, Zohar:2016iic, Alexandru:2021jpm,Lamm:2024jnl,Lamm:2019uyc, Riechert:2021ink, Semeghini:2021wls, Kokail:2018eiw, Gonzalez-Cuadra:2024xul, De:2024smi, PintoBarros:2018bzz, Fontana:2022dil, Su:2024uuc, Bender:2020ztu, Fontana:2020xzp, Schweizer_2019,Kalinowski:2022hze} to push the boundary of knowledge \cite{Feynman:1981tf, Lloyd:1996uni, Preskill:2021apy}. Understanding gauge theories by their quantum simulation or computation requires systematic development of suitable frameworks~\cite{Carena:2022kpg, Banuls:2017ena, Raychowdhury:2019iki, Kadam:2022ipf, Pardo:2022hrp, Chandrasekharan:1996ih, Bauer:2021gek, Zohar:2014qma,Zohar:2018cwb,Zohar:2019ygc,Alexandru:2023qzd,Alexandru:2019nsa,Ji:2022qvr,Ji:2020kjk,Kavaki:2024ijd,DAndrea:2023qnr,Romiti:2023hbd,Zache:2023dko,Wiese:2021djl, Kaplan:2018vnj, Mathur:2015wba, Mathur:2023lky, Fontana:2024rux, Grabowska:2024emw,Hartung:2022hoz} that are qubit-friendly and cost-effective~\cite{Muller:2023nnk,Ciavarella:2024fzw, Haase:2020kaj, Kane:2022ejm, Surace:2023qwo, Zache:2023cfj, Gonzalez-Cuadra:2022hxt,Illa:2024kmf, Zohar:2016iic, Lamm:2024jnl}. 
The theoretical frameworks are the starting point for algorithm developments for various tasks such as state preparation, time evolution, symmetry protection and other necessary operations, finally leading to calculation of the observables via performing quantum measurements~\cite{Davoudi:2024wyv, Farrell:2023fgd, Kane:2023jdo, Stryker:2021asy, Lamm:2019bik, Surace:2020ycc, Nguyen:2021hyk, Gustafson:2023swx,Lamm:2020jwv, Raychowdhury:2018osk, Mathew:2022nep,Mathew:2024bed, Muschik:2016tws, Atas:2021ext, Atas:2022dqm, Farrell:2024fit, Grabowska:2022uos}.

Quantum computation for a physical system is naturally expressed in a Hamiltonian framework, that requires a choice of basis. For gauge theories, specifically when the gauge group is continuous and non-Abelian, choosing a suitable basis stands as a nontrivial task. The Hamiltonian description of gauge theories is equipped with a local constraint structure which is to be preserved in the dynamics as well~\cite{Kogut:1974ag}. A preferred choice of basis that optimizes space resources to study quantum dynamics is a gauge invariant one \cite{Davoudi:2020yln}. However, the gauge invariance for a pure gauge theory is only obtained at the cost of nonlocality; in particular, it is possible to express the entire gauge-invariant Hilbert space of in terms of flux configurations generated by plaquette operators applied to a vacuum state, essentially using Wilson loops as the fundamental degrees of freedom. Mandelstam constraints \cite{Mandelstam:1962us, Mandelstam:1968hz, Mandelstam:1978ed} arise when equivalent loop configurations may be characterized in multiple ways, leading to an overcompleteness problem that has to be resolved. Mandelstam constraints were originally studied as nonlocal constraints but were later recast in a local form using the prepotential formalism for SU(2) \cite{Mathur:2007nu} as well as SU(3) \cite{Anishetty:2009nh}. 
The prepotential framework provides a way of converting the problem of solving Mandelstam constraints to the problem of finding an orthonormal and gauge invariant basis at each lattice site.
The recently developed loop-string-hadron (LSH) approach in SU(2) gauge theory builds upon this by providing a local and orthonormal gauge invariant basis, valid in any dimension~\cite{Raychowdhury:2019iki}.
Generalizing the framework to SU(3) is a crucial next step if it is to aid with quantum simulating and computing quantum chromodynamics in the foreseeable future, but doing so involves a series of nontrivial intermediate steps that are currently being pursued~\cite{Kadam:2022ipf}. The current work provides an intermediate step by prescribing a local and orthonormal basis applicable to pure SU(3) gauge theory in higher dimensions.

The foundation of the LSH framework for SU(2) gauge theory lies in the prepotential (or ``Schwinger boson'') formulation of the same, a reformulation of the original Kogut-Susskind formulation of Hamiltonian lattice gauge theories~\cite{Kogut:1974ag}.
The canonically conjugate variables in the Hamiltonian formulation are the electric fields, the generators of the Lie algebra, and `link operators' which are elements of the gauge group residing on each link of the lattice. While the eigenbasis of the Hamiltonian is essentially never known in closed form, one may work with either the well-understood electric basis or the magnetic (i.e. group-element) basis to study the dynamics of the theory. The former serves as the most conventional basis for lattice gauge theories, and is referred to as the strong coupling eigenbasis. The Hamiltonian framework of gauge theory also possesses an additional structure dictated by the Gauss law constraints. The Hamiltonian commutes with these constraints, implying the dynamics to remain gauge invariant.

The search for a suitable gauge-invariant basis for calculating Hamiltonian dynamics with optimal resources using a classical or quantum computer has been pursued for many years by gauge theorists with varying interest. Wilson loops, particle-pair flux strings, and hadrons belong to the gauge invariant Hilbert space of a gauge theory.
The prepotential formulation based on Schwinger bosons was developed to address the issues with nonlocality and proliferation of the loop basis~\cite{Mathur:2004kr,Mathur:2007nu,Mathur:2010wc,Mathur:2000sv,Raychowdhury:2013rwa}. Properly adapted, the SU(2) prepotential formulation can be used to construct a local, orthonormal, gauge invariant basis that solves nonlocality and the proliferation of loop degrees of freedom, and when combined with staggered matter yields the complete LSH framework for any spatial dimension. Generalization of the same to SU(3) involves multiple layers of complexity. As noted above, Ref.~\cite{Kadam:2022ipf} provided a complete SU(3) LSH framework for a one-dimensional spatial lattice with staggered fermions. The present work is the first building block to the SU(3) LSH construction of higher-dimensional theories.

For a $d$-dimensional Cartesian lattice, a site is connected with $2d$ number of links, each carrying gauge degrees of freedom or an irreducible representation (irrep). For SU(2) gauge theory, each link carries an angular momentum irrep that transforms like a rigid rotor under the two SU(2) generators at the two neighboring sites connected by that particular link. 
Forming an on-site gauge singlet requires solving all the Gauss law constraints (three components for SU(2)), which amounts to adding up $2d$ number of on-site angular momentum irreps to form an irrep with net-zero total angular momentum, following the rules of angular momentum addition. This task can be decomposed into adding two irreps at a time as the fundamental building block. This is geometrically realized as the ``point splitting'' scheme of LSH framework, where a site of a square (cubic) lattice is split into two (four) trivalent vertices, i.e., vertices attached to three edges. A trivalent vertex provides the next complicated scenario from a one-dimensional spatial lattice, where each vertex is only bivalent. Trivalent vertices naturally arise if a 2D continuum theory is discretized on a hexagonal or ``honeycomb'' lattice from the outset. In a recent study \cite{Kavaki:2024ijd}, trivalent vertices for 3D lattices in SU(2) gauge theory were studied, taking them as a starting point to ultimately arrive at triamond lattices as a novel alternative discretization for the continuum theory. In this work, we look into the simplest case of constructing gauge invariant states for SU(3) gauge theory of a trivalent vertex, which is the simplest scenario beyond 1+1 dimensions and a building block to full-fledged multidimensional lattices.
This article is the first of a multi-part series on the loop-string-hadron approach to SU(3) lattice gauge theory in more than 1+1 dimensions, with the present scope being limited to the vertex Hilbert space and associated subtleties.
Our analysis demonstrates how already at a trivalent vertex it is a nontrivial task to establish a basis that is suitable for computation (i.e., is complete and orthonormalized), and we offer a prescriptive solution to that problem that has no need for the infamous Clebsch-Gordon coefficients of SU(3).

The organization of this paper is as follows: The foundation of the LSH framework, SU(3) irreducible Schwinger bosons, is briefly reviewed in section \ref{sec:isb}. This ``irreducible'' construction is necessary because it solves the $\mathrm{Sp}(2,\mathbb{R})$ multiplicity problem that arises from naively constructing SU(3) representations out of ordinary Schwinger bosons.
In section \ref{sec:basis}, we define an LSH basis in analogy to SU(2) by contracting Schwinger bosons into gauge-singlet creation operators. We discuss some properties of the basis states, including why we refer to them as ``naive'', how some of them exhibit degeneracies with respect to the electric-flux irreps, and issues of orthogonality.\footnote{A previous work \cite{Anishetty:2019xge} aimed to obtain local and orthonormal basis for SU(3) by generalizing SU(2) along the same lines. However, effort of reproducing the results has revealed several inconsistencies both in orthogonality and calculated normalization factors, as elaborated in the body of this manuscript. }
In section \ref{sec:nondegen basis}, we further examine those states which are not afflicted by degeneracy and present their orthonormalization.
In section \ref{sec: degenerate}, we discuss the subspaces that exhibit degeneracy, their connection to Littlewood-Richardson coefficients, and the notion of a ``seventh Casimir'' that serves to resolve the degeneracy.
We present a choice of seventh Casimir that appears to implicitly define an orthogonal basis, but remains to be fully understood from an analytical point of view.
We discuss the significance of this analysis and state the road map for proceeding in the future in section \ref{sec: discussion}.

\section{SU(3) Prepotentials: Irreducible Schwinger Bosons}
\noindent
\label{sec:isb}

As mentioned before, the LSH framework is derived from the prepotential framework, which also goes by the name of Schwinger bosons. SU(3) prepotentials are different than the naive Schwinger boson representation of SU(3) algebra, which are adequate for SU(2). In this section we briefly introduce SU(3) irreducible Schwinger bosons, which serve as prepotentials for the theory and are used later in this work to construct a gauge-singlet local basis.

\subsection{Revisiting SU(2)}
Before proceeding to SU(3), it is instructive to have a quick look at the original SU(2) Schwinger boson construction for higher dimensional pure gauge theory \cite{Mathur:2004kr, Mathur:2007nu}, which directly precipitates the SU(2) LSH basis construction in any number of dimensions.
For a site $x$, the electric fields $E^a(x,I)$ are associated with the links connected to that particular site along directions $I=1,2,...,R$. Note that here $R$ denotes the connectivity of the lattice. For a $d$-dimensional hypercubic lattice, $R=2d$, whereas for a hexagonal lattice, $R=3$. \emph{In the rest of this work, we consider $R=3$, which corresponds to each site $x$ being a trivalent vertex.}\footnote{$R=3$ lattice sites can arise from the corresponding geometry of the lattice discretization or from a point-split hypercubic lattice~\cite{Raychowdhury:2018tfj, Raychowdhury:2019iki}.}

The electric fields must satisfy the SU(2) algebra:
\begin{eqnarray}
    \left[ E^a(x,I),E^{b}(x',I') \right]=i\epsilon^{abc}E^c(x,I)\delta_{I,I'}\delta_{x,x'},
\end{eqnarray}
where $\epsilon^{abc}$ is the Levi-Civita tensor, and $a,b,c$ can take values $1,2,3$. 
This requirement is satisfied by the construction
\begin{eqnarray}
    E^{a}(x,I)= a^\dagger_\alpha(x,I) {\left(\frac{\sigma^a}{2}\right)^{\alpha} }_\beta a^\beta(x,I),
\end{eqnarray}
where $\sigma^a$ are the Pauli matrices and $a^\dagger_\alpha(x,I)$'s are the independent Schwinger boson doublets $(\alpha,\beta=1,2)$ associated with each link $(x,I)$ that satisfy
\begin{eqnarray}
\label{eq: sb}
    \left[ a^\alpha(x,I), a^\dagger_\beta(x',I') \right]= \delta^\alpha_\beta \; \delta_{x,x'}\;\delta_{I,I'},~~~~ ~~~~ \left[ a^\dagger_\alpha(x,I), a^\dagger_\beta(x',I') \right]=0.
\end{eqnarray}
The association of Schwinger boson modes with legs of the trivalent vertex is illustrated in Fig.~\ref{subfig:SU2-prepotential}.
The strong coupling basis states or the irreps of SU(2) for $(x,I)$ are characterized by $|j,m\rangle_{(x,I)}$, where $j=0,$ $\frac{1}{2},$ $1,\cdots$ is the angular momentum quantum number and $m$ is the associated azimuthal quantum number.
The orthonormal $|j,m\rangle_{(x,I)}$ states are constructed in terms of the SU(2) Schwinger bosons as 
\begin{eqnarray}
\label{su2irrep}
    |j,m\rangle_{(x,I)}\equiv \frac{1}{\sqrt{(j+m)!(j-m)!}} \left(a^\dagger_1(x,I)\right)^{j+m}\left(a^\dagger_2(x,I)\right)^{j-m}|0,0\rangle. 
\end{eqnarray}
At a trivalent vertex $x$,  one can construct the following gauge-singlet creation operators out of SU(2) Schwinger bosons:
\begin{eqnarray}
\label{su2singlets}
    \epsilon^{\alpha\beta} a^\dagger_{\alpha}(x,1) a^\dagger_{\beta}(x,2)~,~~
    \epsilon^{\alpha\beta} a^\dagger_{\alpha}(x,2) a^\dagger_{\beta}(x,3) ~,~~\text{and}~~
    \epsilon^{\alpha\beta} a^\dagger_{\alpha}(x,3) a^\dagger_{\beta}(x,1).
\end{eqnarray}
Each one of these three composite creation operators can be visualized as creating a unit of directionless flux that runs through the vertex along two legs, as depicted in Fig.~\ref{subfig:SU2-singlets}.
The on-site loop Hilbert space is generated by applying arbitrary numbers of these excitations to the Schwinger boson vacuum, leading to the basis vectors
\begin{eqnarray}
\label{su2loopstate}
   |l_{12},l_{23},l_{31}\rangle_x  \propto (\epsilon^{\alpha\beta} a^\dagger_{\alpha}(x,1) a^\dagger_{\beta}(x,2))^{l_{12}}
    (\epsilon^{\alpha\beta} a^\dagger_{\alpha}(x,2) a^\dagger_{\beta}(x,3))^{l_{23}} (
    \epsilon^{\alpha\beta} a^\dagger_{\alpha}(x,3) a^\dagger_{\beta}(x,1))^{l_{31}}|0\rangle_x.
\end{eqnarray}
Here, $|0\rangle$ is the Schwinger boson vacuum -- the ket which is destroyed by (i.e., is in the kernel of) all six Schwinger boson annihilation operators associated with that site. Note that there exist other gauge-singlet operators involving one or more Schwinger boson annihilation operators, which annihilate $|0\rangle$, but have compound actions on a general loop state constructed above. The actions of all such operators can be obtained using the commutation relations among the Schwinger bosons given in~\eqref{eq: sb} (see Ref. \cite{Raychowdhury:2019iki}).

The net electric flux $j_I$ on each link $I=1,2,3$ can also describe the gauge-singlet state defined in (\ref{su2loopstate}) at a vertex $x$. The relation between $j_I$ and the loop quantum numbers, $l_{IJ}$, is given by
\begin{eqnarray}
  &&  l_{12}=j_1+j_2-j_3 ~~,~~ l_{23}=j_2+j_3-j_1 ~~, ~~ 
    l_{31}=j_3+j_1-j_2, \nonumber \\
    &\mbox{or equivalently,}&  2j_1=l_{12}+l_{31} ~~~~,~~~~
    2j_2=l_{12}+l_{23}~~~~,~~~~
    2j_3=l_{23}+l_{31} ~~.~~
\end{eqnarray}
Thus, the local gauge-invariant states are fully characterized by only the SU(2) Casimirs.
This illustrates a trivial identification of the states $|l_{12},l_{23},l_{31}\rangle_x$ with the gauge-invariant orthonormal basis $|j_1,j_2,j_3\rangle_{x}$ constructed by summing over irrep states $|j,m\rangle_{(x,1)} \otimes |j,m\rangle_{(x,2)} \otimes |j,m\rangle_{(x,3)}$ weighted appropriately with Clebsch-Gordon coefficients.
Note that, in the latter case, the gauge-singlet orthonormal basis is given by the states which have zero total angular momentum, and hence, has no independent azimuthal quantum number.
This condition is a signature of no Mandelstam constraint being present as loop states are orthonormal for a trivalent vertex. 

It is insightful to return to a square lattice and reinvestigate the situation. The bijection between the electric flux values {$\{j_I\}$} and loop quantum numbers $\{l_{IJ}\}$ is not there beyond a trivalent vertex. This is a signature of Mandelstam constraint being present, which is well studied in the literature, see Ref.~\cite{Mathur:2007nu}. As an example, for a site on a square lattice with $R=4$, one has four electric flux values $j_1, ~j_2, ~j_3,$ and $ ~ j_4$ and six loop quantum numbers $l_{12},~l_{13},~l_{14},~l_{23},~l_{24},$ and $~l_{34}$. They are related as:
\begin{eqnarray}
    j_1 = l_{12}+l_{13}+l_{14}~~,~~
    j_2 = l_{12}+l_{23}+l_{24} ~~,~~
    j_3 = l_{13}+l_{23}+l_{34} ~~,~~ j_4 = l_{14} + l_{24} + l_{34}.\nonumber
\end{eqnarray}
The left-hand sides of the above equations remain same when the variables $l_{12},~l_{13},~l_{14},~l_{23},~l_{24}, ~l_{34}$ on the right-hand sides are replaced with  \begin{eqnarray}
    l_{12},~l_{13}+p,~ l_{14}-p,~  l_{23}-p, ~ l_{24} +p,~ l_{34},
\end{eqnarray}
for any integer $p$. In other words, multiple possible combinations of the loop quantum numbers at a vertex correspond to the same value for Casimirs on the attached links. 
This is related to the existence of a Mandelstam constraint depicting overcompleteness of the loop basis in this case.
The loop degrees of freedom combined with a suitable point splitting scheme~\cite{Raychowdhury:2018tfj, Raychowdhury:2019iki} for a particular dimension yields exactly the orthonormal and physical states of an SU(2) lattice gauge theory at each point-split trivalent vertex.

\begin{figure}[ht]
    \centering
    \begin{subfigure}[t]{0.47\textwidth}
        \centering
        \includegraphics{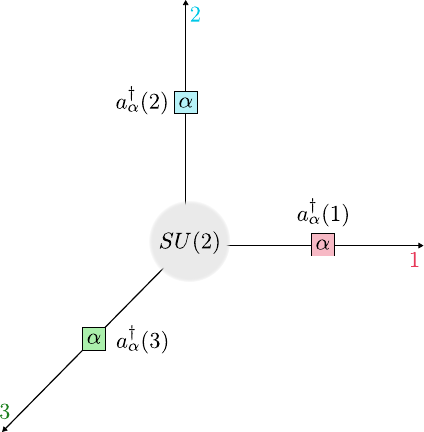}
        \caption{
        }
        \label{subfig:SU2-prepotential}
    \end{subfigure}
    \hfill
    \begin{subfigure}[t]{0.47\textwidth}
        \centering
        \includegraphics{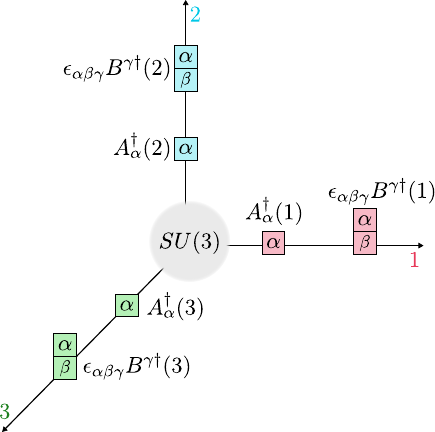}
        \caption{}
        \label{subfig:SU3-prepotential}
    \end{subfigure}
    \caption{Prepotential modes associated with a trivalent vertex in Yang-Mills theory, from which one may construct gauge-invariant
singlet operators by suitable index contractions. (a) SU(2): The prepotentials consist of one doublet of harmonic oscillators (Schwinger
bosons) assigned to each leg $I$. Each doublet transforms in the fundamental irrep, spin $\frac{1}{2}$. (b) SU(3): Each leg $I$ has one triplet of
oscillators in the $(1,0)$ irrep, and one antitriplet in $(0,1)$. The $B^\dagger$ modes are illustrated with the double-box Young tableaux, but their
antisymmetrization with $\epsilon$ renders them as single-index objects, which are used for forming gauge-singlet operators throughout the text.
For example, $\epsilon^{\alpha\beta\gamma}A(I)^\dagger_\alpha\epsilon_{\beta\gamma\sigma}B(J^{\dagger\sigma})\propto A(I)^\dagger_\alpha B(J^{\dagger\alpha})$, with the right-hand side appearing in (\ref{eq: su3singlets}). }
    \label{fig:prepotential}
\end{figure}

\subsection{The SU(3) generalization}
The main focus of this work is constructing SU(3) singlet states at a trivalent vertex. The simple yet successful construction for SU(2) mentioned above may be generalized in a straightforward way for SU(3) gauge theories.
However, such a construction comes with subtleties that will be discussed in detail in Secs.~\ref{sec:basis} and~\ref{sec:nondegen basis}.
In the remaining part of this section, we briefly review the generalization of the Schwinger boson construction for SU(3): the SU(3) prepotentials also known as SU(3) irreducible Schwinger bosons. 

The prepotential formulation of the SU(3) gauge group, which is a rank-two group with a three-dimensional fundamental representation, requires two independent Schwinger boson triplets at both ends of each link. Considering a site $x$, and links along directions $I=1,2,3$, the Schwinger bosons are denoted by $a^\dagger_\alpha(x,I)$ and $b^{\dagger\alpha}(x,I)$. Here, $\alpha$ and $\beta$ are the color indices that take integer values from 1 to 3.  The SU(3) Schwinger bosons obey ordinary bosonic commutation relations:
\begin{align}
    \big[a^\alpha(x,I),a^{\alpha'}(x',I')\big]&=\big[b_\alpha(x,I),b_{\alpha'}(x',I')\big]=0,
    \label{eq: SU3 schwinger boson aa bb commutation}\\
    \big[a^\alpha(x,I),b_{\alpha'}(x',I')\big]&=\big[a^\alpha(x,I),b^{\dagger\alpha'}(x',I')\big]=0,
    \label{eq: SU3 schwinger boson ab commutation}\\
    \big[a^\alpha(x,I),a^\dagger_{\alpha'}(x',I')\big]&=\big[b_{\alpha'}(x,I),b^{\dagger\alpha}(x',I')\big]=\delta_{\alpha '}^{\alpha}\delta_{xx'}\delta_{II'}
    \label{eq: SU3 schwinger boson non-zero commutation},
\end{align}
where $I$ and $I'$ can each take values $1,2,3$ for a trivalent vertex.

Following the Schwinger boson construction for SU(3) generators, the chromoelectric fields are defined in terms of the prepotentials as
\begin{align}
\label{eq: E in terms of a and b}
    E^{\rm a}(x,I)&= a^\dagger_\alpha(x,I)\,\big(T^{\rm a}\big)^\alpha{}_\beta \, a^\beta(x,I) - b^{\dagger\alpha}(x,I)\,\big(T^{*\rm a}\big)_\alpha{}^\beta\, b_\beta(x,I) ,
\end{align}
where $T^{\rm a}$ are one-half times the SU(3) Gell-Mann matrices. As shown in Refs.~\cite{Anishetty:2009ai,Anishetty:2009nh}, this set of ordinary Schwinger bosons is not suitable for the construction of the gauge theory Hilbert space in the above form. Following~\eqref{su2irrep}, a similar monomial of the Schwinger bosons $a^\dagger_{\alpha}$ and $b^{\dagger\beta}$ will likewise create representations of SU(3), however, they are generally not irreducible~\cite{Mukunda1965TensorMA,Chaturvedi:2002si}. The simplest case of the $(1,1)$ irrep of SU(3) is naively constructed as $a^\dagger_{\alpha} b^{\dagger\beta}|0,0\rangle$\footnote{For the remainder of this section, we will be discussing the Schwinger bosons at the same site and link-end. Thus, position and link-end labels on Schwinger bosons are suppressed for brevity.}, however, its proper construction is instead
\begin{equation}
   |1,1\rangle_\alpha^\beta= a^\dagger_{\alpha} b^{\dagger\beta}|0,0\rangle - \frac{\delta_{\alpha}^{\beta}}{3} (a^\dagger \cdot b^\dagger)|0,0\rangle,  \label{eq: SU3-11-irrep}
\end{equation}
such that the irrep is traceless, which is a fundamental property of any irrep. Above, $\alpha\,,\beta=1, 2, 3$ are the color indices,  $a^\dagger\cdot b^\dagger \equiv \sum_{\gamma=1}^3 a^\dagger_\gamma b^{\dagger\gamma}$, and the $(0,0)$ irrep state, $|0,0\rangle$, satisfies $a^\alpha|0,0\rangle = b_{\alpha}|0,0\rangle=0$ for all values of $\alpha$. For a general irrep $(p,q)$, one has to extract out all the traces from the monomial state $a^\dagger_{\alpha_1}\cdots a^\dagger_{\alpha_p} b^{\dagger\beta_1}\cdots b^{\dagger\beta_q}|0,0\rangle $ to satisfy the tracelessness condition (see, for example, equation (35) of Ref.~\cite{Mathur:2000sv}). Such a traceless construction is cumbersome to use, and on top of that, there is also an $\mathrm{Sp}(2,\mathbb{R})$ multiplicity problem because $(a^\dagger\cdot b^{\dagger})^{\rho}\displaystyle|P,Q\rangle_{\vec{\alpha}}^{\vec{\beta}}$ for any positive integer $\rho$ transforms in the same way under SU(3) as the $(p,q)$ irrep \cite{Chaturvedi:2002si}.

A solution to the $\mathrm{Sp}(2,\mathbb{R})$ multiplicity problem is obtained in terms of the irreducible Schwinger boson (ISB) construction ~\cite{Anishetty:2009ai}, 
where the state in~\eqref{eq: SU3-11-irrep} can be constructed as a monomial of ISBs:
\begin{equation}
   |1,1\rangle_\alpha^\beta= A^\dagger_{\alpha} B^{\dagger\beta}|0,0\rangle. \label{eq: SU3-11-ISB}
\end{equation}
Above, $A^\dagger_{\alpha}$ and $ B^{\dagger\beta}$ are the ISB creation operators, constructed in terms of the naive SU(3) Schwinger bosons as
\begin{eqnarray}
    A^{\dagger}_{\alpha}&\equiv& a^{\dagger}_{\alpha}-\frac{1}{P+Q+1}(a^\dagger\cdot b^\dagger)b_{\alpha},
    \label{eq: Adagg-def} \\
    B^{\dagger\alpha}&\equiv& b^{\dagger\alpha}-\frac{1}{P+Q+1}(a^\dagger\cdot b^\dagger)a^{\alpha}.
    \label{eq: Bdagg-def}
\end{eqnarray}
In \eqref{eq: Adagg-def} and \eqref{eq: Bdagg-def}, we have introduced the number operators
\begin{equation}
    P \equiv a^{\dagger}\cdot a = \sum_{\alpha=1}^{3} a^{\dagger}_\alpha\,a^{\alpha} \quad\text{and}\quad Q \equiv b^{\dagger}\cdot b = \sum_{\beta=1}^{3} b^{\dagger\beta}\,b_{\beta}.
    \label{eq: prepotential number operators}
\end{equation}
The association of ISB modes with legs of the trivalent vertex is illustrated in Fig.~\ref{subfig:SU3-prepotential}. Construction of irreps as monomials of ISBs, as done in~\eqref{eq: SU3-11-ISB}, is valid for any arbitrary irrep, and it is given by
\begin{equation}
    |p,q\rangle_{\vec{\alpha}}^{\vec{\beta}}= \mathcal{N} A^\dagger_{\alpha_1}\ldots A^\dagger_{\alpha_p} B^{\dagger\beta_1}\ldots B^{\dagger\beta_q}|0\rangle.
    \label{eq: SU3-irrep}
\end{equation}
 Here, $\mathcal N$ corresponds to the normalization factor, the indices $\alpha_{i}$ with $i=1, \cdots, p$ and  $\beta_{j}$ with $j=1, \cdots, q$ can take integer values between 1 to 3, and $\alpha_{i}$ and $\beta_{j}$ determine the isospin and hypercharge of the irrep state~\cite{Mukunda1965TensorMA,Chaturvedi:2002si}. The monomial states constructed in~\eqref{eq: SU3-irrep} are free from any $\mathrm{Sp}(2,\mathbb{R})$ multiplicity problem as the states satisfy 
\begin{equation}
    A^\dagger\cdot B^\dagger\, |p,q\rangle_{\vec{\alpha}}^{\vec{\beta}} =0  \quad\text{and}\quad A\cdot B\, |p,q\rangle_{\vec{\alpha}}^{\vec{\beta}} =0,
    \label{eq: ISB multiplicity constraint}
\end{equation}
which ensures that they are restricted to the $\rho=0$ subspace or equivalently the kernel of the operator $a\cdot b$; see Ref.~\cite{Anishetty:2009nh} for details. This implies that the operators $A^\dagger\cdot B^\dagger$ and $A\cdot B$ are effectively the null operators within the Hilbert space spanned by states in~\eqref{eq: SU3-irrep}. This leads to the following modified commutation relations for ISBs:
\begin{eqnarray}
    &&[A^\alpha, A^\dagger_\beta] \simeq \left( \delta^\alpha_{\beta}-\frac{1}{P+Q+2}B^{\dagger\alpha}B_\beta \right),
    \label{eq: AAdagg-commutator}\\
    &&[B_\alpha, B^{\dagger\beta}] \simeq  \left( \delta_\alpha^{\beta}-\frac{1}{P+Q+2}A^{\dagger}_{\alpha}A^\beta \right),
    \label{eq: BBdagg-commutator}\\
    && [A^\alpha,B^{\dagger\beta}] \simeq   -\frac{1}{P+Q+2}B^{\dagger\alpha}A^\beta, 
    \label{eq: ABdagg-commutator}\\
    && [B_\alpha,A^\dagger_\beta] \simeq   -\frac{1}{P+Q+2}A^{\dagger}_\alpha B_\beta,
    \label{eq: AdaggB-commutator}
\end{eqnarray}
along with
\begin{eqnarray}
    [ A^\dagger_\alpha, A^\dagger_{\beta}]=[ A^\alpha, A^{\beta}]=[B_\alpha,B_\beta]=[ B^{\dagger\alpha}, B^{\dagger\beta}]=[A^\alpha,B_\beta]=[A^\dagger_\alpha, B^{\dagger\beta} ]=0,
    \label{eq: zero-commutators}
\end{eqnarray}
where $\simeq$ indicates that the above set of commutation relations is valid within the vector subspace spanned by SU(3) irreps defined in~\eqref{eq: SU3-irrep}. Note that, in the same subspace, the number operators for $A$-type and $B-$type ISBs satisfy
\begin{equation}
    A^\dagger\cdot A \simeq P \quad \text{and} \quad B^\dagger\cdot B \simeq Q\,.
    \label{eq: NA Na and NB Nb equivalence}
\end{equation}
The eigenvalues of $A^\dagger\cdot A$ and $B^\dagger\cdot B$ acting on an SU(3) irrep are given by $p$ and $q$, respectively, which serve as the quantum numbers to characterize the irrep as well.

To summarize, the SU(2) and SU(3) prepotentials provide an alternate construction for the group irreps by introducing harmonic oscillator operators and the associated vacuum states.
The prepotentials, depicted in Figs.~\ref{subfig:SU2-prepotential} and~\ref{subfig:SU3-prepotential} for the SU(2) and SU(3) group, respectively, can be used to construct the gauge-singlet excitations at a trivalent vertex.
We showed this for the SU(2) group in this section (see equation \eqref{su2singlets}) and provided a pictorial representation of the degrees of freedom associated with their excitations in Fig.~\ref{subfig:SU2-singlets}.
In the next section, we provide a similar construction for the SU(3) gauge-singlets at a trivalent vertex using the corresponding prepotentials, and discuss the complications associated with it.

\section{SU(3)-invariant basis at the vertex
\label{sec:basis}}
\noindent
\begin{figure}[ht]
    \centering
    \begin{subfigure}[t]{0.47\textwidth}
        \centering
        \includegraphics{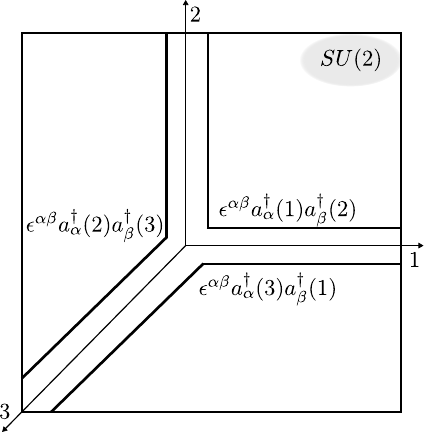}
        \caption{}
        \label{subfig:SU2-singlets}
    \end{subfigure}
    \hfill
    \begin{subfigure}[t]{0.47\textwidth}
        \centering
        \includegraphics{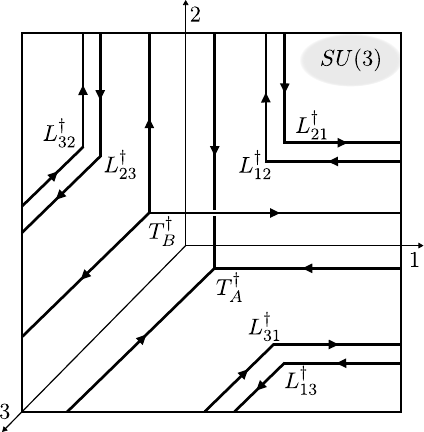}
       \caption{}
        \label{subfig:SU3-singlets}
    \end{subfigure}
    \caption{ On-site singlets constructed out of prepotential operators for the SU(2) and SU(3) trivalent vertices. Young-tableaux
representation of the same is discussed in Fig. \ref{fig:prepotential}. (a) SU(2), being a rank $1$ group, contains only a single fundamental irrep. This is
translated to the fact that SU(2) loops/ Wilson lines are directionless. At a trivalent vertex, flux can flow along only three directions as
shown, consistent with on-site SU(2) singlets constructed using prepotentials as in (\ref{su2singlets}). (b) SU(3), being a rank $2$ group, contains two
fundamental irreps  and $(0, 1)$. With the convention that SU(3) loops/ Wilson lines are always directed from a $(0, 1)$ to a $(1, 0)$, at a
trivalent vertex, flux can flow along six possible directions, as well can emerge/absorb from all three directions, consistent with on-site
SU(3) singlets constructed using prepotentials as in \eqref{eq: su3singlets}.
 }
    \label{fig:singlet}
\end{figure}
Equipped with the irreducible Schwinger bosons, the construction of on-site SU(2) gauge singlet pure-creation operators at a trivalent vertex $x$ as given in (\ref{su2singlets}) can be generalized to SU(3) as below:
\begin{eqnarray}
\label{eq: su3singlets}
    L^\dagger_{12}= A^\dagger_\alpha(1) B^{\dagger\alpha}(2), ~~&& ~~ L^\dagger_{21}= A^\dagger_\alpha(2) B^{\dagger\alpha}(1), \nonumber \\
    L^\dagger_{23}= A^\dagger_\alpha(2) B^{\dagger\alpha}(3), ~~&& ~~ L^\dagger_{32}= A^\dagger_\alpha(3) B^{\dagger\alpha}(2), \nonumber \\
    L^\dagger_{31}= A^\dagger_\alpha(3) B^{\dagger\alpha}(1), ~~&& ~~ L^\dagger_{13}= A^\dagger_\alpha(1) B^{\dagger\alpha}(3), \nonumber \\
    T^\dagger_{A}= \epsilon^{\alpha\beta\gamma}A^\dagger_\alpha(1) A^\dagger_\beta(2)A^\dagger_\gamma(3), ~~&& ~~  T^\dagger_{B}= \epsilon_{\alpha\beta\gamma}B^{\dagger\alpha}(1) B^{\dagger\beta}(2)B^{\dagger\gamma}(3).
\end{eqnarray}
It is often useful to think of the bilinear excitations $L_{IJ}^\dagger$ as creating an oriented unit of gauge flux, entering the vertex from $I$ and exiting along $J$.
For SU(3), it is possible that three units of flux may terminate on or emanate from a point, represented by the trilinear $T_A^\dagger$ and $T_B^\dagger$ excitations, respectively.
All of these elementary excitations are illustrated in Fig.~\ref{subfig:SU3-singlets}.

Arbitrary combinations of the above creation operators can be applied to the strong-coupling vacuum to create the gauge invariant Hilbert space of the vertex.
However, arbitrary products of creation operators cannot be thought of as all being independent.
Indeed, when applied to the Hilbert space created by ISBs, the above set of gauge singlet operators satisfy
\begin{eqnarray}
    \label{eq:TT=LLL+LLL}
    T^\dagger_{A}T^\dagger_{B} \simeq L^\dagger_{12}L^\dagger_{23}L^\dagger_{31}+L^\dagger_{21}L^\dagger_{32}L^\dagger_{13} .
\end{eqnarray}
The above identity can be checked with help from the algebraic identity  
\begin{eqnarray}
    \epsilon_{\alpha_1\beta_1\gamma_1} \epsilon^{\alpha_2\beta_2\gamma_2}=  \left| \begin{array}{ccc}
        \delta_{\alpha_1}^{\alpha_2} &  \delta_{\alpha_1}^{\beta_2} &  \delta_{\alpha_1}^{\gamma_2}\\
         \delta_{\beta_1}^{\alpha_2} &  \delta_{\beta_1}^{\beta_2} &  \delta_{\beta_1}^{\gamma_2}\\
         \delta_{\gamma_1}^{\alpha_2} &  \delta_{\gamma_1}^{\beta_2} &  \delta_{\gamma_1}^{\gamma_2}\\
    \end{array}\right|,
\end{eqnarray}
and the properties of irreducible Schwinger bosons as given in 
\eqref{eq: ISB multiplicity constraint}-\eqref{eq: zero-commutators}.
Equation \eqref{eq:TT=LLL+LLL} indicates that whenever both $T$-type excitations are present, a pair can be deleted and replaced by a linear combination of triple $L$-type excitations instead; we return to this point below.

\begin{figure}
    \centering
    \includegraphics{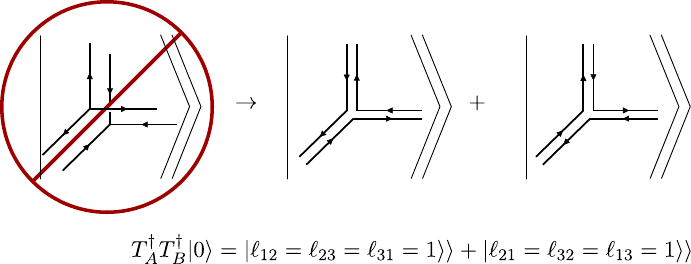}
    \caption{Pictorial representation of the identity (\ref{eq:TT=LLL+LLL}), which implies that the simultaneous presence of $T^\dagger_A$ and $T^\dagger_B$ excitations present at a site is equivalent is equivalent to having a linear combination of flux lines flowing around the site in a cyclic and anticyclic manner. This linear dependence explains why definition \eqref{eq:naive basis kets} was designed to exclude simultaneous $T_A^\dagger$ and $T_B^\dagger$ excitations.
    \label{fig:no-TA-TB}}
\end{figure}
In addition to the purely creation-type operators listed in (\ref{eq: su3singlets}) and their purely annihilation adjoints, there exist numerous other mixed-type gauge singlet operators.
A few of these arise in normalization evaluations of Appendix \ref{app: norm calc}, while the others will be covered later in this series.
Next, we discuss how to construct basis states out of these gauge singlet operators acting on the total bosonic vacuum. 

We define the naive LSH states as \footnote{In this work, any operator raised to the zero power is understood to be shorthand for the identity operator: $(L_{IJ}^\dagger)^0 =  (T_A^\dagger)^0 = (T_B^\dagger)^0 \equiv \mathds{1} $.}$^,$~\footnote{When working with LSH basis kets, it can be useful to group quantum numbers together into the sets $\{ \ell_{12}, \ell_{23}, \ell_{31} \}$, $\{ \ell_{21}, \ell_{32}, \ell_{13} \}$, and $t$ (alone). Therefore, in the definition of \eqref{eq:naive basis kets} and whenever we wish to emphasize this grouping, we use semicolons (;) to separate the groups. Whenever a basis state is written with semicolons or entire quantum numbers omitted, it should be clear from context how to relate the state back to the definition \eqref{eq:naive basis kets}.}
\begin{align}
    \label{eq:naive basis kets}
    \kett{\ell_{12}\, \ell_{23}\, \ell_{31} ; \ell_{21}\, \ell_{32}\, \ell_{13}; t } &\equiv L_{12}^{\dagger \, \ell_{12}} L_{23}^{\dagger \, \ell_{23}} L_{31}^{\dagger \, \ell_{31}} L_{21}^{\dagger \, \ell_{21}} L_{32}^{\dagger \, \ell_{32}} L_{13}^{\dagger \, \ell_{13}} \times \begin{cases}
        T_A^{\dagger \,  t} \ket{0}, & t \geq 0 \\ 
        T_B^{\dagger \, -t} \ket{0}, & t < 0
    \end{cases} \\
    \ell_{IJ} &\in \{ 0, \ 1, \ 2, \ 3, \, \ldots \} , \\
    t &\in \{ 0, \ \pm 1, \ \pm 2 , \ \pm 3, \ \ldots \}.
\end{align}
Above, the notation $\kett{\cdots}$ is used to indicate a basis ket that is not necessarily normalized to unity.
In this notation, the total vacuum ket $\ket{0}$ may also be written as
\begin{align}
    \kett{0} &\equiv \ket{0} = \kett{0 \, 0 \, 0 ; 0 \, 0 \, 0 ; 0 }.
\end{align}
The vacuum is normalized to unity by definition: $\braket{0|0}=1$.
Supplementing the definition of \eqref{eq:naive basis kets}, we also have a pictorial representation of the basis states: Inside a ket, rather than listing the seven quantum numbers, we may instead draw appropriate quantities of all the elementary bilinear and trilinear flux excitations.
(We refer the reader to the beginning of this section and to Fig.~\ref{subfig:SU3-singlets} for relevant conventions.)
As an example of these ``cartoon'' states, we have illustrated a consequence of the Mandelstam identity \eqref{eq:TT=LLL+LLL} in Fig.~\ref{fig:no-TA-TB} and its relation to the basis definition.

It is no surprise that the kets of \eqref{eq:naive basis kets} are generally not normalized, but what is less obvious is the fact that they are not all orthogonal -- a fact that was overlooked in Ref.~\cite{Anishetty:2019xge}.
Thus, the LSH ``quantum numbers'' are not always ``good'' quantum numbers.
It is for this reason that we refer to the states in \eqref{eq:naive basis kets} as the ``naive'' LSH states.
The simplest example where nonorthogonality can be found is in the sector of $p_1=q_1=p_2=q_2=p_3=q_3=1$, as depicted in Fig.~\ref{fig:NZ-overlap}.
The issue of nonorthogonality is of practical importance for at least two reasons:
i) Orthogonality is a great aid to normalization evaluations, and wrongly assuming orthogonality will lead to incorrect normalizations \cite{Anishetty:2019xge}. Such normalization factors arise when evaluating matrix elements of operators.
ii) One would like an orthogonal basis for quantum computation, such that the basis states are distinguishable and can be mapped onto distinct measurements of qubits or qudits.

Despite issues of nonorthogonality, we have found that the above set nonetheless constitutes a valid basis for the Hilbert space of gauge-singlet states of the vertex.
That is, \eqref{eq:naive basis kets} defines a complete basis for the Hilbert space that is linearly independent.
This empirical observation is supported by our computer-assisted investigations and is discussed more below.
Note that the naive basis does not admit any states containing both $T_A^\dagger$ and $T_B^\dagger$ excitations explicitly; following earlier remarks, the identity \eqref{eq:TT=LLL+LLL} implies that, whenever both excitations are present, the state can be decomposed into a linear combination of states with suitably incremented $\ell_{IJ}$ quantum numbers.
The omission of LSH states with simultaneous $T_A^\dagger$ and $T_B^\dagger$ creation operators therefore serves to fulfill the requirement of linear independence, and below we confirm that it leads to the correct number of states.

Nonorthogonality can only arise between states sharing the same $(p_1,q_1,p_2,q_2,p_3,q_3)$ quantum numbers (as is the case in Fig.~\ref{fig:NZ-overlap}).
We shall frequently refer to this collection of six quantum numbers, so we here introduce the shorthand notation
\begin{align}
    \vecpq &\equiv (p_1,q_1,p_2,q_2,p_3,q_3) .
\end{align}
The $\vecpq$ quantum numbers for a given naive basis ket are obtained from the naive LSH quantum numbers as
\begin{subequations}
\label{eq:PQ from llllllt}
\begin{align}
    p_1 &= \ell_{12} + \ell_{13} + |t| \, \theta( t) , \\
    q_1 &= \ell_{31} + \ell_{21} + |t| \, \theta(-t) , \\
    p_2 &= \ell_{23} + \ell_{21} + |t| \, \theta( t) , \\
    q_2 &= \ell_{12} + \ell_{32} + |t| \, \theta(-t) , \\
    p_3 &= \ell_{31} + \ell_{32} + |t| \, \theta( t) , \\
    q_3 &= \ell_{23} + \ell_{13} + |t| \, \theta(-t) ,
\end{align}
\end{subequations}
where $\theta$ denotes the Heaviside function.
Inspection of \eqref{eq:PQ from llllllt} reveals a symmetry of $\vecpq$ under the following transformation:
\begin{align}
\label{eq:PQ shift symmetry}
    \bigl( \ell_{12}, \, \ell_{23}, \, \ell_{31} ; \ell_{21}, \, \ell_{32}, \, \ell_{13}; t \bigr) \to& \bigl( \ell_{12}, \, \ell_{23}, \, \ell_{31} ; \ell_{21}, \, \ell_{32}, \, \ell_{13}; t \bigr) \JS{+n} (1,1,1;-1,-1,-1;0) \nonumber \\
    &\JS{=} \bigl( \ell_{12} + n , \, \ell_{23} + n , \, \ell_{31} + n ; \ell_{21} - n , \, \ell_{32} - n , \, \ell_{13} - n ; t \bigr) , 
\end{align}
where $n$ is any integer satisfying
\begin{align}
    - \min ( \ell_{12}, \, \ell_{23}, \, \ell_{31}) \leq & n  \leq  \min ( \ell_{21}, \, \ell_{32}, \, \ell_{13} ),
\end{align}
ensuring that all loop numbers remain nonnegative. 
This indicates that the quantum numbers $\vecpq$ are insufficient to completely characterize the states. 
Thus, nonorthogonality may arise if two naive basis kets are related by \eqref{eq:PQ shift symmetry}.
Moreover, it will be argued below that all naive basis kets belonging to a common $\vecpq$ sector must be related by \eqref{eq:PQ shift symmetry}.

\begin{figure}[t]
    \centering
    \includegraphics{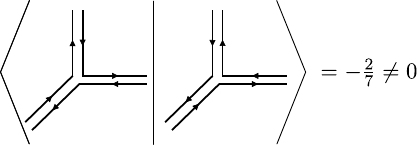}
    \caption{The bra, $\bra{0 \, 0 \, 0 ; 1 \, 1 \, 1 ; 0} $, and the ket, $\ket{1 \, 1 \, 1 ; 0 \, 0 \, 0 ; 0}$ are two normalized naive LSH states characterized by $p_1=q_1=p_2=q_2=p_3=q_3=1$. Although they have distinct LSH labels, explicit calculation confirms that these states are neither parallel nor orthogonal (the angle between them is $\cos^{-1}(-2/7)$) and presents the simplest counterexample to total orthogonality of the $\{\ell_{ij}\}$ and $t$ quantum numbers that was claimed in Ref.~\cite{Anishetty:2019xge}.}
    \label{fig:NZ-overlap}
\end{figure}

Some choices of $\vecpq$ are, however, free from problems with orthogonality, because the subspace they characterize is one-dimensional.
This is the situation for an LSH state that satisifies $\min ( \ell_{12}, \, \ell_{23}, \, \ell_{31}) = \min ( \ell_{21}, \, \ell_{32}, \, \ell_{13} ) =  0$.
We shall refer to these one-dimensional subspaces as being \emph{nondegenerate}, in reference to the $\vecpq$ quantum numbers. 
Multidimensional subspaces will be referred to as being \emph{degenerate}, again in reference to the $\vecpq$ quantum numbers.
In the degenerate sectors, another quantum number is necessary to distinguish orthogonal states.

The need for seven quantum numbers to characterize any general state of the vertex can be understood by the following argument:
The characterization of the state of an individual link end requires five pieces of data, conventionally chosen as $p$, $q$, the total isospin, the third component of isospin, and the hypercharge.
The 15 degrees of freedom belonging to three link ends are then tied together and constrained by eight components of Gauss's law, leading to $15-8=7$ remnant gauge-invariant degrees of freedom.
It is well-known that $p$ and $q$ can be inverted to give eigenvalues of the quadratic and cubic Casimir operators; hence, the specification of $\vecpq$ is equivalent to knowing the eigenvalues of six distinct Casimir operators.
We shall use the term ``seventh Casimir'' to refer to some other independent operator that can be simultaneously diagonalized with the first six, with its eigenvalues being sufficient to distinguish orthogonal states.
Formally, the seventh Casimir should be a Hermitian operator that commutes with the other six Casimir operators and has a nondegenerate spectrum when projected into any given $\vecpq$ sector.

\section{Nondegenerate basis states
\label{sec:nondegen basis}}

Whenever $\dim (\vecpq)=1$, the only LSH state belonging to the one-dimensional subspace is automatically orthogonal to all other naive basis states, and it can be normalized in closed form.
The nondegenerate basis states may be categorized into three classes.
To enumerate them, note that the criterion $\min ( \ell_{12}, \, \ell_{23}, \, \ell_{31}) = \min ( \ell_{21}, \, \ell_{32}, \, \ell_{13} ) =  0$ implies that one of $ \{ \ell_{12}, \, \ell_{23}, \, \ell_{31} \} $ and one of $\{ \ell_{21}, \, \ell_{32}, \, \ell_{13} \}$ must be zero.
The vertex labels $(I,J,K)$ can always be associated with legs 1, 2, and 3 such that $\ell_{IK}$ denotes one of the vanishing $\ell$-type quantum numbers.
We then obtain three classes according to the choice of which member of $\{ \ell_{JI},\ell_{KJ},\ell_{IK} \}$ is also zero.
\begin{itemize}
    \item \emph{Class I}. States of the form $\kett{\ell_{IJ},\ell_{JK},\ell_{JI},\ell_{KJ};t}$, i.e.,
\begin{align}
    \kett{\ell_{IJ},\ell_{JK},\ell_{JI},\ell_{KJ};t} &= (L_{IJ}^{\dagger})^{\ell_{IJ}} (L_{JK}^{\dagger})^{\ell_{JK}} (L_{JI}^{\dagger})^{\ell_{JI}} (L_{KJ}^{\dagger})^{\ell_{KJ}} \times \begin{cases}
        T_A^{\dagger \,  t} \ket{0}, & t \geq 0 \\ 
        T_B^{\dagger \, -t} \ket{0}, & t < 0.
    \end{cases}
\end{align}
    \item \emph{Class IIa}. States of the form $\kett{\ell_{IJ},\ell_{JK},\ell_{JI},\ell_{IK};t}$, i.e.,
\begin{align}
    \kett{\ell_{IJ},\ell_{JK},\ell_{JI},\ell_{IK};t} &= (L_{IJ}^{\dagger})^{\ell_{IJ}} (L_{JK}^{\dagger})^{\ell_{JK}} (L_{JI}^{\dagger})^{\ell_{JI}} (L_{IK}^{\dagger})^{\ell_{IK}} \times \begin{cases}
        T_A^{\dagger \,  t} \ket{0}, & t \geq 0 \\ 
        T_B^{\dagger \, -t} \ket{0}, & t < 0 .
    \end{cases}
\end{align}
    \item \emph{Class IIb}. States of the form $\kett{\ell_{IJ},\ell_{JK},\ell_{KJ},\ell_{IK};t}$, i.e.,
\begin{align}
    \kett{\ell_{IJ},\ell_{JK},\ell_{KJ},\ell_{IK};t} &= (L_{IJ}^{\dagger})^{\ell_{IJ}} (L_{JK}^{\dagger})^{\ell_{JK}} (L_{KJ}^{\dagger})^{\ell_{KJ}} (L_{IK}^{\dagger})^{\ell_{IK}} \times \begin{cases}
        T_A^{\dagger \,  t} \ket{0}, & t \geq 0 \\ 
        T_B^{\dagger \, -t} \ket{0}, & t < 0 .
    \end{cases}
\end{align}
\end{itemize}
Above, we have written the naive basis by suppressing the two $\ell$-type quantum numbers that must be zero.
Classes IIa and IIb are related to each other by  exchanging $A \leftrightarrow B$ for the Schwinger bosons, and $I\leftrightarrow K$ for the vertex legs.
A pictorial representation of some representative states in these nondegenerate classes is given in Fig. \ref{fig:state-classes}.
\begin{figure}
    \centering
    \includegraphics{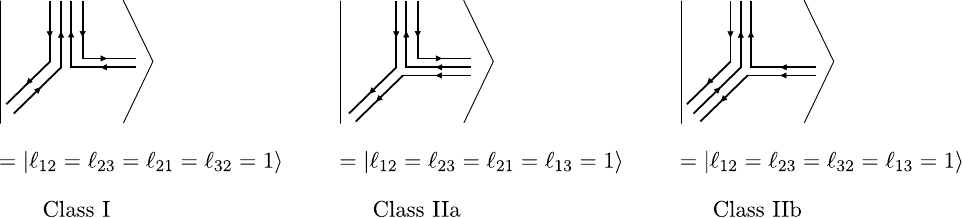}
    \caption{One-dimensional nondegenerate subspaces are states where at least two $\ell_{IJ}$ quantum numbers are zero. Putting aside the trilinear $t$ excitation, any choice of the bilinear $\ell_{IJ}$ quantum numbers can be pictorially represented similar to the above examples.
}
    \label{fig:state-classes}
\end{figure}

Note that nondegenerate basis kets with more than the minimal number of vanishing $\ell$-type quantum numbers belong to multiple classes;
e.g., states of the form $\kett{\ell_{12};t}$ belong to all three classes, while states of the form $\kett{\ell_{23},\ell_{31},\ell_{32};t}$ belong to Classes I and IIa.

The naive basis kets that are nondegenerate can be normalized in closed form by evaluating the norm-squared of the states, using recursion, on a class-by-class basis.
The process to do this is provided in some detail in Appendix \ref{app: norm calc}.
For Class I, the squared norms are evaluated to be
\begin{align}
    &\bbrakett{ \ell_{IJ},\ell_{JK}, \ell_{JI}, \ell_{KJ} ; t | \ell_{IJ},\ell_{JK}, \ell_{JI}, \ell_{KJ} ; t } \nonumber\\
    &= \tfrac{1}{2} (\ell_{IJ}+\ell_{JI}+\ell_{JK}+\ell_{KJ}+|t| + 2) \ell_{IJ}! \, \ell_{JK}! \, \ell_{JI}! \, \ell_{KJ}! \, |t|! \, (\ell_{IJ}+\ell_{KJ}+|t| +1)! \, (\ell_{JK}+\ell_{JI}+|t| +1)! \, .
\end{align}
For Classes IIa and IIb, the squared norms are evaluated to be
\begin{align}
    &\bbrakett{\ell_{IJ}, \ell_{JK}, \ell_{JI} ,\ell_{IK}; t | \ell_{IJ} , \ell_{JK}, \ell_{JI},\ell_{IK}; t } \nonumber \\
    &= \frac{1}{2} \tfrac{(\ell_{IJ}+\ell_{JK}+\ell_{JI}+\ell_{IK}+|t|+2) \, \ell_{IJ}! \, \ell_{JK}! \, \ell_{JI}! \, \ell_{IK}! \, |t| ! \, (\ell_{IJ}+\ell_{IK}+|t|+1)! \, (\ell_{JK}+\ell_{JI}+|t|+1)! \, \binom{\ell_{IJ}+\ell_{JK}+\ell_{JI}+\ell_{IK}+|t|+1}{\ell_{IK}}}{\binom{\ell_{IJ}+\ell_{JI}+\ell_{IK}+|t|+1}{\ell_{IK}}} \, , \\
    &\bbrakett{\ell_{IJ},\ell_{JK},\ell_{KJ},\ell_{IK};t | \ell_{IJ},\ell_{JK},\ell_{KJ},\ell_{IK};t } \nonumber \\
    &= \frac{1}{2} \tfrac{(\ell_{IJ}+\ell_{JK}+\ell_{KJ}+\ell_{IK}+|t|+2) \, \ell_{IJ}! \, \ell_{JK}! \, \ell_{KJ}! \, \ell_{IK}! \, |t| ! \, (\ell_{IJ}+\ell_{KJ}+|t|+1)! \, (\ell_{JK}+\ell_{IK}+|t|+1)! \, \binom{\ell_{IJ}+\ell_{JK}+\ell_{KJ}+\ell_{IK}+|t|+1}{\ell_{IK}}}{\binom{\ell_{JK}+\ell_{KJ}+\ell_{IK}+|t|+1}{\ell_{IK}}}  \, ,
\end{align}
(with binomial coefficients in the denominators).Equipped with these norms, we can express orthonormal states as
\begin{align}
    & \ket{\ell_{IJ},\ell_{JK}, \ell_{JI}, \ell_{KJ} ; t} = \tfrac{(L_{IJ}^{\dagger})^{\ell_{IJ}} (L_{JK}^{\dagger})^{\ell_{JK}} (L_{JI}^{\dagger})^{\ell_{JI}} (L_{KJ}^{\dagger})^{\ell_{KJ}} \ket{0}}{\sqrt{
\tfrac{1}{2} (\ell_{IJ}+\ell_{JI}+\ell_{JK}+\ell_{KJ}+|t| + 2) \ell_{IJ}! \, \ell_{JK}! \, \ell_{JI}! \, \ell_{KJ}! \, |t|! \, (\ell_{IJ}+\ell_{KJ}+|t| +1)! \, (\ell_{JK}+\ell_{JI}+|t| +1)!}} , \\
    & \ket{\ell_{IJ},\ell_{JK},\ell_{JI},\ell_{IK};t} = \nonumber \\
    & \quad \tfrac{(L_{IJ}^{\dagger})^{\ell_{IJ}} (L_{JK}^{\dagger})^{\ell_{JK}} (L_{JI}^{\dagger})^{\ell_{JI}} (L_{IK}^{\dagger})^{\ell_{IK}} \ket{0}}{
\sqrt{\frac{1}{2} \tfrac{(\ell_{IJ}+\ell_{JK}+\ell_{JI}+\ell_{IK}+|t|+2) \, \ell_{IJ}! \, \ell_{JK}! \, \ell_{JI}! \, \ell_{IK}! \, |t| ! \, (\ell_{IJ}+\ell_{IK}+|t|+1)! \, (\ell_{JK}+\ell_{JI}+|t|+1)! \, \binom{\ell_{IJ}+\ell_{JK}+\ell_{JI}+\ell_{IK}+|t|+1}{\ell_{IK}}}{\binom{\ell_{IJ}+\ell_{JI}+\ell_{IK}+|t|+1}{\ell_{IK}}}}} , \\
    & \ket{\ell_{IJ},\ell_{JK},\ell_{KJ},\ell_{IK};t} = \nonumber \\
    & \quad \tfrac{(L_{IJ}^{\dagger})^{\ell_{IJ}} (L_{JK}^{\dagger})^{\ell_{JK}} (L_{KJ}^{\dagger})^{\ell_{KJ}} (L_{IK}^{\dagger})^{\ell_{IK}} \ket{0}}{
   \sqrt{\frac{1}{2} \tfrac{(\ell_{IJ}+\ell_{JK}+\ell_{KJ}+\ell_{IK}+|t|+2) \, \ell_{IJ}! \, \ell_{JK}! \, \ell_{KJ}! \, \ell_{IK}! \, |t| ! \, (\ell_{JK}+\ell_{IK}+|t|+1)! \, (\ell_{IJ}+\ell_{KJ}+|t|+1)! \, \binom{\ell_{IJ}+\ell_{JK}+\ell_{KJ}+\ell_{IK}+|t|+1}{\ell_{IK}}}{\binom{\ell_{JK}+\ell_{KJ}+\ell_{IK}+|t|+1}{\ell_{IK}}}}} .
\end{align}
Above, the unit-normalized basis kets have been written with the notation $\ket{\cdots}$ to distinguish them from the naive basis kets $\kett{\cdots}$.
The above formulas have been confirmed to agree with software-assisted implementations (see the Supplementary Material \cite{SM}) for a wide variety of choices for the quantum numbers.

\section{\label{sec: degenerate}Degenerate subspaces}

\subsection{SU(3) multiplicities, or Littlewood-Richardson coefficients}

We have noted that, in the case of SU(2), the irreps $j_I$ for $I=1,2,3$ are enough to completely characterize a gauge-invariant state of a trivalent vertex. 
This is because if irreps $j_1$, $j_2$, and $j_3$ can be added to form angular momentum 0, then there is a unique way to do this by way of the Wigner $3j$ symbols.
An equivalent statement is that if the tensor product of $j_1 \otimes j_2$ contains $j_3$, then there is one and only one copy of $j_3$.

For general groups, the tensor product of two irreps $\lambda$ and $\mu$ may contain irreps for which there are multiple copies:
\begin{align}
    \lambda \otimes \mu &= \bigoplus_\nu d_{\lambda,\mu}^\nu \, \nu
\end{align}
where $d_{\lambda,\mu}^\nu$ is the multiplicity of irrep $\nu$, a nonnegative integer also known as a Littlewood-Richardson coefficient (LRC)~\cite{Littlewood1934GroupCA,knuth1970permutations,GopalkrishnaGadiyar:1991ah,2023CMaPh.400..179C}.
This is most familiar for the case of SU(3) in the example of $(1,1) \otimes (1,1)$, which contains two copies of the (1,1) octet: $$d_{(1,1),(1,1)}^{(1,1)}=2.$$
For SU(2), the nonzero LRCs are always equal to one.
But for SU(3), the nonzero LRCs can be any positive integer.
Three irreps $(p_I,q_I)$, for $I=1,2,3$, meeting at the vertex then characterize a subspace with dimension equal to the LRC:
\begin{align}
    \dim ( \vecpq ) &= d_{(p_1,q_1),(p_2 , q_2)}^{(q_3, p_3)} = d_{(p_1,q_1),(p_3 , q_3)}^{(q_2, p_2)} = d_{(p_2,q_2),(p_3 , q_3)}^{(q_1, p_1)} .
\end{align}
Since $\dim ( \vecpq ) > 1$ is possible, the $(p_I,q_I)$ do not in general provide a full characterization of a state, which is to say the irreps alone do not constitute a complete set of commuting observables.

In order to have a basis for the subspace characterized by the $(p_I, q_I)$, one must identify a number of linearly independent states equal to the associated LRC.
Our reference point for evaluating LRCs is the Littlewood-Richardson Calculator \texttt{lrcalc} library available in SageMath~\cite{sagemath}, with suitable processing of the output to make it applicable to SU(3).\footnote{\JS{The output of \texttt{lrcalc.mult} was filtered for Young tableaux with a maximum column height of three; each surviving tableau had columns with height equal to three deleted from it.}}
Numerical checks across all $\vecpq$ sectors we have explored indicate that LRCs can be correctly predicted by assuming that all naive LSH basis kets belonging to a given sector $\vecpq$ form a basis for that sector, and then counting up the number of basis kets.
We will now present how that counting is done. 

Suppose that we have a valid choice of $\vecpq$, i.e., a choice of $\vecpq$ such that $\dim ( \vecpq )\geq 1$, and we wish to find all LSH basis kets that belong to that sector.
We effectively want to invert \eqref{eq:PQ from llllllt} to find solutions to the $\ell_{IJ}$ and $t$ in terms of $\vecpq$.
It is easy to solve for $t$ uniquely because \eqref{eq:PQ from llllllt} can all be combined to obtain
\begin{align}
    t = \frac{1}{3} ( p_1 + p_2 + p_3 - q_1 - q_2 - q_3).
\end{align}
What is not uniquely determined is the choice of $\ell_{IJ}$'s.

Knowing $t$, and how it contributes to the $p_I$ and $q_I$ values, we may focus on the contributions to $p_I$ and $q_I$ coming from the $\ell_{IJ}$'s:
\begin{subequations}
\begin{align}
    p_1' &= \ell_{12} + \ell_{13} = p_1 - |t| \, \theta( t) , \\
    q_1' &= \ell_{31} + \ell_{21} = q_1 - |t| \, \theta(-t) , \\
    p_2' &= \ell_{23} + \ell_{21} = p_2 - |t| \, \theta( t) , \\
    q_2' &= \ell_{12} + \ell_{32} = q_2 - |t| \, \theta(-t) , \\
    p_3' &= \ell_{31} + \ell_{32} = p_3 - |t| \, \theta( t) , \\
    q_3' &= \ell_{23} + \ell_{13} = q_3 - |t| \, \theta(-t) .
\end{align}
\end{subequations}
It is easy to see that there is a linear dependence among these equations because $\sum_I p_I ' = \sum_I q_I '$. 
We can therefore remove one equation, say, that for $q_3'$:
\begin{subequations}
\begin{align}
    p_1' &= \ell_{12} + \ell_{13} , \\
    q_1' &= \ell_{31} + \ell_{21} , \\
    p_2' &= \ell_{23} + \ell_{21} , \\
    q_2' &= \ell_{12} + \ell_{32} , \\
    p_3' &= \ell_{31} + \ell_{32} .
\end{align}
\end{subequations}
This system of five equations has six unknown $\ell_{IJ}$ variables, so it is reasonable to expect a family of solutions.
Since all the $\ell_{IJ}$ are on equal footings, let us \JS{arbitrarily} choose to treat $\ell_{31}$ as a parameter.
\JS{The five equations can then be inverted} to obtain
\begin{subequations}
\label{eq: solution to lij's}
\begin{align}
    \ell_{12} &= \ell_{31} - p_3' + q_2' , \\
    \ell_{23} &= \ell_{31} + p_2' - q_1' , \\
    \ell_{21} &= -\ell_{31} + q_1' , \\
    \ell_{32} &= -\ell_{31} + p_3' , \\
    \ell_{13} &= -\ell_{31} + p_1' + p_3' - q_2' .
\end{align}
\end{subequations}
\JS{One may now freely} choose $\ell_{31}$, so long as all of the other $\ell_{IJ}$'s evaluate to nonnegative integers.
Using the above equations, the requirements that $\ell_{12}$, $\ell_{23}$ and $\ell_{31}$ be nonnegative \JS{lead to}
\begin{align}
    \ell_{31} - p_3' + q_2' & \geq 0 , \\
    \ell_{31} + p_2' - q_1' & \geq 0 , \\
    \ell_{31} &\geq 0 , \\
    \Rightarrow \ell_{31} &\geq \max (p_3' - q_2', q_1' - p_2', 0 ).
\end{align}
Meanwhile, the requirements that $\ell_{21}$, $\ell_{32}$ and $\ell_{13}$ be nonnegative give
\begin{align}
    -\ell_{31} + q_1' &\geq 0 , \\
    -\ell_{31} + p_3' &\geq 0 , \\
    -\ell_{31} + p_1' + p_3' - q_2' &\geq 0 , \\
    \Rightarrow \ell_{31} &\leq \min( q_1', p_3', p_1' + p_3' - q_2' ) .
\end{align}
In the end, we are not interested in the actual value of $\ell_{31}$, but rather the number of valid solutions.
The above bounds indicate that the number of solutions is
\begin{align}
\label{eq: dimension of pq sector}
    \dim (p_1, q_1, p_2, q_2, p_3, q_3 ) &= 1 + \min( q_1', p_3', p_1' + p_3' - q_2' ) - \max (p_3' - q_2', q_1' - p_2', 0 ).
\end{align}

To obtain \eqref{eq: dimension of pq sector}, we assumed that the choices of $(p_I, q_I)$ were valid for forming a singlet configuration.
This derivation in fact gives a test for detecting valid choices of $(p_I, q_I)$:
if the calculated $t$ is an integer, and the limits for $\ell_{31}$ are sensible, i.e., \[\min( q_1', p_3', p_1' + p_3' - q_2' ) \geq \max (p_3' - q_2', q_1' - p_2', 0 ) \geq 0,\] then there exists an LSH basis state with the given $\vecpq$ quantum numbers, and therefore they are valid.

As mentioned earlier, we have found the LSH-state-counting method of evaluating $\dim (\vecpq)$ to agree with the output derived from $\texttt{lrcalc}$ in every sector we have checked.
These data can be aggregated into a summary of the total Hilbert space dimensionality of states up to and including some cutoff value $p_{\textrm{max}}=q_{\textrm{max}}=\Lambda$.
Table \ref{tab: total dimensionality with cutoff} presents these dimensionalities, obtained by both methods of counting, up to $\Lambda=10$.

\begin{table}[t]
    \centering
    \begin{tabular}{c|r}
        $\Lambda=p_{\textrm{max}}=q_{\textrm{max}}$ & Dimensionality \\
        \hline \hline
        1 & 20 \\
        2 & 185 \\
        3 & 1,023 \\
        4 & 4,072 \\
        5 & 12,937 \\
        6 & 34,940 \\
        7 & 83,454 \\
        8 & 181,107 \\
        9 & 363,910 \\
        10 & 686,531
    \end{tabular}
    \caption{The total number of independent states for the trivalent vertex, subject to a cutoff on irreps: $\max(p_1,q_1,p_2,q_2,p_3,q_3)\leq \Lambda$.}
    \label{tab: total dimensionality with cutoff}
\end{table}

One other observation is in order.
Earlier, \JS{it was} remarked that all naive LSH states belonging to a given $\vecpq$ sector must be related by \eqref{eq:PQ shift symmetry}.
This is apparent from the solution to $\ell_{IJ}$'s given in \eqref{eq: solution to lij's}, where we see that increasing the value of $\ell_{31}$ by one induces the same change to $\ell_{12}$ and $\ell_{23}$, while simultaneously decreasing $\ell_{21}$, $\ell_{32}$, and $\ell_{13}$ by one.

\subsection{Overlap matrices}

Within those subspaces for which $\dim (p_1,q_1,p_2,q_2,p_3,q_3)>1$, we find that the naive LSH basis states generally fail to be orthogonal, presenting an obstacle to a complete and orthonormal basis suitable for computation.
To orthogonalize the degenerate subspaces, the most obvious brute-force solution is to apply the Gram-Schmidt procedure on the set of naive basis states.
If one knows the overlaps between all naive basis states belonging to a given sector, they will have sufficient information to construct an orthonormal basis.
In principle, Gram-Schmidt orthogonalization of each $\vecpq$ sector constitutes a solution to the complete construction of an orthonormal basis.
\JS{Unfortunately, that} approach offers \JS{no} insight into the nature of the \JS{overlap} problem and has no connection to a seventh Casimir.

An ideal choice for the seventh Casimir would be a Hermitian operator that i) commutes with all $P_I$ and $Q_I$, ii) always has a nondegenerate spectrum once $\vecpq$ is fixed (that way no further labels are needed), and iii) has eigenstates that can be constructed systematically.
Using the LSH operators, several choices can easily be made that manifestly satisfy (i).
To satisfy (ii), we have deferred to computer algebra software (included in the Supplementary Material \cite{SM}) to confirm the nondegenerate spectrum, and the choices we have examined seem to generally satisfy this criterion.
We have not yet been successful in finding a solution to (iii);
this problem remains open and we return to it in Sec.~\ref{sec: discussion}.

Below, we tabulate the overlaps between naive LSH basis kets in a variety of sectors of fixed $\vecpq$.
This serves a few purposes:
Follow-up works can be done with exact results to compare with, 
the values themselves may be used as a reference,
and patterns might later be identified that help toward an analytic resolution to the orthogonalization problem.
Below, we display results for a selection of two- and three-dimensional sectors, and some additional data for four-dimensional sectors may be found in Appendix \ref{app: extra results}.

\arraycolsep=1.4pt\def\arraystretch{1.2}

\noindent $\vecpq=(1,1,1,1,1,1)$:
\begin{align}
    \left( \begin{array}{c}
        \bbra{0\,0\,0;1\,1\,1;0 } \\
        \bbra{1\,1\,1;0\,0\,0;0 } 
    \end{array} \right)
    \left( \begin{array}{ccc}
        \kett{0\,0\,0;1\,1\,1;0} , &
        \kett{1\,1\,1;0\,0\,0;0}
    \end{array} \right)
    &=
    \left( \begin{array}{cc}
        \frac{56}{3} & \frac{-16}{3} \\
        \frac{-16}{3} & \frac{56}{3} 
    \end{array} \right),
\end{align}

\noindent $\vecpq=(2,1,2,1,2,1)$:
\begin{align}
    \left( \begin{array}{c}
        \bbra{0\,0\,0;1\,1\,1;1 } \\
        \bbra{1\,1\,1;0\,0\,0;1 }
    \end{array} \right)
    \left( \begin{array}{ccc}
        \kett{0\,0\,0;1\,1\,1;1} , &
        \kett{1\,1\,1;0\,0\,0;1}
    \end{array} \right)
    &=
\left(
\begin{array}{cc}
 315 & -45 \\
 -45 & 315 \\
\end{array}
\right), 
\end{align}

\noindent $\vecpq=(2,1,1,2,1,1)$:
\begin{align}
    \label{eq: overlap211211}
    \left( \begin{array}{c}
        \bbra{1\,0\,0;1\,1\,1;0 } \\
        \bbra{2\,1\,1;0\,0\,0;0 }
    \end{array} \right)
    \left( \begin{array}{ccc}
        \kett{1\,0\,0;1\,1\,1;0} , &
        \kett{2\,1\,1;0\,0\,0;0}
    \end{array} \right)
    &=
\left(
\begin{array}{cc}
 \frac{185}{2} & -35 \\
 -35 & 130 \\
\end{array}
\right),
\end{align}

\noindent $\vecpq=(2,1,2,2,1,2)$:
\begin{align}
    \left( \begin{array}{c}
        \bbra{1\,1\,0;1\,1\,1;0 } \\
        \bbra{2\,2\,1;0\,0\,0;0 }
    \end{array} \right)
    \left( \begin{array}{ccc}
      \kett{1\,1\,0;1\,1\,1;0} , &
      \kett{2\,2\,1;0\,0\,0;0}
    \end{array} \right)
    &=
\left(
\begin{array}{cc}
 432 & -216 \\
 -216 & 864 \\
\end{array}
\right),
\end{align}

\noindent $\vecpq=(2,2,2,2,2,2)$:
\begin{align}
    \left( \begin{array}{c}
      \bbra{0\,0\,0;2\,2\,2;0} \\
      \bbra{1\,1\,1;1\,1\,1;0} \\
      \bbra{2\,2\,2;0\,0\,0;0}
    \end{array} \right)
    \left( \begin{array}{ccc}
      \kett{0\,0\,0;2\,2\,2;0} , &
      \kett{1\,1\,1;1\,1\,1;0} , &
      \kett{2\,2\,2;0\,0\,0;0}
    \end{array} \right)
    &=
    \left( \begin{array}{ccc}
        \frac{136512}{25} & -\frac{31104}{25} & \frac{15552}{25} \\
        -\frac{31104}{25} & \frac{47088}{25} & -\frac{31104}{25} \\
        \frac{15552}{25} & -\frac{31104}{25} & \frac{136512}{25} 
    \end{array} \right), 
\end{align}

\noindent $\vecpq=(3,2,2,3,2,2)$:
\begin{align}    
    \left( \begin{array}{c}
      \bbra{1\,0\,0;2\,2\,2;0} \\
      \bbra{2\,1\,1;1\,1\,1;0} \\
      \bbra{3\,2\,2;0\,0\,0;0}
    \end{array} \right)
    \left( \begin{array}{ccc}
      \kett{1\,0\,0;2\,2\,2;0} , &
      \kett{2\,1\,1;1\,1\,1;0} , &
      \kett{3\,2\,2;0\,0\,0;0}
    \end{array} \right)
    &=
    \left( \begin{array}{ccc}
 \frac{200256}{5} & -\frac{65184}{5} & \frac{36288}{5} \\
 -\frac{65184}{5} & \frac{101136}{5} & -\frac{74592}{5} \\
 \frac{36288}{5} & -\frac{74592}{5} & \frac{350784}{5}
    \end{array} \right).
\end{align}

\subsection{Seventh Casimir candidates and orthogonalization}

We have noted that a good seventh quantum number or seventh Casimir would be a Hermitian operator that conserves all the $p_I$ and $q_I$, and has a nondegenerate spectrum within any fixed $\vecpq$ sector.
To conserve all the $(p_I, q_I)$, note that LSH gauge-singlet operators all induce definite changes on $\vecpq$; for example, $L_{12}^\dagger$ applied to a ket increases $p_1$ and $q_2$ by 1 and does not change any of $\{ p_2,q_1,p_3,q_3 \}$, while $T_A$ decreases all $p_I$ by one and leaves all $q_I$ alone.
One can then multiply operators together to form a candidate seventh Casimir that makes no net change on any of the $p_I$ or $q_I$.

Additionally, it may be useful to construct the seventh Casimir such that it destroys states that are necessarily nondegenerate.
Noting that degeneracy requires all $p_I \geq 1$ and $q_I \geq 1$, this property can be enforced by putting to the right-hand side of the operator some product that lowers all $p_I$ and $q_I$
by one;
three obvious choices are $L_{12} L_{23} L_{31} $, $L_{21} L_{32} L_{13} $, and $ T_A T_B $.
If we then multiply with the adjoint of such a product, we will have something that conserves the $P_I$ and $Q_I$ and is manifestly Hermitian:
$(L_{12} L_{23} L_{31}) ^\dagger L_{12} L_{23} L_{31} $, $(L_{21} L_{32} L_{13})^\dagger L_{21} L_{32} L_{13} $, and $(T_A T_B)^\dagger T_A T_B $ are three candidates for the seventh Casimir.
In the following, we will focus on the choice
\newcommand{\cseven}{C_T}
\begin{align}
    C_T &\equiv (T_A T_B)^\dagger T_A T_B .
\end{align}
This is not the only candidate we have studied, but in the low-dimensional subspaces we have calculated it seems to have matrix elements and eigenstates that are relatively simple when compared to other candidates.
Furthermore, the use of $T_A T_B$ seems like a natural choice given its appearance in the special constraint \eqref{eq:TT=LLL+LLL}.

To diagonalize (i.e., find the eigenbasis of) $\cseven$ in a given $\vecpq$ sector, it is sufficient to evaluate the overlap matrix in that sector and the matrix elements of $\cseven$ between the naive basis states.
We let $O$ and $\tilde{C}$ denote the $\dim(\vecpq)\times\dim (\vecpq)$ matrices whose elements are defined by
\begin{align}
    O_{ij} &\equiv \bbrakett{i | j}, \\
    \tilde{C}_{ij} &\equiv \bbra{i} \cseven \kett{j},
\end{align}
where $\kett{j}$ is the $j^{\mathrm{th}}$ naive basis ket in the sector $\vecpq$ (ordered by increasing $\ell_{12}$, $\ell_{23}$, and $\ell_{31}$).
The matrices $O$ and $\tilde{C}$ can be used to obtain the matrix representation of $\cseven$ with respect to the naive basis, that is, the $\dim(\vecpq)\times\dim (\vecpq)$ matrix $[\cseven]_{\mathrm{naive}}$ whose elements $([\cseven]_{\mathrm{naive}})_{ij}$ are defined by
\begin{align}
    \cseven \kett{j} &\equiv \sum_i ([\cseven]_{\mathrm{naive}})_{ij} \kett{i}.
\end{align}
The relationship among these three matrices is $\tilde{C} = O [\cseven]_{\mathrm{naive}} $, or
\begin{align}
    [\cseven]_{\mathrm{naive}} &= O^{-1} \tilde{C} .
\end{align}
One then finds the eigenvalues and eigenvectors of $[\cseven]_{\mathrm{naive}}$, where of course the eigenvectors will have coordinates with respect to the naive basis kets.

For normalization, let $\vec{v}_i$ denote the eigenvectors of $[\cseven]_{\mathrm{naive}}$.
If $S$ is a matrix with components
\begin{align}
    S_{ij} &= (\vec{v}_j)_i,
\end{align}
so that the unnormalized eigenstates are
\begin{align}
    \kett{\phi_i} &= \sum_j S_{ji} \kett{j},
\end{align}
then the eigenstates' squared norms are evaluated as
\begin{align}
    \bbrakett{ \phi_i | \phi_i } &= (S^\dagger O S)_{ii}  \qquad \textrm{(no sum).}
\end{align}
Thus, the normalized eigenstate is
\begin{align}
    \ket{\phi_i} &= \frac{1}{\sqrt{(S^\dagger O S)_{ii}}} \sum_j S_{ji} \kett{j} .
\end{align}

Below we present the diagonalization of $\cseven$ for a variety of low-dimensional sectors.
The results are presented as follows:
In the sector $(p_1,q_1,p_2,q_2,p_3,q_3)$, the eigenvalues of $\cseven$ are denoted by $\mathrm{Spec}_{p_1 q_1 p_2 q_2 p_3 q_3} ( \cseven ) $ and presented in increasing order.
The eigenstates of $\cseven$, $\kett{\phi_n}_{p_1 q_1 p_2 q_2 p_3 q_3}$, are ordered in the same way, with $n=1$ [$n=\dim (p_1,q_1,p_2,q_2,p_3,q_3)$] identifying the eigenvector with the lowest (highest) eigenvalue.

Note that, in all sectors we have explored, we have confirmed the eigenvalues of $\cseven$ to be nondegenerate, and that the lowest eigenvalue is zero.
We have also confirmed some of the nondegenerate states to be in the nullspace of $\cseven$, including nontrivial examples where $\min (p_1,q_1,p_2,q_2,p_3,q_3) \geq 1$.
We believe these properties of $\cseven$---nondegeneracy, and a lowest eigenvalue equal to zero---are applicable to any given $\vecpq$ sector, but we do not have a proof.

\begin{align}
    \mathrm{Spec}_{111111} \bigl( \cseven \bigr) &= \left\{ 0 , \frac{80}{3} \right\} , \\
    \left( \begin{array}{c}
      \kett{\phi_1 }_{111111} \\
      \kett{\phi_2 }_{111111}
    \end{array} \right)
    &= 
\left(
\begin{array}{cc}
 1 & -1 \\
 1 & 1 \\
\end{array}
\right)
    \left( \begin{array}{c}
         \kett{0\,0\,0;1\,1\,1;0} \\
         \kett{1\,1\,1;0\,0\,0;0} \\
    \end{array} \right) .
\end{align}

\begin{align}
    \mathrm{Spec}_{212121} \bigl( \cseven \bigr) &= \left\{ 0 , 90 \right\} , \\
    \left( \begin{array}{c}
      \kett{\phi_1 }_{212121} \\
      \kett{\phi_2 }_{212121}
    \end{array} \right)
    &= 
\left(
\begin{array}{cc}
 1 & -1 \\
 1 & 1 \\
\end{array}
\right)
    \left( \begin{array}{c}
         \kett{0\,0\,0;1\,1\,1;1} \\
         \kett{1\,1\,1;0\,0\,0;1} \\
    \end{array} \right) .
\end{align}

\begin{align}
    \mathrm{Spec}_{211211} \bigl( \cseven \bigr) &= \left\{ 0, \frac{305}{6} \right\} , \\
    \left( \begin{array}{c}
      \kett{\phi_1 }_{211211} \\
      \kett{\phi_2 }_{211211}
    \end{array} \right)
    &= 
\left(
\begin{array}{cc}
 1 & -\frac{23}{38} \\
 1 & 1 \\
\end{array}
\right)
    \left( \begin{array}{c}
         \kett{1\,0\,0;1\,1\,1;0} \\
         \kett{2\,1\,1;0\,0\,0;0} \\
    \end{array} \right) .
\end{align}

\begin{align}
    \mathrm{Spec}_{212212} \bigl( \cseven \bigr) &= \left\{ 0, 108 \right\} , \\
    \left( \begin{array}{c}
      \kett{\phi_1 }_{212212} \\
      \kett{\phi_2 }_{212212}
    \end{array} \right)
    &= 
\left(
\begin{array}{cc}
 1 & -\frac{1}{3} \\
 1 & 1 \\
\end{array}
\right)
    \left( \begin{array}{c}
         \kett{1\,1\,0;1\,1\,1;0} \\
         \kett{2\,2\,1;0\,0\,0;0} \\
    \end{array} \right) .
\end{align}

\begin{align}
    \mathrm{Spec}_{222222} \bigl( \cseven \bigr) &= \left\{0, \frac{1008}{5}, \frac{45684}{125} \right\} , \\
    \left( \begin{array}{c}
      \kett{\phi_1 }_{222222} \\
      \kett{\phi_2 }_{222222} \\
      \kett{\phi_3 }_{222222} \\
    \end{array} \right)
    &= 
\left(
\begin{array}{ccc}
    -\tfrac{37}{208} & 1 & -\tfrac{37}{208} \\
    1 & 0 & -1 \\
    \tfrac{1}{2} & 1 & \tfrac{1}{2} \\
\end{array}
\right)
    \left( \begin{array}{c}
         \kett{0\,0\,0;2\,2\,2;0} \\
         \kett{1\,1\,1;1\,1\,1;0} \\
         \kett{2\,2\,2;0\,0\,0;0} \\
    \end{array} \right) .
\end{align}

\begin{align}
    \mathrm{Spec}_{322322} \bigl( \cseven \bigr) &= \left\{
    0,\tfrac{7}{900} \left(59609-\sqrt{383666161}\right),\tfrac{7}{900} \left(59609+\sqrt{383666161}\right)
    \right\} , \\
    \left( \begin{array}{c}
      \kett{\phi_1 }_{322322} \\
      \kett{\phi_2 }_{322322} \\
      \kett{\phi_3 }_{322322} \\
    \end{array} \right)
    &= 
\left(
\begin{array}{ccc}
    -\frac{101}{334} & 1 & -\frac{64}{501} \\ 
    1 & \frac{31369-\sqrt{383666161}}{24500} & \frac{6869-\sqrt{383666161}}{24500} \\
    \frac{31369-\sqrt{383666161}}{24504} & 1 & \frac{\sqrt{383666161}-6865}{24504} \\
\end{array}
\right)
    \left( \begin{array}{c}
         \kett{1\,0\,0;2\,2\,2;0} \\
         \kett{2\,1\,1;1\,1\,1;0} \\
         \kett{3\,2\,2;0\,0\,0;0} \\
    \end{array} \right) .
\end{align}

\section{Discussion}
\noindent
\label{sec: discussion}

In this paper, we have presented a solution to the construction of SU(3) gauge-singlet operators and an orthonormal basis at a trivalent vertex.
We have done this without any need for SU(3) Clebsch-Gordon coefficients, and the underlying Hilbert space in which calculations are ultimately performed is nothing but the Hilbert space of a collection of harmonic oscillators that are all on the same footing.
The use of irreducible Schwinger bosons at the operator level was necessary to avoid $\mathrm{Sp}(2,\mathbb{R})$ redundancies that would arise from creating representations out of ordinary bosonic creation and annihilation operators.

Certain features carried over from the SU(2) loop-string-hadron formulation:
The construction of elementary excitations by contracting fundamental and/or antifundamental bosonic creation operators together via SU(3) invariant tensors again turns out to provide an elegant way of writing down a complete basis.
In SU(3), the LSH states can be categorized into a few different classes depending on which of the $\ell_{IJ}$ are vanishing.
For the subspace of states belonging to what we have called classes I, IIa, and IIb, we have calculated a closed-form analytic solution to the orthonormal basis states.
These states bear a clear resemblance to the orthonormal loop states of SU(2) Yang-Mills \cite{Raychowdhury:2018tfj,Raychowdhury:2019iki}.

States living outside of Classes I and II, however, differ in an important way from the states of the SU(2) theory.
For these states, the irreps of the vertex legs provide an incomplete characterization of the state, and one must appeal to a seventh Casimir -- a degree of freedom that we interpret as being `internal' to the vertex.
We are not the first group to explore this ``missing label'' or ``multiplicity problem'', that is an old topic of interest in the group theory of SU(3).
As mentioned, one may always deal with the multiplicity problem by applying Gram-Schmidt orthogonalization on the naive LSH basis states, but we propose to alternatively find and diagonalize a Hermitian operator whose spectrum is nondegenerate within any given $\vecpq$ sector.
With the tools available to us from the loop-string-hadron approach, we could easily identify a few candidate Hermitian operators to serve as the seventh Casimir.
The candidate that we have focused on, $\cseven$, seems well-motivated theoretically, and appears to have some nice properties.
In all examples we have checked, for a given $\vecpq$, we have found $\cseven$ to be nondegenerate and to have zero as its lowest eigenvalue.
In some of the low-dimensional sectors, we have found a number of eigenstates of $\cseven$ whose coefficients in the naive basis have simple ratios.
The departures from the simple ratios of coefficients in the higher-dimensional sectors could be a hint that the ideal choice of seventh Casimir is related to our $\cseven$ but somehow different.
Unfortunately, evaluating the matrix elements of $\cseven$ gets computationally expensive quickly as a function of the subspace dimension, such that finding the eigenstates becomes a much slower process than simple Gram-Schmidt orthogonalization.

We note that ultimately what is sought is a systematically constructible orthonormal basis; having this would allow one to evaluate the matrix elements of any gauge-singlet contracted operator in full generality, at which point the underlying basis for computation could be entirely transitioned into seven LSH quantum numbers and one could retire the system of $18$ harmonic oscillators.
We are optimistic that a closed-form solution does exist, although we have not been able to find it yet.
We think finding the eigenstates of a well-chosen Hermitian operator is a promising path toward the solution.
We speculate that what may be missing is some kind of ladder operator that would enable one to reach all eigenstates in a $\vecpq$ sector once one of them has been found.
The fact that there is always a zero in the spectrum of our $\cseven$ seems to support this idea; this would correspond to the state at the ``bottom of the ladder.''
It is perhaps not an accident that the lowest eigenvector of $\cseven$ in the example sectors we have shown always had rational coefficients in the naive basis (up to overall normalization), while we have confirmed higher states to often have irrational coefficients.
Another attractive possibility is that there may be a correspondence of states between different $\vecpq$ sectors;
the knowledge of an orthogonal basis in some sector might be used to construct the orthogonal basis of another sector.
Yet another possibility is that there may be ladder operators that connect orthogonal states in different $\vecpq$ sectors.

In \JS{the follow-up papers of this series}, we will further examine the properties of LSH gauge-singlet operators---particularly those that appear in the plaquette operators of higher-dimensional theories---in combination with bases that are suitable for computation.
\JS{One crucial step in this process will be translating all calculations into a purely LSH framework (freed from the underlying Schwinger boson framework), including the representation of gauge-singlets operators in the LSH basis.}
Additionally, we seek to understand the details of how ``point splitting'' should work for SU(3) vertices with four or six legs, which are relevant to square and cubic lattices, repectively.

\section{Acknowledgments}
The authors would like to thank Anthony Ciavarella, Zohreh Davoudi, David B.~Kaplan, John Lombard and Himadri Mukherjee for insightful conversations at various points throughout this work.
Work by JRS was supported by the U.S. Department of Energy (DOE), Office of Science under contract DE-AC02-05CH11231, partially through Quantum Information Science Enabled Discovery (QuantISED) for High Energy Physics (KA2401032).
JRS and SVK both received support from the U.S. Department of Energy’s Office of Science Early Career Award under award DE-SC0020271, for theoretical developments for simulating lattice gauge theories on quantum computers.
SVK acknowledges support by the U.S. DOE, Office of Science, Office of Nuclear Physics, InQubator for Quantum Simulation (IQuS) (award no. DE-SC0020970), and by the DOE QuantISED program through the theory consortium ``Intersections of QIS and Theoretical Particle Physics'' at Fermilab (Fermilab subcontract no. 666484).
SVK further acknowledges the support from the Department of Physics and the College of Arts and Sciences at the University of Washington.
Research of IR is supported by the  OPERA award (FR/SCM/11-Dec-2020/PHY) from BITS-Pilani, the Start-up Research Grant (SRG/2022/000972) and Core-Research Grant (CRG/2022/007312) from ANRF, India and the cross-discipline research fund (C1/23/185) from BITS-Pilani. AN is supported by the Start-up Research Grant (SRG/2022/000972) from ANRF, India received by IR.

\section{Data Availability}
The data supporting this study's findings are available within the article and the Supplemental Material \cite{SM}. The Supplemental Material includes software for generating (reproducing) the original data. The software and the pre-generated data are also publicly available via GitHub \cite{SM}.

\appendix

\section{\label{app: norm calc}Normalization of nondegenerate states}

In this appendix, we provide partial proofs for the normalization of SU(3) LSH basis kets at a trivalent vertex, for choices of $\vecpq=(p_1, q_1, p_2, q_2, p_3, q_3)$ that are nondegenerate.
[Analogous calculations were discussed in an appendix of Ref.~\cite{Raychowdhury:2019iki} for the SU(2) loop-string-hadron theory with one flavor of fermions.]
As discussed in the main text, such nondegenerate states must have at least two vanishing $\ell_{IJ}$ quantum numbers, and they were organized into classes I and II depending on which quantum numbers out of $\{ \ell_{IJ}\}$ and $t$ are vanishing.
The full proof for either class I or II, in each case involving five or more layers of proof-by-induction, is lengthy and we do not reproduce the complete arguments here.
But to give a taste of how these calculations are done, below we present the norm evaluation for naive basis states having up to two nonzero quantum numbers, $t$ and a single $\ell_{IJ}$, which can serve as the starting point for normalizing either class I or II.

In the norm evaluations, one encounters numerous gauge-singlet operators beyond just the pure-creation gauge singlet operators defined \eqref{eq: su3singlets} and their pure-annihilation adjoints.
For example, there are bilinear gauge-singlet operators constructed as a mixture of a creation and annihilation operators:
\begin{align}
N_{IJ}&\equiv A^\dagger_\alpha(I)A^\alpha(J) , \\
M_{IJ}&\equiv B^{\dagger\alpha}(I)B_\alpha(J) , 
\end{align}
where $I$ and $J$ label two different link ends attached to the vertex.
For the trivalent vertex, $I$ and $J$ each take values from 1 to 3 with $I \not= J$.
We will also make frequent reference to the irrep quantum numbers and a function of them:
\begin{align}
    P_I &= a(I)^{\dagger }_{\alpha} a(I)^\alpha \simeq A(I)^{\dagger }_{\alpha} A(I)^\alpha , \\
    Q_I &= b(I)^{\dagger \, \alpha} b(I)_\alpha \simeq B(I)^{\dagger \, \alpha} B(I)_\alpha , \\
    F_I &\equiv (P_I + Q_I + 2)^{-1} ,
\end{align}
where $\simeq$ means the identity is valid in the Hilbert space of irreducible Schwinger bosons.

The first set of excitations we will normalize is the $t$-type ones:
\begin{align}
\label{eq: t-only states}
    \kett{t} &= (T_A^\dagger )^{|t| \, \theta(t)} (T_B^\dagger )^{|t| \, \theta(-t)} \ket{0}.
\end{align}
It is convenient to first establish the following intermediate results.

\begin{lemma}
For two vertex legs $I$ and $J$, $I\neq J$, and any integer $t\geq 0$, we have
\begin{align}
    N_{IJ} T_A^{\dagger \, t} \ket{0} =0.
\end{align}
\end{lemma}

\begin{proof}
We use induction on $t$.

When $t=0$, the $A(J)$ annihilation operators in $N_{IJ}$ are directly applied to the vacuum state, giving $N_{IJ} \ket{0} = 0$.
Now taking $t\geq1$, we assume $N_{IJ} T_A^{\dagger \, t-1} \ket{0} =0$.
Then we evaluate
\begin{align}
    N_{IJ} T_A^{\dagger \, t} \ket{0} &= N_{IJ} T_A^\dagger T_A^{\dagger \, t-1} \ket{0} \\
    &= [ N_{IJ} , T_A^\dagger ] T_A^{\dagger \, t-1} \ket{0},
\end{align}
where we have used the inductive hypothesis for $t-1$.
We need the commutator
\begin{align}
    [ N_{IJ} , T_A^\dagger ]
    &= [ A(I)^\dagger_\mu A(J)^\mu , \epsilon^{\alpha \beta \gamma} A(1)^\dagger_\alpha A(2)^\dagger_\beta A(3)^\dagger_\gamma ] \\
    &= \epsilon^{\alpha \beta \gamma} A(I)^\dagger_\mu [ A(J)^\mu , A(1)^\dagger_\alpha A(2)^\dagger_\beta A(3)^\dagger_\gamma ] \\
    &= \epsilon^{\alpha \beta \gamma} A(I)^\dagger_\mu 
    \left(
    [ A(J)^\mu , A(1)^\dagger_\alpha ] A(2)^\dagger_\beta A(3)^\dagger_\gamma +
    A(1)^\dagger_\alpha [ A(J)^\mu , A(2)^\dagger_\beta ] A(3)^\dagger_\gamma + \nonumber \right. \\
    & \qquad \qquad \qquad \quad \left. + A(1)^\dagger_\alpha A(2)^\dagger_\beta [ A(J)^\mu , A(3)^\dagger_\gamma ]
    \right) \\
    &\simeq \epsilon^{\alpha \beta \gamma} A(I)^\dagger_\mu 
    \left(
    \delta_{J \, 1} \bigl( \delta^\mu_\alpha - F_1 B(1)^{\dagger \, \mu} B(1)_\alpha \bigr) A(2)^\dagger_\beta A(3)^\dagger_\gamma + \nonumber \right. \\
    & \qquad \qquad \qquad \quad + \delta_{J \, 2} A(1)^\dagger_\alpha \bigl( \delta^\mu_\beta - F_2 B(2)^{\dagger \, \mu} B(2)_\beta \bigr) A(3)^\dagger_\gamma + \nonumber \\
    & \qquad \qquad \qquad \quad \left. + \delta_{J \, 3} A(1)^\dagger_\alpha A(2)^\dagger_\beta \bigl( \delta^\mu_\gamma - F_3 B(3)^{\dagger \, \mu} B(3)_\gamma \bigr)
    \right) .
\end{align}
Looking at the first term of the last line, if $\delta_{J \, 1} \neq 0$, then $I=2$ or $I=3$;
in either case, the term will contain two factors of $A(I)^\dagger$ components, that is a symmetric tensor with two lower indices.
Contraction with the antisymmetric $\epsilon^{\alpha \beta \gamma}$ will then produce zero.
The same logic can be applied to eliminate the second and third terms.
Thus,
\begin{align}
    N_{IJ} T_A^{\dagger \, t} \ket{0} &= [ N_{IJ} , T_A^\dagger ] T_A^{\dagger \, t-1} \ket{0}
    = (0) T_A^{\dagger \, t-1} \ket{0} = 0.
\end{align}
\end{proof}
The next result is proved by a similar, but not identical, argument.
\begin{lemma}
For two vertex legs $I$ and $J$, $I\neq J$, and any integer $p \geq 0$, we have
\begin{align}
    N_{IJ} T_B^{\dagger \, p} \ket{0} =0.
\end{align}
\end{lemma}

\begin{proof}
We use induction on $p$.

When $p=0$, the $A(J)$ annihilation operators in $N_{IJ}$ are directly applied to the vacuum state, giving $N_{IJ} \ket{0} = 0$.
Now taking $p\geq1$, we assume $N_{IJ} T_B^{\dagger \, p-1} \ket{0} =0$.
Then we evaluate
\begin{align}
    N_{IJ} T_B^{\dagger \, p} \ket{0} &= N_{IJ} T_B^\dagger T_B^{\dagger \, p-1} \ket{0} \\
    &= [ N_{IJ} , T_B^\dagger ] T_B^{\dagger \, p-1} \ket{0},
\end{align}
where we have used the inductive hypothesis for $p-1$.
We need the commutator
\begin{align}
    [ N_{IJ} , T_B^\dagger ]
    &= [ A(I)^\dagger_\mu A(J)^\mu , \epsilon_{\alpha \beta \gamma} B(1)^{\dagger \, \alpha} B(2)^{\dagger \, \beta} B(3)^{\dagger \, \gamma} ] \\
    &= \epsilon_{\alpha \beta \gamma} A(I)^\dagger_\mu [ A(J)^\mu , B(1)^{\dagger \, \alpha} B(2)^{\dagger \, \beta} B(3)^{\dagger \, \gamma} ] \\
    &= \epsilon_{\alpha \beta \gamma} A(I)^\dagger_\mu 
    \left(
    [ A(J)^\mu , B(1)^{\dagger \, \alpha} ] B(2)^{\dagger \, \beta} B(3)^{\dagger \, \gamma} +
    B(1)^{\dagger \, \alpha} [ A(J)^\mu , B(2)^{\dagger \, \beta} ] B(3)^{\dagger \, \gamma} + \nonumber \right. \\
    & \qquad \qquad \qquad \quad \left. + B(1)^{\dagger \, \alpha} B(2)^{\dagger \, \beta} [ A(J)^\mu , B(3)^{\dagger \, \gamma} ]
    \right) \\
    &\simeq \epsilon^{\alpha \beta \gamma} A(I)^\dagger_\mu 
    \left(
    \delta_{J \, 1} \bigl(- F_1 B(1)^{\dagger \, \mu} A(1)^\alpha \bigr) B(2)^{\dagger \, \beta} B(3)^{\dagger \, \gamma} +  \delta_{J \, 2} B(1)^{\dagger \, \alpha} \bigl( - F_2 B(2)^{\dagger \, \mu} A(2)^\beta \bigr) B(3)^{\dagger \, \gamma} + \right. \nonumber \\
    & \qquad \qquad \qquad \quad \left. + \delta_{J \, 3} B(1)^{\dagger \, \alpha} B(2)^{\dagger \, \beta} \bigl( - F_3 B(3)^{\dagger \, \mu} A(3)^\gamma \bigr)
    \right) .
\end{align}
Looking at the first term of the last line, we see a factor of $A(1)^\alpha$ with no components of $A(1)^\dagger$ anywhere to the right of it.
Since $[ N_{IJ} , T_B^\dagger ]$ is being applied to the ket $T_B^{\dagger \, p-1} \ket{0}$, a $P_1=0$ eigenstate, the term carrying $A(1)^\alpha$ must evaluate to zero.
The same logic can be applied to eliminate the second and third terms.
Thus,
\begin{align}
    N_{IJ} T_B^{\dagger \, p} \ket{0} &= [ N_{IJ} , T_B^\dagger ] T_B^{\dagger \, p-1} \ket{0}
    = (0) T_B^{\dagger \, p-1} \ket{0} = 0.
\end{align}

\end{proof}
The two Lemmas above are combined into the following Corollary.
\begin{corollary}
\label{cor: NIJ on t states}
For two vertex legs $I$ and $J$, $I\neq J$, and any integer $t$, we have
\begin{align}
    N_{IJ} \kett{t} &= 0.
\end{align}
\end{corollary}

Equipped with \ref{cor: NIJ on t states}, we are ready to evaluate the normalization of $\kett{t}$ states.
\begin{theorem}
The norm-squared of states $\kett{t}$ in \eqref{eq: t-only states} is
\begin{align}
    \label{eq: t state norm squared}
    \bbrakett{t|t} &= \tfrac{1}{2}(|t|+2) \, |t|! \, (|t|+1)! \, (|t|+1)! \nonumber \\
    &= \tfrac{1}{2} (|t|+2)! \, (|t|+1)! \, |t|! \, . 
\end{align}
\end{theorem}
\begin{proof}
    
Without loss of generality, assume $t\geq0$;
the proof for $t\leq 0 $ follows identical reasoning.
We prove the formula by induction.

For the base case $t=0$, the squared norm evaluates as
\begin{align}
    \bbrakett{0|0} &= \braket{0|0}= 1,
\end{align}
because the complete bosonic vacuum $\ket{0}$ is normalized to unity by definition.
This agrees with the right-hand side of \eqref{eq: t state norm squared} when evaluated for $|t|=0$.

For the inductive step, take $t\geq 1$ and assume the formula holds for the state $\kett{t-1}$.
The norm evaluation begins with
\begin{align}
    \bbrakett{t|t} &= \bbra{t-1} T_A T_A^\dagger \kett{t-1} \\
    &= \bbra{t-1} T_A T_A^{\dagger \, t} \ket{0} \\
    &= \bbra{t-1} \bigl[ T_A,  T_A^{\dagger \, t} \bigr] \ket{0} \\
    &= \bbra{t-1} \sum_{k=0}^{t-1} T_A^{\dagger \, t-1-k}\bigl[ T_A,  T_A^{\dagger} \bigr] T_A^{\dagger \, k} \ket{0} . \label{eq: t state normalization, summation}
\end{align}
To proceed, one needs the commutator $\bigl[ T_A,  T_A^{\dagger} \bigr]$; but going into the evaluation of this commutator, one should take note of the ket on the right to which this operator is applied: $T_A^{\dagger \, k} \ket{0}$.

One may first expand out the definitions to obtain
\begin{align}
    \left[ T_A,  T_A^{\dagger} \right]
    &= \epsilon^{\lambda \mu \nu} \bigl[ T_A ,\, A(1)^\dagger_\lambda A(2)^\dagger_\mu A(3)^\dagger_\nu \bigr] \\
    &= \epsilon^{\lambda \mu \nu} \left(
    A(1)^\dagger_\lambda A(2)^\dagger_\mu \bigl[ T_A ,\, A(3)^\dagger_\nu \bigr] +
    A(1)^\dagger_\lambda \bigl[ T_A ,\, A(2)^\dagger_\mu \bigr] A(3)^\dagger_\nu +
    \bigl[ T_A ,\, A(1)^\dagger_\lambda \bigr] A(2)^\dagger_\mu A(3)^\dagger_\nu
    \right). \label{eq: TAm TAp commutator as three terms}
\end{align}
Here, we will use the information about the ket to the right when evaluating
\begin{align}
    [ T_A , A(3)^\dagger_\nu ]
    &= \epsilon_{\alpha \beta \gamma} [ A(1)^\alpha A(2)^\beta A(3)^\gamma , A(3)^\dagger_\nu ] \\
    &= \epsilon_{\alpha \beta \gamma} A(1)^\alpha A(2)^\beta [ A(3)^\gamma , A(3)^\dagger_\nu ] \\
    &\simeq \epsilon_{\alpha \beta \gamma} A(1)^\alpha A(2)^\beta ( \delta^\gamma_\nu - F_3 B(3)^{\dagger \gamma} B(3)_\nu ) \\
    &\to \epsilon_{\alpha \beta \gamma} A(1)^\alpha A(2)^\beta \delta^\gamma_\nu - 0 \\
    &= \epsilon_{\alpha \beta \nu } A(1)^\alpha A(2)^\beta,
\end{align}
and similarly for the other two terms of \eqref{eq: TAm TAp commutator as three terms}.
So now
\begin{align}
    \left[ T_A,  T_A^{\dagger} \right]
    &\to \epsilon^{\lambda \mu \nu} \bigl(
    A(1)^\dagger_\lambda A(2)^\dagger_\mu \epsilon_{\alpha \beta \nu } A(1)^\alpha A(2)^\beta +
    A(1)^\dagger_\lambda \epsilon_{\alpha \beta \mu } A(3)^\alpha A(1)^\beta A(3)^\dagger_\nu + \nonumber \\
    & \qquad \qquad + \epsilon_{\alpha \beta \lambda } A(2)^\alpha A(3)^\beta A(2)^\dagger_\mu A(3)^\dagger_\nu
    \bigr).
\end{align}
For the first term, it is easy to show that $ \epsilon^{\lambda \mu \nu} \epsilon_{\alpha \beta \nu } A(1)^\alpha A(2)^\beta = A(1)^\lambda A(2)^\mu - A(1)^\mu A(2)^\lambda $, so that
\begin{align}
    \epsilon^{\lambda \mu \nu} A(1)^\dagger_\lambda A(2)^\dagger_\mu \epsilon_{\alpha \beta \nu } A(1)^\alpha A(2)^\beta &= A(1)^\dagger \cdot A(1) A(2)^\dagger \cdot A(2) - A(1)^\dagger_\lambda N_{21} A(2)^\lambda \\
    &\simeq P_1 P_2 - A(1)^\dagger_\lambda N_{21} A(2)^\lambda .
\end{align}
Again using the information about the ket to the right, one can straightforwardly show that
$ A(1)^\dagger_\lambda N_{21} A(2)^\lambda \to N_{21} N_{12} - P_2 $ in this expression.
We have now ``fully evaluated'' the first term of \eqref{eq: TAm TAp commutator as three terms} as
\begin{align}
    \label{eq: TAm TAp commutator, 1-2 term}
    \epsilon^{\lambda \mu \nu} A(1)^\dagger_\lambda A(2)^\dagger_\mu \bigl[ T_A ,\, A(3)^\dagger_\nu \bigr] &\to P_1 P_2 + P_2 - N_{21} N_{12} ,
\end{align}
by which we mean that it is written entirely in terms of LSH gauge singlet operators.

For the second term of \eqref{eq: TAm TAp commutator as three terms}, similar steps and substitutions can be used to obtain
\begin{align}
    \epsilon^{\lambda \mu \nu} A(1)^\dagger_\lambda \bigl[ T_A ,\, A(2)^\dagger_\mu \bigr] A(3)^\dagger_\nu &\simeq P_1 A(3) \cdot A(3)^\dagger - N_{13} N_{31} .
\end{align}
Here one should use another substitution that is valid here but not in general: $A(3) \cdot A(3)^\dagger \to P_3 + 3$.
So then the second term of \eqref{eq: TAm TAp commutator as three terms} is evaluated as
\begin{align}
    \label{eq: TAm TAp commutator, 3-1 term}
    \epsilon^{\lambda \mu \nu} A(1)^\dagger_\lambda \bigl[ T_A ,\, A(2)^\dagger_\mu \bigr] A(3)^\dagger_\nu \to P_1 ( P_3 + 3 ) - N_{13} N_{31}.
\end{align}
The third term of \eqref{eq: TAm TAp commutator as three terms} can be evaluated using similar steps to those used above, arriving at
\begin{align}
    \label{eq: TAm TAp commutator, 2-3 term}
    \epsilon^{\lambda \mu \nu} \bigl[ T_A ,\, A(1)^\dagger_\lambda \bigr] A(2)^\dagger_\mu A(3)^\dagger_\nu \to ( P_2 +3) ( P_3 +3) - ( P_2 +3) - N_{32} N_{23} .
\end{align}

Putting together the three intermediate results \eqref{eq: TAm TAp commutator, 1-2 term}, \eqref{eq: TAm TAp commutator, 3-1 term}, and \eqref{eq: TAm TAp commutator, 2-3 term}, and simplifying, it has been shown that
\begin{align}
    \bigl[ T_A,  T_A^{\dagger} \bigr] &\to P_1 P_2 + P_2 P_3 + P_3 P_1 + 3 \bigl( P_1 + P_2 + P_3 \bigr) + 6 - N_{21} N_{12} - N_{32} N_{23} - N_{13} N_{31} ,
\end{align}
for the purposes of \eqref{eq: t state normalization, summation}.
The final step to dealing with $\bigl[ T_A,  T_A^{\dagger} \bigr]$ is addressing the $N_{IJ} N_{JI}$ terms.
It is at this point that we invoke \ref{cor: NIJ on t states}, which tells us that any $N_{JI}$ applied to a state containing only $t$-type excitations will evaluate to zero.
Thus, the last three terms above are dropped.

Returning to the norm evaluation,
\begin{align}
    &\bbrakett{t|t} \nonumber \\
    &= \bbra{t-1} \sum_{k=0}^{t-1} T_A^{\dagger \, t-1-k}
    \bigl(
    P_1 P_2 + P_2 P_3 + P_3 P_1 + 3 \bigl( P_1 + P_2 + P_3 \bigr) + 6 
    \bigr)
    T_A^{\dagger \, k} \ket{0} \\
    &= \bbra{t-1} \sum_{k=0}^{t-1} T_A^{\dagger \, t-1-k} ( k^2 + k^2 + k^2 + 3(k+k+k)+6 ) T_A^{\dagger \, k} \ket{0} \\
    &= \sum_{k=0}^{t-1} (3k^2 + 9k +6) \bbrakett{t-1|t-1} \\
    &= (t+2) (t+1) t \, \bbrakett{t-1|t-1} .
\end{align}
By hypothesis, $\bbrakett{t-1|t-1} = \tfrac{1}{2} (t+1)! \, t! \, (t-1)!$, so that
\begin{align}
    \bbrakett{t|t} &= (t+2) (t+1) \, t \tfrac{1}{2} (t+1)! \, t! \, (t-1)! = \tfrac{1}{2} (t+2)! \, (t+1)! \, t!
\end{align}
\end{proof}

The next type of excitation to be included and normalized is a single $\ell$-type one:
\begin{align}
\label{eq: lij and t states}
    \kett{\ell_{IJ} , t} &= (L_{IJ}^{\dagger})^{\ell_{IJ}} \kett{t}.
\end{align}
Similar to the evaluation of $ \bbrakett{t | t }$, it is  convenient to first establish an  intermediate result as follows.
\begin{lemma}
\label{lem: Nij on l and t states}
    For two vertex legs $I$ and $J$, $I \neq J$, any integer $t$, and any integer $\ell \geq 0$, we have
\begin{align}
    N_{IJ} L_{IJ}^{\dagger \, \ell} T_A^{\dagger \, t} \ket{0} =0.
\end{align}
\end{lemma}

\begin{proof}
 We use induction on $\ell$.

For the base case of $\ell=0$, the assertion is that $N_{IJ} T_A^{\dagger \, t} \ket{0} =0$, which is nothing but the statement of \ref{cor: NIJ on t states}.

Now we move on to $\ell \geq 1$. Evaluating gives
\begin{align}
    N_{IJ} L_{IJ}^{\dagger \, \ell} T_A^{\dagger \, t} \ket{0} 
    &=  N_{IJ} L_{IJ}^\dagger L_{IJ}^{\dagger \, \ell-1 } T_A^{\dagger \, t} \ket{0} \\
    &=  [ N_{IJ} , L_{IJ}^\dagger ] L_{IJ}^{\dagger \, \ell-1 } T_A^{\dagger \, t} \ket{0} ,
\end{align}
where we have used the inductive hypothesis for $\ell-1$.
Using the commutator identity $[ N_{IJ} , L_{IJ}^\dagger ] \simeq - F_J L_{IJ}^\dagger N_{IJ}$, we have
\begin{align}
    N_{IJ} L_{IJ}^{\dagger \, \ell} T_A^{\dagger \, t} \ket{0} 
    &= - F_J L_{IJ}^\dagger N_{IJ} L_{IJ}^{\dagger \, \ell-1 } T_A^{\dagger \, t} \ket{0} \\
    &= 0,
\end{align}
where we have again used the inductive hypothesis for $\ell-1$.
This concludes the proof.
\end{proof}

\begin{theorem}
    The norm-squared of states $\kett{\ell_{IJ}, t}$ in \eqref{eq: lij and t states} is
\begin{align}
    \label{eq: lij t state norm squared}
     \bbrakett{\ell_{IJ} , t | \ell_{IJ}, t } &= \tfrac{1}{2} (\ell_{IJ}+|t|+2) \ell_{IJ}! \, |t|! \, (\ell_{IJ}+|t|+1)! \, (|t|+1)! \nonumber \\
    &= \tfrac{1}{2} (\ell_{IJ}+|t|+2) ! \, (|t|+1)! \, \ell_{IJ}! \, |t|! \, . 
\end{align}
\end{theorem}

\begin{proof}
We prove the formula by induction on $\ell_{IJ} $.

For the base case $\ell_{IJ}=0$, the squared norm evaluates as
\begin{align}
    \bbrakett{0,t|0,t} &= \bbrakett{t|t} \\
    &= \tfrac{1}{2} (|t|+2)! \, (|t|+1)! \, |t|!
\end{align}
This agrees with the right-hand side of \eqref{eq: lij t state norm squared} when $\ell_{IJ} \to 0$.

We now proceed to the inductive step, taking $\ell_{IJ} \geq 1$ and assuming that \eqref{eq: lij t state norm squared} holds for $\kett{\ell_{IJ}-1,t}$.
The norm evaluation begins with
\begin{align}
    \bbrakett{\ell_{IJ} , t | \ell_{IJ}, t }
    &= \bbra{\ell_{IJ} - 1 , t} L_{IJ} L_{IJ}^\dagger \kett{\ell_{IJ} - 1 , t} \\
    &= \bbra{\ell_{IJ} - 1 , t} L_{IJ} L_{IJ}^{\dagger \, \ell_{IJ}} \kett{ 0, t} \\
    &= \bbra{\ell_{IJ} - 1 , t} \bigl[ L_{IJ} , L_{IJ}^{\dagger \, \ell_{IJ}} \bigr] \kett{ 0, t} \\
    &= \bbra{\ell_{IJ} - 1 , t} \sum_{k=0}^{\ell_{IJ}-1} L_{IJ}^{\dagger \, \ell_{IJ}-1-k} \bigl[ L_{IJ} , L_{IJ}^{\dagger} \bigr] L_{IJ}^{\dagger \, k } \kett{ 0, t}.
\end{align}
To proceed, one needs the commutator identity
\begin{align}
    \bigl[ L_{IJ} , L_{IJ}^{\dagger} \bigr] &\simeq 3 + P_I + Q_J - F_I F_J L_{JI}^\dagger L_{JI} - F_I M_{IJ} M_{JI} - F_J N_{JI} N_{IJ} .
\end{align}
We will simplify the expression for $\bigl[ L_{IJ} , L_{IJ}^{\dagger} \bigr]$ by considering its action on the ket $L_{IJ}^{\dagger \, k } \kett{ 0, t} $:
\begin{itemize}
    \item $ P_I + Q_J $ evaluates to $ |t| + 2 k $.
    \item $L_{JI} $ evaluates to zero because $L_{IJ}^{\dagger \, k } \kett{ 0, t}$ is a $P_J=0$ or $Q_I=0$ eigenstate (depending on the sign of $t$), in which case the $A(J)$ or $B(I)$ annihilation components must evaluate to zero, respectively. Therefore, we can drop $F_I F_J L_{JI}^\dagger L_{JI} $ from the commutator.
    \item Considering the two terms $F_I M_{IJ} M_{JI}$ and $F_J N_{JI} N_{IJ}$, at least one of these is necessarily zero, depending on the value of $t$;
    without loss of generality, take $t \geq 0$.
    Then $F_I M_{IJ} M_{JI}$ must evaluate to zero due to the lack of any $B(I)^\dagger$ components to the right of $M_{JI}$.
    
    We also have that $N_{IJ} L_{IJ}^{\dagger \, k } \kett{ 0, t} =0$ by \ref{lem: Nij on l and t states}, allowing us to drop $F_J N_{JI} N_{IJ}$ as well.
\end{itemize}
Thus,
\begin{align}
    \bbrakett{\ell_{IJ} , t | \ell_{IJ}, t }
    &= \bbra{\ell_{IJ} - 1 , t} \sum_{k=0}^{\ell_{IJ}-1} L_{IJ}^{\dagger \, \ell_{IJ}-1-k} \bigl[ L_{IJ} , L_{IJ}^{\dagger} \bigr] L_{IJ}^{\dagger \, k } \kett{ 0, t} \\
    &= \bbra{\ell_{IJ} - 1 , t} \sum_{k=0}^{\ell_{IJ}-1} L_{IJ}^{\dagger \, \ell_{IJ}-1-k} ( 3 + |t| + 2k - 0 - 0 - 0) L_{IJ}^{\dagger \, k } \kett{ 0, t} \\
    &= \bbrakett{ \ell_{IJ} - 1 , t |  \ell_{IJ}-1 , t } \sum_{k=0}^{\ell_{IJ}-1} ( 3 + |t| + 2k ) \\
    &= \bbrakett{ \ell_{IJ} - 1 , t |  \ell_{IJ}-1 , t } (\ell_{IJ}+|t|+2) \ell_{IJ} \\
    &= \tfrac{1}{2} (\ell_{IJ}+|t|+1) ! \, (|t|+1)! \, (\ell_{IJ}-1)! \, |t|! (\ell_{IJ}+|t|+2) \ell_{IJ} \\
    &= \tfrac{1}{2} (\ell_{IJ}+|t|+2) ! \, (|t|+1)! \, \ell_{IJ}! \, |t|!.
\end{align}
This confirms \eqref{eq: lij t state norm squared}.
\end{proof}
As one adds in more and more excitations, more and more gauge-singlet contracted operators appear throughout the evaluation and increasingly cannot be replaced with zero and must be ``fully evaluated''.
However, the general approach to proving the norm does otherwise carry over as in the above two examples.

\section{\label{app: extra results}Some results in four-dimensional subspaces}

Here we present overlap matrices for a few four-dimensional subspaces, which provide sufficient information to construct an orthogonal basis for them.

\noindent $\vecpq=(3,3,3,3,3,3)$:\\
\begin{align}
    &\left( \begin{array}{c}
    \bbra{0\,0\,0;3\,3\,3;0 }  \\
    \bbra{1\,1\,1;2\,2\,2;0 }  \\
    \bbra{2\,2\,2;1\,1\,1;0 }  \\
    \bbra{3\,3\,3;0\,0\,0;0 }  \\
    \end{array} \right)
    \left( \begin{array}{cccc}
    \kett{0\,0\,0;3\,3\,3;0 } , &
    \kett{1\,1\,1;2\,2\,2;0 } , &
    \kett{2\,2\,2;1\,1\,1;0 } , &
    \kett{3\,3\,3;0\,0\,0;0 }
    \end{array} \right)
    \nonumber \\
    &= \left(
    \begin{array}{cccc}
     \frac{2630651904}{245} & -\frac{468467712}{245} & \frac{211673088}{245} & -\frac{131383296}{245} \\
     -\frac{468467712}{245} & \frac{503156736}{245} & -\frac{305086464}{245} & \frac{211673088}{245} \\
     \frac{211673088}{245} & -\frac{305086464}{245} & \frac{503156736}{245} & -\frac{468467712}{245} \\
     -\frac{131383296}{245} & \frac{211673088}{245} & -\frac{468467712}{245} & \frac{2630651904}{245} \\
    \end{array}
\right) .
\end{align}

\noindent $\vecpq=(4,3,3,4,3,3)$:\\
\begin{align}
    &\left( \begin{array}{c}
    \bbra{1\,0\,0;3\,3\,3;0 }  \\
    \bbra{2\,1\,1;2\,2\,2;0 }  \\
    \bbra{3\,2\,2;1\,1\,1;0 }  \\
    \bbra{4\,3\,3;0\,0\,0;0 }  \\
    \end{array} \right)
    \left( \begin{array}{cccc}
    \kett{1\,0\,0;3\,3\,3;0 } , &
    \kett{2\,1\,1;2\,2\,2;0 } , &
    \kett{3\,2\,2;1\,1\,1;0 } , &
    \kett{4\,3\,3;0\,0\,0;0 }
    \end{array} \right)
    \nonumber \\
    &=
\left(
\begin{array}{cccc}
 \frac{5165239104}{49} & -\frac{1415636352}{49} & \frac{731472768}{49} & -\frac{480370176}{49} \\
 -\frac{1415636352}{49} & \frac{1564344576}{49} & -\frac{1068297984}{49} & \frac{786060288}{49} \\
 \frac{731472768}{49} & -\frac{1068297984}{49} & \frac{1807189056}{49} & -\frac{1797748992}{49} \\
 -\frac{480370176}{49} & \frac{786060288}{49} & -\frac{1797748992}{49} & \frac{10890630144}{49}. \\
\end{array}
\right)
\end{align}

\noindent $\vecpq=(4,3,4,4,3,4)$:\\
\begin{align}
    &\left( \begin{array}{c}
    \bbra{1\,1\,0;3\,3\,3;0 }  \\ 
    \bbra{2\,2\,1;2\,2\,2;0 }  \\
    \bbra{3\,3\,2;1\,1\,1;0 }  \\
    \bbra{4\,4\,3;0\,0\,0;0 }  \\
    \end{array} \right)
    \left( \begin{array}{cccc}
    \kett{1\,1\,0;3\,3\,3;0 } , &
    \kett{2\,2\,1;2\,2\,2;0 } , &
    \kett{3\,3\,2;1\,1\,1;0 } , &
    \kett{4\,4\,3;0\,0\,0;0 }
    \end{array} \right)
    \nonumber \\
    &= \left(
\begin{array}{cccc}
 1055592000 & -\frac{2970432000}{7} & \frac{1743768000}{7} & -\frac{1213056000}{7} \\
 -\frac{2970432000}{7} & \frac{3382560000}{7} & -369360000 & \frac{2021760000}{7} \\
 \frac{1743768000}{7} & -369360000 & \frac{4502304000}{7} & -\frac{4805568000}{7} \\
 -\frac{1213056000}{7} & \frac{2021760000}{7} & -\frac{4805568000}{7} & \frac{31726080000}{7} \\
\end{array}
\right) .
\end{align}

\bibliography{part-1.bib}

\begin{thebibliography}{163}%
\makeatletter
\providecommand \@ifxundefined [1]{%
 \@ifx{#1\undefined}
}%
\providecommand \@ifnum [1]{%
 \ifnum #1\expandafter \@firstoftwo
 \else \expandafter \@secondoftwo
 \fi
}%
\providecommand \@ifx [1]{%
 \ifx #1\expandafter \@firstoftwo
 \else \expandafter \@secondoftwo
 \fi
}%
\providecommand \natexlab [1]{#1}%
\providecommand \enquote  [1]{``#1''}%
\providecommand \bibnamefont  [1]{#1}%
\providecommand \bibfnamefont [1]{#1}%
\providecommand \citenamefont [1]{#1}%
\providecommand \href@noop [0]{\@secondoftwo}%
\providecommand \href [0]{\begingroup \@sanitize@url \@href}%
\providecommand \@href[1]{\@@startlink{#1}\@@href}%
\providecommand \@@href[1]{\endgroup#1\@@endlink}%
\providecommand \@sanitize@url [0]{\catcode `\\12\catcode `\$12\catcode
  `\&12\catcode `\#12\catcode `\^12\catcode `\_12\catcode `\%12\relax}%
\providecommand \@@startlink[1]{}%
\providecommand \@@endlink[0]{}%
\providecommand \url  [0]{\begingroup\@sanitize@url \@url }%
\providecommand \@url [1]{\endgroup\@href {#1}{\urlprefix }}%
\providecommand \urlprefix  [0]{URL }%
\providecommand \Eprint [0]{\href }%
\providecommand \doibase [0]{http://dx.doi.org/}%
\providecommand \selectlanguage [0]{\@gobble}%
\providecommand \bibinfo  [0]{\@secondoftwo}%
\providecommand \bibfield  [0]{\@secondoftwo}%
\providecommand \translation [1]{[#1]}%
\providecommand \BibitemOpen [0]{}%
\providecommand \bibitemStop [0]{}%
\providecommand \bibitemNoStop [0]{.\EOS\space}%
\providecommand \EOS [0]{\spacefactor3000\relax}%
\providecommand \BibitemShut  [1]{\csname bibitem#1\endcsname}%
\let\auto@bib@innerbib\@empty
\bibitem [{\citenamefont {Gross}\ \emph {et~al.}(2023)\citenamefont {Gross}
  \emph {et~al.}}]{Gross:2022hyw}%
  \BibitemOpen
  \bibfield  {author} {\bibinfo {author} {\bibfnamefont {Franz}\ \bibnamefont
  {Gross}} \emph {et~al.},\ }\bibfield  {title} {\enquote {\bibinfo {title}
  {{50 Years of Quantum Chromodynamics}},}\ }\href {\doibase
  10.1140/epjc/s10052-023-11949-2} {\bibfield  {journal} {\bibinfo  {journal}
  {Eur. Phys. J. C}\ }\textbf {\bibinfo {volume} {83}},\ \bibinfo {pages}
  {1125} (\bibinfo {year} {2023})},\ \Eprint {http://arxiv.org/abs/2212.11107}
  {arXiv:2212.11107 [hep-ph]} \BibitemShut {NoStop}%
\bibitem [{\citenamefont {Bauer}\ \emph
  {et~al.}(2023{\natexlab{a}})\citenamefont {Bauer} \emph
  {et~al.}}]{Bauer:2022hpo}%
  \BibitemOpen
  \bibfield  {author} {\bibinfo {author} {\bibfnamefont {Christian~W.}\
  \bibnamefont {Bauer}} \emph {et~al.},\ }\bibfield  {title} {\enquote
  {\bibinfo {title} {{Quantum Simulation for High-Energy Physics}},}\ }\href
  {\doibase 10.1103/PRXQuantum.4.027001} {\bibfield  {journal} {\bibinfo
  {journal} {PRX Quantum}\ }\textbf {\bibinfo {volume} {4}},\ \bibinfo {pages}
  {027001} (\bibinfo {year} {2023}{\natexlab{a}})},\ \Eprint
  {http://arxiv.org/abs/2204.03381} {arXiv:2204.03381 [quant-ph]} \BibitemShut
  {NoStop}%
\bibitem [{\citenamefont {Di~Meglio}\ \emph {et~al.}(2024)\citenamefont
  {Di~Meglio} \emph {et~al.}}]{DiMeglio:2023nsa}%
  \BibitemOpen
  \bibfield  {author} {\bibinfo {author} {\bibfnamefont {Alberto}\ \bibnamefont
  {Di~Meglio}} \emph {et~al.},\ }\bibfield  {title} {\enquote {\bibinfo {title}
  {{Quantum Computing for High-Energy Physics: State of the Art and
  Challenges}},}\ }\href {\doibase 10.1103/PRXQuantum.5.037001} {\bibfield
  {journal} {\bibinfo  {journal} {PRX Quantum}\ }\textbf {\bibinfo {volume}
  {5}},\ \bibinfo {pages} {037001} (\bibinfo {year} {2024})},\ \Eprint
  {http://arxiv.org/abs/2307.03236} {arXiv:2307.03236 [quant-ph]} \BibitemShut
  {NoStop}%
\bibitem [{\citenamefont {Bauer}\ \emph
  {et~al.}(2023{\natexlab{b}})\citenamefont {Bauer}, \citenamefont {Davoudi},
  \citenamefont {Klco},\ and\ \citenamefont {Savage}}]{Bauer:2023qgm}%
  \BibitemOpen
  \bibfield  {author} {\bibinfo {author} {\bibfnamefont {Christian~W.}\
  \bibnamefont {Bauer}}, \bibinfo {author} {\bibfnamefont {Zohreh}\
  \bibnamefont {Davoudi}}, \bibinfo {author} {\bibfnamefont {Natalie}\
  \bibnamefont {Klco}}, \ and\ \bibinfo {author} {\bibfnamefont {Martin~J.}\
  \bibnamefont {Savage}},\ }\bibfield  {title} {\enquote {\bibinfo {title}
  {{Quantum simulation of fundamental particles and forces}},}\ }\href
  {\doibase 10.1038/s42254-023-00599-8} {\bibfield  {journal} {\bibinfo
  {journal} {Nature Rev. Phys.}\ }\textbf {\bibinfo {volume} {5}},\ \bibinfo
  {pages} {420--432} (\bibinfo {year} {2023}{\natexlab{b}})},\ \Eprint
  {http://arxiv.org/abs/2404.06298} {arXiv:2404.06298 [hep-ph]} \BibitemShut
  {NoStop}%
\bibitem [{\citenamefont {Ba\~nuls}\ and\ \citenamefont
  {Cichy}(2020)}]{Banuls:2019rao}%
  \BibitemOpen
  \bibfield  {author} {\bibinfo {author} {\bibfnamefont {Mari~Carmen}\
  \bibnamefont {Ba\~nuls}}\ and\ \bibinfo {author} {\bibfnamefont {Krzysztof}\
  \bibnamefont {Cichy}},\ }\bibfield  {title} {\enquote {\bibinfo {title}
  {{Review on Novel Methods for Lattice Gauge Theories}},}\ }\href {\doibase
  10.1088/1361-6633/ab6311} {\bibfield  {journal} {\bibinfo  {journal} {Rept.
  Prog. Phys.}\ }\textbf {\bibinfo {volume} {83}},\ \bibinfo {pages} {024401}
  (\bibinfo {year} {2020})},\ \Eprint {http://arxiv.org/abs/1910.00257}
  {arXiv:1910.00257 [hep-lat]} \BibitemShut {NoStop}%
\bibitem [{\citenamefont {Meurice}\ \emph {et~al.}(2022)\citenamefont
  {Meurice}, \citenamefont {Sakai},\ and\ \citenamefont
  {Unmuth-Yockey}}]{Meurice:2020pxc}%
  \BibitemOpen
  \bibfield  {author} {\bibinfo {author} {\bibfnamefont {Yannick}\ \bibnamefont
  {Meurice}}, \bibinfo {author} {\bibfnamefont {Ryo}\ \bibnamefont {Sakai}}, \
  and\ \bibinfo {author} {\bibfnamefont {Judah}\ \bibnamefont
  {Unmuth-Yockey}},\ }\bibfield  {title} {\enquote {\bibinfo {title} {{Tensor
  lattice field theory for renormalization and quantum computing}},}\ }\href
  {\doibase 10.1103/RevModPhys.94.025005} {\bibfield  {journal} {\bibinfo
  {journal} {Rev. Mod. Phys.}\ }\textbf {\bibinfo {volume} {94}},\ \bibinfo
  {pages} {025005} (\bibinfo {year} {2022})},\ \Eprint
  {http://arxiv.org/abs/2010.06539} {arXiv:2010.06539 [hep-lat]} \BibitemShut
  {NoStop}%
\bibitem [{\citenamefont {Dalmonte}\ and\ \citenamefont
  {Montangero}(2016)}]{Dalmonte:2016alw}%
  \BibitemOpen
  \bibfield  {author} {\bibinfo {author} {\bibfnamefont {M.}~\bibnamefont
  {Dalmonte}}\ and\ \bibinfo {author} {\bibfnamefont {S.}~\bibnamefont
  {Montangero}},\ }\bibfield  {title} {\enquote {\bibinfo {title} {{Lattice
  gauge theory simulations in the quantum information era}},}\ }\href {\doibase
  10.1080/00107514.2016.1151199} {\bibfield  {journal} {\bibinfo  {journal}
  {Contemp. Phys.}\ }\textbf {\bibinfo {volume} {57}},\ \bibinfo {pages}
  {388--412} (\bibinfo {year} {2016})},\ \Eprint
  {http://arxiv.org/abs/1602.03776} {arXiv:1602.03776 [cond-mat.quant-gas]}
  \BibitemShut {NoStop}%
\bibitem [{\citenamefont {Preskill}(2018)}]{Preskill:2018fag}%
  \BibitemOpen
  \bibfield  {author} {\bibinfo {author} {\bibfnamefont {John}\ \bibnamefont
  {Preskill}},\ }\bibfield  {title} {\enquote {\bibinfo {title} {{Simulating
  quantum field theory with a quantum computer}},}\ }\href {\doibase
  10.22323/1.334.0024} {\bibfield  {journal} {\bibinfo  {journal} {PoS}\
  }\textbf {\bibinfo {volume} {LATTICE2018}},\ \bibinfo {pages} {024} (\bibinfo
  {year} {2018})},\ \Eprint {http://arxiv.org/abs/1811.10085} {arXiv:1811.10085
  [hep-lat]} \BibitemShut {NoStop}%
\bibitem [{\citenamefont {Ba\~nuls}\ \emph {et~al.}(2020)\citenamefont
  {Ba\~nuls} \emph {et~al.}}]{Banuls:2019bmf}%
  \BibitemOpen
  \bibfield  {author} {\bibinfo {author} {\bibfnamefont {M.~C.}\ \bibnamefont
  {Ba\~nuls}} \emph {et~al.},\ }\bibfield  {title} {\enquote {\bibinfo {title}
  {{Simulating Lattice Gauge Theories within Quantum Technologies}},}\ }\href
  {\doibase 10.1140/epjd/e2020-100571-8} {\bibfield  {journal} {\bibinfo
  {journal} {Eur. Phys. J. D}\ }\textbf {\bibinfo {volume} {74}},\ \bibinfo
  {pages} {165} (\bibinfo {year} {2020})},\ \Eprint
  {http://arxiv.org/abs/1911.00003} {arXiv:1911.00003 [quant-ph]} \BibitemShut
  {NoStop}%
\bibitem [{\citenamefont {Banerjee}\ \emph {et~al.}(2012)\citenamefont
  {Banerjee}, \citenamefont {Dalmonte}, \citenamefont {Muller}, \citenamefont
  {Rico}, \citenamefont {Stebler}, \citenamefont {Wiese},\ and\ \citenamefont
  {Zoller}}]{Banerjee:2012pg}%
  \BibitemOpen
  \bibfield  {author} {\bibinfo {author} {\bibfnamefont {D.}~\bibnamefont
  {Banerjee}}, \bibinfo {author} {\bibfnamefont {M.}~\bibnamefont {Dalmonte}},
  \bibinfo {author} {\bibfnamefont {M.}~\bibnamefont {Muller}}, \bibinfo
  {author} {\bibfnamefont {E.}~\bibnamefont {Rico}}, \bibinfo {author}
  {\bibfnamefont {P.}~\bibnamefont {Stebler}}, \bibinfo {author} {\bibfnamefont
  {U.~J.}\ \bibnamefont {Wiese}}, \ and\ \bibinfo {author} {\bibfnamefont
  {P.}~\bibnamefont {Zoller}},\ }\bibfield  {title} {\enquote {\bibinfo {title}
  {{Atomic Quantum Simulation of Dynamical Gauge Fields coupled to Fermionic
  Matter: From String Breaking to Evolution after a Quench}},}\ }\href
  {\doibase 10.1103/PhysRevLett.109.175302} {\bibfield  {journal} {\bibinfo
  {journal} {Phys. Rev. Lett.}\ }\textbf {\bibinfo {volume} {109}},\ \bibinfo
  {pages} {175302} (\bibinfo {year} {2012})},\ \Eprint
  {http://arxiv.org/abs/1205.6366} {arXiv:1205.6366 [cond-mat.quant-gas]}
  \BibitemShut {NoStop}%
\bibitem [{\citenamefont {Banerjee}\ \emph {et~al.}(2013)\citenamefont
  {Banerjee}, \citenamefont {B\"ogli}, \citenamefont {Dalmonte}, \citenamefont
  {Rico}, \citenamefont {Stebler}, \citenamefont {Wiese},\ and\ \citenamefont
  {Zoller}}]{Banerjee:2012xg}%
  \BibitemOpen
  \bibfield  {author} {\bibinfo {author} {\bibfnamefont {D.}~\bibnamefont
  {Banerjee}}, \bibinfo {author} {\bibfnamefont {M.}~\bibnamefont {B\"ogli}},
  \bibinfo {author} {\bibfnamefont {M.}~\bibnamefont {Dalmonte}}, \bibinfo
  {author} {\bibfnamefont {E.}~\bibnamefont {Rico}}, \bibinfo {author}
  {\bibfnamefont {P.}~\bibnamefont {Stebler}}, \bibinfo {author} {\bibfnamefont
  {U.~J.}\ \bibnamefont {Wiese}}, \ and\ \bibinfo {author} {\bibfnamefont
  {P.}~\bibnamefont {Zoller}},\ }\bibfield  {title} {\enquote {\bibinfo {title}
  {{Atomic Quantum Simulation of U(N) and SU(N) Non-Abelian Lattice Gauge
  Theories}},}\ }\href {\doibase 10.1103/PhysRevLett.110.125303} {\bibfield
  {journal} {\bibinfo  {journal} {Phys. Rev. Lett.}\ }\textbf {\bibinfo
  {volume} {110}},\ \bibinfo {pages} {125303} (\bibinfo {year} {2013})},\
  \Eprint {http://arxiv.org/abs/1211.2242} {arXiv:1211.2242
  [cond-mat.quant-gas]} \BibitemShut {NoStop}%
\bibitem [{\citenamefont {Huffman}\ \emph {et~al.}(2022)\citenamefont
  {Huffman}, \citenamefont {Garc\'\i{}a~Vera},\ and\ \citenamefont
  {Banerjee}}]{Huffman:2021gsi}%
  \BibitemOpen
  \bibfield  {author} {\bibinfo {author} {\bibfnamefont {Emilie}\ \bibnamefont
  {Huffman}}, \bibinfo {author} {\bibfnamefont {Miguel}\ \bibnamefont
  {Garc\'\i{}a~Vera}}, \ and\ \bibinfo {author} {\bibfnamefont {Debasish}\
  \bibnamefont {Banerjee}},\ }\bibfield  {title} {\enquote {\bibinfo {title}
  {{Toward the real-time evolution of gauge-invariant $\mathbb Z_2$ and $U(1)$
  quantum link models on noisy intermediate-scale quantum hardware with error
  mitigation}},}\ }\href {\doibase 10.1103/PhysRevD.106.094502} {\bibfield
  {journal} {\bibinfo  {journal} {Phys. Rev. D}\ }\textbf {\bibinfo {volume}
  {106}},\ \bibinfo {pages} {094502} (\bibinfo {year} {2022})},\ \Eprint
  {http://arxiv.org/abs/2109.15065} {arXiv:2109.15065 [quant-ph]} \BibitemShut
  {NoStop}%
\bibitem [{\citenamefont {Klco}\ \emph {et~al.}(2018)\citenamefont {Klco},
  \citenamefont {Dumitrescu}, \citenamefont {McCaskey}, \citenamefont {Morris},
  \citenamefont {Pooser}, \citenamefont {Sanz}, \citenamefont {Solano},
  \citenamefont {Lougovski},\ and\ \citenamefont {Savage}}]{Klco:2018kyo}%
  \BibitemOpen
  \bibfield  {author} {\bibinfo {author} {\bibfnamefont {N.}~\bibnamefont
  {Klco}}, \bibinfo {author} {\bibfnamefont {E.~F.}\ \bibnamefont
  {Dumitrescu}}, \bibinfo {author} {\bibfnamefont {A.~J.}\ \bibnamefont
  {McCaskey}}, \bibinfo {author} {\bibfnamefont {T.~D.}\ \bibnamefont
  {Morris}}, \bibinfo {author} {\bibfnamefont {R.~C.}\ \bibnamefont {Pooser}},
  \bibinfo {author} {\bibfnamefont {M.}~\bibnamefont {Sanz}}, \bibinfo {author}
  {\bibfnamefont {E.}~\bibnamefont {Solano}}, \bibinfo {author} {\bibfnamefont
  {P.}~\bibnamefont {Lougovski}}, \ and\ \bibinfo {author} {\bibfnamefont
  {M.~J.}\ \bibnamefont {Savage}},\ }\bibfield  {title} {\enquote {\bibinfo
  {title} {{Quantum-classical computation of Schwinger model dynamics using
  quantum computers}},}\ }\href {\doibase 10.1103/PhysRevA.98.032331}
  {\bibfield  {journal} {\bibinfo  {journal} {Phys. Rev. A}\ }\textbf {\bibinfo
  {volume} {98}},\ \bibinfo {pages} {032331} (\bibinfo {year} {2018})},\
  \Eprint {http://arxiv.org/abs/1803.03326} {arXiv:1803.03326 [quant-ph]}
  \BibitemShut {NoStop}%
\bibitem [{\citenamefont {Klco}\ \emph {et~al.}(2020)\citenamefont {Klco},
  \citenamefont {Stryker},\ and\ \citenamefont {Savage}}]{Klco:2019evd}%
  \BibitemOpen
  \bibfield  {author} {\bibinfo {author} {\bibfnamefont {Natalie}\ \bibnamefont
  {Klco}}, \bibinfo {author} {\bibfnamefont {Jesse~R.}\ \bibnamefont
  {Stryker}}, \ and\ \bibinfo {author} {\bibfnamefont {Martin~J.}\ \bibnamefont
  {Savage}},\ }\bibfield  {title} {\enquote {\bibinfo {title} {{SU(2)
  non-Abelian gauge field theory in one dimension on digital quantum
  computers}},}\ }\href {\doibase 10.1103/PhysRevD.101.074512} {\bibfield
  {journal} {\bibinfo  {journal} {Phys. Rev. D}\ }\textbf {\bibinfo {volume}
  {101}},\ \bibinfo {pages} {074512} (\bibinfo {year} {2020})},\ \Eprint
  {http://arxiv.org/abs/1908.06935} {arXiv:1908.06935 [quant-ph]} \BibitemShut
  {NoStop}%
\bibitem [{\citenamefont {Ciavarella}\ \emph {et~al.}(2021)\citenamefont
  {Ciavarella}, \citenamefont {Klco},\ and\ \citenamefont
  {Savage}}]{Ciavarella:2021nmj}%
  \BibitemOpen
  \bibfield  {author} {\bibinfo {author} {\bibfnamefont {Anthony}\ \bibnamefont
  {Ciavarella}}, \bibinfo {author} {\bibfnamefont {Natalie}\ \bibnamefont
  {Klco}}, \ and\ \bibinfo {author} {\bibfnamefont {Martin~J.}\ \bibnamefont
  {Savage}},\ }\bibfield  {title} {\enquote {\bibinfo {title} {{Trailhead for
  quantum simulation of SU(3) Yang-Mills lattice gauge theory in the local
  multiplet basis}},}\ }\href {\doibase 10.1103/PhysRevD.103.094501} {\bibfield
   {journal} {\bibinfo  {journal} {Phys. Rev. D}\ }\textbf {\bibinfo {volume}
  {103}},\ \bibinfo {pages} {094501} (\bibinfo {year} {2021})},\ \Eprint
  {http://arxiv.org/abs/2101.10227} {arXiv:2101.10227 [quant-ph]} \BibitemShut
  {NoStop}%
\bibitem [{\citenamefont {Stetina}\ \emph {et~al.}(2022)\citenamefont
  {Stetina}, \citenamefont {Ciavarella}, \citenamefont {Li},\ and\
  \citenamefont {Wiebe}}]{Stetina:2020abi}%
  \BibitemOpen
  \bibfield  {author} {\bibinfo {author} {\bibfnamefont {Torin~F.}\
  \bibnamefont {Stetina}}, \bibinfo {author} {\bibfnamefont {Anthony}\
  \bibnamefont {Ciavarella}}, \bibinfo {author} {\bibfnamefont {Xiaosong}\
  \bibnamefont {Li}}, \ and\ \bibinfo {author} {\bibfnamefont {Nathan}\
  \bibnamefont {Wiebe}},\ }\bibfield  {title} {\enquote {\bibinfo {title}
  {{Simulating Effective QED on Quantum Computers}},}\ }\href {\doibase
  10.22331/q-2022-01-18-622} {\bibfield  {journal} {\bibinfo  {journal}
  {Quantum}\ }\textbf {\bibinfo {volume} {6}},\ \bibinfo {pages} {622}
  (\bibinfo {year} {2022})},\ \Eprint {http://arxiv.org/abs/2101.00111}
  {arXiv:2101.00111 [quant-ph]} \BibitemShut {NoStop}%
\bibitem [{\citenamefont {Ciavarella}\ and\ \citenamefont
  {Chernyshev}(2022)}]{Ciavarella:2021lel}%
  \BibitemOpen
  \bibfield  {author} {\bibinfo {author} {\bibfnamefont {Anthony~N.}\
  \bibnamefont {Ciavarella}}\ and\ \bibinfo {author} {\bibfnamefont {Ivan~A.}\
  \bibnamefont {Chernyshev}},\ }\bibfield  {title} {\enquote {\bibinfo {title}
  {{Preparation of the SU(3) lattice Yang-Mills vacuum with variational quantum
  methods}},}\ }\href {\doibase 10.1103/PhysRevD.105.074504} {\bibfield
  {journal} {\bibinfo  {journal} {Phys. Rev. D}\ }\textbf {\bibinfo {volume}
  {105}},\ \bibinfo {pages} {074504} (\bibinfo {year} {2022})},\ \Eprint
  {http://arxiv.org/abs/2112.09083} {arXiv:2112.09083 [quant-ph]} \BibitemShut
  {NoStop}%
\bibitem [{\citenamefont {Farrell}\ \emph
  {et~al.}(2023{\natexlab{a}})\citenamefont {Farrell}, \citenamefont
  {Chernyshev}, \citenamefont {Powell}, \citenamefont {Zemlevskiy},
  \citenamefont {Illa},\ and\ \citenamefont {Savage}}]{Farrell:2022wyt}%
  \BibitemOpen
  \bibfield  {author} {\bibinfo {author} {\bibfnamefont {Roland~C.}\
  \bibnamefont {Farrell}}, \bibinfo {author} {\bibfnamefont {Ivan~A.}\
  \bibnamefont {Chernyshev}}, \bibinfo {author} {\bibfnamefont {Sarah J.~M.}\
  \bibnamefont {Powell}}, \bibinfo {author} {\bibfnamefont {Nikita~A.}\
  \bibnamefont {Zemlevskiy}}, \bibinfo {author} {\bibfnamefont {Marc}\
  \bibnamefont {Illa}}, \ and\ \bibinfo {author} {\bibfnamefont {Martin~J.}\
  \bibnamefont {Savage}},\ }\bibfield  {title} {\enquote {\bibinfo {title}
  {{Preparations for quantum simulations of quantum chromodynamics in 1+1
  dimensions. I. Axial gauge}},}\ }\href {\doibase 10.1103/PhysRevD.107.054512}
  {\bibfield  {journal} {\bibinfo  {journal} {Phys. Rev. D}\ }\textbf {\bibinfo
  {volume} {107}},\ \bibinfo {pages} {054512} (\bibinfo {year}
  {2023}{\natexlab{a}})},\ \Eprint {http://arxiv.org/abs/2207.01731}
  {arXiv:2207.01731 [quant-ph]} \BibitemShut {NoStop}%
\bibitem [{\citenamefont {Turro}\ \emph {et~al.}(2024)\citenamefont {Turro},
  \citenamefont {Ciavarella},\ and\ \citenamefont {Yao}}]{Turro:2024pxu}%
  \BibitemOpen
  \bibfield  {author} {\bibinfo {author} {\bibfnamefont {Francesco}\
  \bibnamefont {Turro}}, \bibinfo {author} {\bibfnamefont {Anthony}\
  \bibnamefont {Ciavarella}}, \ and\ \bibinfo {author} {\bibfnamefont
  {Xiaojun}\ \bibnamefont {Yao}},\ }\bibfield  {title} {\enquote {\bibinfo
  {title} {{Classical and quantum computing of shear viscosity for (2+1)D SU(2)
  gauge theory}},}\ }\href {\doibase 10.1103/PhysRevD.109.114511} {\bibfield
  {journal} {\bibinfo  {journal} {Phys. Rev. D}\ }\textbf {\bibinfo {volume}
  {109}},\ \bibinfo {pages} {114511} (\bibinfo {year} {2024})},\ \Eprint
  {http://arxiv.org/abs/2402.04221} {arXiv:2402.04221 [hep-lat]} \BibitemShut
  {NoStop}%
\bibitem [{\citenamefont {Ciavarella}\ and\ \citenamefont
  {Bauer}(2024)}]{Ciavarella:2024fzw}%
  \BibitemOpen
  \bibfield  {author} {\bibinfo {author} {\bibfnamefont {Anthony~N.}\
  \bibnamefont {Ciavarella}}\ and\ \bibinfo {author} {\bibfnamefont
  {Christian~W.}\ \bibnamefont {Bauer}},\ }\bibfield  {title} {\enquote
  {\bibinfo {title} {{Quantum Simulation of SU(3) Lattice Yang-Mills Theory at
  Leading Order in Large-Nc Expansion}},}\ }\href {\doibase
  10.1103/PhysRevLett.133.111901} {\bibfield  {journal} {\bibinfo  {journal}
  {Phys. Rev. Lett.}\ }\textbf {\bibinfo {volume} {133}},\ \bibinfo {pages}
  {111901} (\bibinfo {year} {2024})},\ \Eprint
  {http://arxiv.org/abs/2402.10265} {arXiv:2402.10265 [hep-ph]} \BibitemShut
  {NoStop}%
\bibitem [{\citenamefont {Farrell}\ \emph
  {et~al.}(2023{\natexlab{b}})\citenamefont {Farrell}, \citenamefont
  {Chernyshev}, \citenamefont {Powell}, \citenamefont {Zemlevskiy},
  \citenamefont {Illa},\ and\ \citenamefont {Savage}}]{Farrell:2022vyh}%
  \BibitemOpen
  \bibfield  {author} {\bibinfo {author} {\bibfnamefont {Roland~C.}\
  \bibnamefont {Farrell}}, \bibinfo {author} {\bibfnamefont {Ivan~A.}\
  \bibnamefont {Chernyshev}}, \bibinfo {author} {\bibfnamefont {Sarah J.~M.}\
  \bibnamefont {Powell}}, \bibinfo {author} {\bibfnamefont {Nikita~A.}\
  \bibnamefont {Zemlevskiy}}, \bibinfo {author} {\bibfnamefont {Marc}\
  \bibnamefont {Illa}}, \ and\ \bibinfo {author} {\bibfnamefont {Martin~J.}\
  \bibnamefont {Savage}},\ }\bibfield  {title} {\enquote {\bibinfo {title}
  {{Preparations for quantum simulations of quantum chromodynamics in 1+1
  dimensions. II. Single-baryon \ensuremath{\beta}-decay in real time}},}\
  }\href {\doibase 10.1103/PhysRevD.107.054513} {\bibfield  {journal} {\bibinfo
   {journal} {Phys. Rev. D}\ }\textbf {\bibinfo {volume} {107}},\ \bibinfo
  {pages} {054513} (\bibinfo {year} {2023}{\natexlab{b}})},\ \Eprint
  {http://arxiv.org/abs/2209.10781} {arXiv:2209.10781 [quant-ph]} \BibitemShut
  {NoStop}%
\bibitem [{\citenamefont {Farrell}\ \emph
  {et~al.}(2024{\natexlab{a}})\citenamefont {Farrell}, \citenamefont {Illa},
  \citenamefont {Ciavarella},\ and\ \citenamefont {Savage}}]{Farrell:2024fit}%
  \BibitemOpen
  \bibfield  {author} {\bibinfo {author} {\bibfnamefont {Roland~C.}\
  \bibnamefont {Farrell}}, \bibinfo {author} {\bibfnamefont {Marc}\
  \bibnamefont {Illa}}, \bibinfo {author} {\bibfnamefont {Anthony~N.}\
  \bibnamefont {Ciavarella}}, \ and\ \bibinfo {author} {\bibfnamefont
  {Martin~J.}\ \bibnamefont {Savage}},\ }\bibfield  {title} {\enquote {\bibinfo
  {title} {{Quantum simulations of hadron dynamics in the Schwinger model using
  112 qubits}},}\ }\href {\doibase 10.1103/PhysRevD.109.114510} {\bibfield
  {journal} {\bibinfo  {journal} {Phys. Rev. D}\ }\textbf {\bibinfo {volume}
  {109}},\ \bibinfo {pages} {114510} (\bibinfo {year} {2024}{\natexlab{a}})},\
  \Eprint {http://arxiv.org/abs/2401.08044} {arXiv:2401.08044 [quant-ph]}
  \BibitemShut {NoStop}%
\bibitem [{\citenamefont {Farrell}\ \emph {et~al.}(2025)\citenamefont
  {Farrell}, \citenamefont {Illa},\ and\ \citenamefont
  {Savage}}]{Farrell:2024mgu}%
  \BibitemOpen
  \bibfield  {author} {\bibinfo {author} {\bibfnamefont {Roland~C.}\
  \bibnamefont {Farrell}}, \bibinfo {author} {\bibfnamefont {Marc}\
  \bibnamefont {Illa}}, \ and\ \bibinfo {author} {\bibfnamefont {Martin~J.}\
  \bibnamefont {Savage}},\ }\bibfield  {title} {\enquote {\bibinfo {title}
  {{Steps toward quantum simulations of hadronization and energy loss in dense
  matter}},}\ }\href {\doibase 10.1103/PhysRevC.111.015202} {\bibfield
  {journal} {\bibinfo  {journal} {Phys. Rev. C}\ }\textbf {\bibinfo {volume}
  {111}},\ \bibinfo {pages} {015202} (\bibinfo {year} {2025})},\ \Eprint
  {http://arxiv.org/abs/2405.06620} {arXiv:2405.06620 [quant-ph]} \BibitemShut
  {NoStop}%
\bibitem [{\citenamefont {Haase}\ \emph {et~al.}(2021)\citenamefont {Haase},
  \citenamefont {Dellantonio}, \citenamefont {Celi}, \citenamefont {Paulson},
  \citenamefont {Kan}, \citenamefont {Jansen},\ and\ \citenamefont
  {Muschik}}]{Haase:2020kaj}%
  \BibitemOpen
  \bibfield  {author} {\bibinfo {author} {\bibfnamefont {Jan~F.}\ \bibnamefont
  {Haase}}, \bibinfo {author} {\bibfnamefont {Luca}\ \bibnamefont
  {Dellantonio}}, \bibinfo {author} {\bibfnamefont {Alessio}\ \bibnamefont
  {Celi}}, \bibinfo {author} {\bibfnamefont {Danny}\ \bibnamefont {Paulson}},
  \bibinfo {author} {\bibfnamefont {Angus}\ \bibnamefont {Kan}}, \bibinfo
  {author} {\bibfnamefont {Karl}\ \bibnamefont {Jansen}}, \ and\ \bibinfo
  {author} {\bibfnamefont {Christine~A.}\ \bibnamefont {Muschik}},\ }\bibfield
  {title} {\enquote {\bibinfo {title} {{A resource efficient approach for
  quantum and classical simulations of gauge theories in particle physics}},}\
  }\href {\doibase 10.22331/q-2021-02-04-393} {\bibfield  {journal} {\bibinfo
  {journal} {Quantum}\ }\textbf {\bibinfo {volume} {5}},\ \bibinfo {pages}
  {393} (\bibinfo {year} {2021})},\ \Eprint {http://arxiv.org/abs/2006.14160}
  {arXiv:2006.14160 [quant-ph]} \BibitemShut {NoStop}%
\bibitem [{\citenamefont {Dasgupta}\ and\ \citenamefont
  {Raychowdhury}(2022)}]{Dasgupta:2020itb}%
  \BibitemOpen
  \bibfield  {author} {\bibinfo {author} {\bibfnamefont {Raka}\ \bibnamefont
  {Dasgupta}}\ and\ \bibinfo {author} {\bibfnamefont {Indrakshi}\ \bibnamefont
  {Raychowdhury}},\ }\bibfield  {title} {\enquote {\bibinfo {title} {{Cold-atom
  quantum simulator for string and hadron dynamics in non-Abelian lattice gauge
  theory}},}\ }\href {\doibase 10.1103/PhysRevA.105.023322} {\bibfield
  {journal} {\bibinfo  {journal} {Phys. Rev. A}\ }\textbf {\bibinfo {volume}
  {105}},\ \bibinfo {pages} {023322} (\bibinfo {year} {2022})},\ \Eprint
  {http://arxiv.org/abs/2009.13969} {arXiv:2009.13969 [hep-lat]} \BibitemShut
  {NoStop}%
\bibitem [{\citenamefont {Bennewitz}\ \emph {et~al.}(2024)\citenamefont
  {Bennewitz} \emph {et~al.}}]{Bennewitz:2024ixi}%
  \BibitemOpen
  \bibfield  {author} {\bibinfo {author} {\bibfnamefont {Elizabeth~R.}\
  \bibnamefont {Bennewitz}} \emph {et~al.},\ }\bibfield  {title} {\enquote
  {\bibinfo {title} {{Simulating Meson Scattering on Spin Quantum
  Simulators}},}\ }\href@noop {} {\  (\bibinfo {year} {2024})},\ \Eprint
  {http://arxiv.org/abs/2403.07061} {arXiv:2403.07061 [quant-ph]} \BibitemShut
  {NoStop}%
\bibitem [{\citenamefont {Davoudi}\ \emph {et~al.}(2024)\citenamefont
  {Davoudi}, \citenamefont {Hsieh},\ and\ \citenamefont
  {Kadam}}]{Davoudi:2024wyv}%
  \BibitemOpen
  \bibfield  {author} {\bibinfo {author} {\bibfnamefont {Zohreh}\ \bibnamefont
  {Davoudi}}, \bibinfo {author} {\bibfnamefont {Chung-Chun}\ \bibnamefont
  {Hsieh}}, \ and\ \bibinfo {author} {\bibfnamefont {Saurabh~V.}\ \bibnamefont
  {Kadam}},\ }\bibfield  {title} {\enquote {\bibinfo {title} {{Scattering wave
  packets of hadrons in gauge theories: Preparation on a quantum computer}},}\
  }\href {\doibase 10.22331/q-2024-11-11-1520} {\bibfield  {journal} {\bibinfo
  {journal} {Quantum}\ }\textbf {\bibinfo {volume} {8}},\ \bibinfo {pages}
  {1520} (\bibinfo {year} {2024})},\ \Eprint {http://arxiv.org/abs/2402.00840}
  {arXiv:2402.00840 [quant-ph]} \BibitemShut {NoStop}%
\bibitem [{\citenamefont {Belyansky}\ \emph {et~al.}(2024)\citenamefont
  {Belyansky}, \citenamefont {Whitsitt}, \citenamefont {Mueller}, \citenamefont
  {Fahimniya}, \citenamefont {Bennewitz}, \citenamefont {Davoudi},\ and\
  \citenamefont {Gorshkov}}]{Belyansky:2023rgh}%
  \BibitemOpen
  \bibfield  {author} {\bibinfo {author} {\bibfnamefont {Ron}\ \bibnamefont
  {Belyansky}}, \bibinfo {author} {\bibfnamefont {Seth}\ \bibnamefont
  {Whitsitt}}, \bibinfo {author} {\bibfnamefont {Niklas}\ \bibnamefont
  {Mueller}}, \bibinfo {author} {\bibfnamefont {Ali}\ \bibnamefont
  {Fahimniya}}, \bibinfo {author} {\bibfnamefont {Elizabeth~R.}\ \bibnamefont
  {Bennewitz}}, \bibinfo {author} {\bibfnamefont {Zohreh}\ \bibnamefont
  {Davoudi}}, \ and\ \bibinfo {author} {\bibfnamefont {Alexey~V.}\ \bibnamefont
  {Gorshkov}},\ }\bibfield  {title} {\enquote {\bibinfo {title} {{High-Energy
  Collision of Quarks and Mesons in the Schwinger Model: From Tensor Networks
  to Circuit QED}},}\ }\href {\doibase 10.1103/PhysRevLett.132.091903}
  {\bibfield  {journal} {\bibinfo  {journal} {Phys. Rev. Lett.}\ }\textbf
  {\bibinfo {volume} {132}},\ \bibinfo {pages} {091903} (\bibinfo {year}
  {2024})},\ \Eprint {http://arxiv.org/abs/2307.02522} {arXiv:2307.02522
  [quant-ph]} \BibitemShut {NoStop}%
\bibitem [{\citenamefont {Davoudi}\ \emph
  {et~al.}(2023{\natexlab{a}})\citenamefont {Davoudi}, \citenamefont {Shaw},\
  and\ \citenamefont {Stryker}}]{Davoudi:2022xmb}%
  \BibitemOpen
  \bibfield  {author} {\bibinfo {author} {\bibfnamefont {Zohreh}\ \bibnamefont
  {Davoudi}}, \bibinfo {author} {\bibfnamefont {Alexander~F.}\ \bibnamefont
  {Shaw}}, \ and\ \bibinfo {author} {\bibfnamefont {Jesse~R.}\ \bibnamefont
  {Stryker}},\ }\bibfield  {title} {\enquote {\bibinfo {title} {{General
  quantum algorithms for Hamiltonian simulation with applications to a
  non-Abelian lattice gauge theory}},}\ }\href {\doibase
  10.22331/q-2023-12-20-1213} {\bibfield  {journal} {\bibinfo  {journal}
  {Quantum}\ }\textbf {\bibinfo {volume} {7}},\ \bibinfo {pages} {1213}
  (\bibinfo {year} {2023}{\natexlab{a}})},\ \Eprint
  {http://arxiv.org/abs/2212.14030} {arXiv:2212.14030 [hep-lat]} \BibitemShut
  {NoStop}%
\bibitem [{\citenamefont {Davoudi}\ \emph {et~al.}(2020)\citenamefont
  {Davoudi}, \citenamefont {Hafezi}, \citenamefont {Monroe}, \citenamefont
  {Pagano}, \citenamefont {Seif},\ and\ \citenamefont
  {Shaw}}]{Davoudi:2019bhy}%
  \BibitemOpen
  \bibfield  {author} {\bibinfo {author} {\bibfnamefont {Zohreh}\ \bibnamefont
  {Davoudi}}, \bibinfo {author} {\bibfnamefont {Mohammad}\ \bibnamefont
  {Hafezi}}, \bibinfo {author} {\bibfnamefont {Christopher}\ \bibnamefont
  {Monroe}}, \bibinfo {author} {\bibfnamefont {Guido}\ \bibnamefont {Pagano}},
  \bibinfo {author} {\bibfnamefont {Alireza}\ \bibnamefont {Seif}}, \ and\
  \bibinfo {author} {\bibfnamefont {Andrew}\ \bibnamefont {Shaw}},\ }\bibfield
  {title} {\enquote {\bibinfo {title} {{Towards analog quantum simulations of
  lattice gauge theories with trapped ions}},}\ }\href {\doibase
  10.1103/PhysRevResearch.2.023015} {\bibfield  {journal} {\bibinfo  {journal}
  {Phys. Rev. Res.}\ }\textbf {\bibinfo {volume} {2}},\ \bibinfo {pages}
  {023015} (\bibinfo {year} {2020})},\ \Eprint
  {http://arxiv.org/abs/1908.03210} {arXiv:1908.03210 [quant-ph]} \BibitemShut
  {NoStop}%
\bibitem [{\citenamefont {Mueller}\ \emph {et~al.}(2023)\citenamefont
  {Mueller}, \citenamefont {Carolan}, \citenamefont {Connelly}, \citenamefont
  {Davoudi}, \citenamefont {Dumitrescu},\ and\ \citenamefont
  {Yeter-Aydeniz}}]{Mueller:2022xbg}%
  \BibitemOpen
  \bibfield  {author} {\bibinfo {author} {\bibfnamefont {Niklas}\ \bibnamefont
  {Mueller}}, \bibinfo {author} {\bibfnamefont {Joseph~A.}\ \bibnamefont
  {Carolan}}, \bibinfo {author} {\bibfnamefont {Andrew}\ \bibnamefont
  {Connelly}}, \bibinfo {author} {\bibfnamefont {Zohreh}\ \bibnamefont
  {Davoudi}}, \bibinfo {author} {\bibfnamefont {Eugene~F.}\ \bibnamefont
  {Dumitrescu}}, \ and\ \bibinfo {author} {\bibfnamefont {K\"ubra}\
  \bibnamefont {Yeter-Aydeniz}},\ }\bibfield  {title} {\enquote {\bibinfo
  {title} {{Quantum Computation of Dynamical Quantum Phase Transitions and
  Entanglement Tomography in a Lattice Gauge Theory}},}\ }\href {\doibase
  10.1103/PRXQuantum.4.030323} {\bibfield  {journal} {\bibinfo  {journal} {PRX
  Quantum}\ }\textbf {\bibinfo {volume} {4}},\ \bibinfo {pages} {030323}
  (\bibinfo {year} {2023})},\ \Eprint {http://arxiv.org/abs/2210.03089}
  {arXiv:2210.03089 [quant-ph]} \BibitemShut {NoStop}%
\bibitem [{\citenamefont {Davoudi}\ \emph
  {et~al.}(2023{\natexlab{b}})\citenamefont {Davoudi}, \citenamefont
  {Mueller},\ and\ \citenamefont {Powers}}]{Davoudi:2022uzo}%
  \BibitemOpen
  \bibfield  {author} {\bibinfo {author} {\bibfnamefont {Zohreh}\ \bibnamefont
  {Davoudi}}, \bibinfo {author} {\bibfnamefont {Niklas}\ \bibnamefont
  {Mueller}}, \ and\ \bibinfo {author} {\bibfnamefont {Connor}\ \bibnamefont
  {Powers}},\ }\bibfield  {title} {\enquote {\bibinfo {title} {{Towards Quantum
  Computing Phase Diagrams of Gauge Theories with Thermal Pure Quantum
  States}},}\ }\href {\doibase 10.1103/PhysRevLett.131.081901} {\bibfield
  {journal} {\bibinfo  {journal} {Phys. Rev. Lett.}\ }\textbf {\bibinfo
  {volume} {131}},\ \bibinfo {pages} {081901} (\bibinfo {year}
  {2023}{\natexlab{b}})},\ \Eprint {http://arxiv.org/abs/2208.13112}
  {arXiv:2208.13112 [hep-lat]} \BibitemShut {NoStop}%
\bibitem [{\citenamefont {Nguyen}\ \emph {et~al.}(2022)\citenamefont {Nguyen},
  \citenamefont {Tran}, \citenamefont {Zhu}, \citenamefont {Green},
  \citenamefont {Alderete}, \citenamefont {Davoudi},\ and\ \citenamefont
  {Linke}}]{Nguyen:2021hyk}%
  \BibitemOpen
  \bibfield  {author} {\bibinfo {author} {\bibfnamefont {Nhung~H.}\
  \bibnamefont {Nguyen}}, \bibinfo {author} {\bibfnamefont {Minh~C.}\
  \bibnamefont {Tran}}, \bibinfo {author} {\bibfnamefont {Yingyue}\
  \bibnamefont {Zhu}}, \bibinfo {author} {\bibfnamefont {Alaina~M.}\
  \bibnamefont {Green}}, \bibinfo {author} {\bibfnamefont {C.~Huerta}\
  \bibnamefont {Alderete}}, \bibinfo {author} {\bibfnamefont {Zohreh}\
  \bibnamefont {Davoudi}}, \ and\ \bibinfo {author} {\bibfnamefont
  {Norbert~M.}\ \bibnamefont {Linke}},\ }\bibfield  {title} {\enquote {\bibinfo
  {title} {{Digital Quantum Simulation of the Schwinger Model and Symmetry
  Protection with Trapped Ions}},}\ }\href {\doibase
  10.1103/PRXQuantum.3.020324} {\bibfield  {journal} {\bibinfo  {journal} {PRX
  Quantum}\ }\textbf {\bibinfo {volume} {3}},\ \bibinfo {pages} {020324}
  (\bibinfo {year} {2022})},\ \Eprint {http://arxiv.org/abs/2112.14262}
  {arXiv:2112.14262 [quant-ph]} \BibitemShut {NoStop}%
\bibitem [{\citenamefont {Atas}\ \emph {et~al.}(2023)\citenamefont {Atas},
  \citenamefont {Haase}, \citenamefont {Zhang}, \citenamefont {Wei},
  \citenamefont {Pfaendler}, \citenamefont {Lewis},\ and\ \citenamefont
  {Muschik}}]{Atas:2022dqm}%
  \BibitemOpen
  \bibfield  {author} {\bibinfo {author} {\bibfnamefont {Yasar~Y.}\
  \bibnamefont {Atas}}, \bibinfo {author} {\bibfnamefont {Jan~F.}\ \bibnamefont
  {Haase}}, \bibinfo {author} {\bibfnamefont {Jinglei}\ \bibnamefont {Zhang}},
  \bibinfo {author} {\bibfnamefont {Victor}\ \bibnamefont {Wei}}, \bibinfo
  {author} {\bibfnamefont {Sieglinde M.~L.}\ \bibnamefont {Pfaendler}},
  \bibinfo {author} {\bibfnamefont {Randy}\ \bibnamefont {Lewis}}, \ and\
  \bibinfo {author} {\bibfnamefont {Christine~A.}\ \bibnamefont {Muschik}},\
  }\bibfield  {title} {\enquote {\bibinfo {title} {{Simulating one-dimensional
  quantum chromodynamics on a quantum computer: Real-time evolutions of tetra-
  and pentaquarks}},}\ }\href {\doibase 10.1103/PhysRevResearch.5.033184}
  {\bibfield  {journal} {\bibinfo  {journal} {Phys. Rev. Res.}\ }\textbf
  {\bibinfo {volume} {5}},\ \bibinfo {pages} {033184} (\bibinfo {year}
  {2023})},\ \Eprint {http://arxiv.org/abs/2207.03473} {arXiv:2207.03473
  [quant-ph]} \BibitemShut {NoStop}%
\bibitem [{\citenamefont {Atas}\ \emph {et~al.}(2021)\citenamefont {Atas},
  \citenamefont {Zhang}, \citenamefont {Lewis}, \citenamefont {Jahanpour},
  \citenamefont {Haase},\ and\ \citenamefont {Muschik}}]{Atas:2021ext}%
  \BibitemOpen
  \bibfield  {author} {\bibinfo {author} {\bibfnamefont {Yasar~Y.}\
  \bibnamefont {Atas}}, \bibinfo {author} {\bibfnamefont {Jinglei}\
  \bibnamefont {Zhang}}, \bibinfo {author} {\bibfnamefont {Randy}\ \bibnamefont
  {Lewis}}, \bibinfo {author} {\bibfnamefont {Amin}\ \bibnamefont {Jahanpour}},
  \bibinfo {author} {\bibfnamefont {Jan~F.}\ \bibnamefont {Haase}}, \ and\
  \bibinfo {author} {\bibfnamefont {Christine~A.}\ \bibnamefont {Muschik}},\
  }\bibfield  {title} {\enquote {\bibinfo {title} {{SU(2) hadrons on a quantum
  computer via a variational approach}},}\ }\href {\doibase
  10.1038/s41467-021-26825-4} {\bibfield  {journal} {\bibinfo  {journal}
  {Nature Commun.}\ }\textbf {\bibinfo {volume} {12}},\ \bibinfo {pages} {6499}
  (\bibinfo {year} {2021})},\ \Eprint {http://arxiv.org/abs/2102.08920}
  {arXiv:2102.08920 [quant-ph]} \BibitemShut {NoStop}%
\bibitem [{\citenamefont {Martinez}\ \emph {et~al.}(2016)\citenamefont
  {Martinez} \emph {et~al.}}]{Martinez:2016yna}%
  \BibitemOpen
  \bibfield  {author} {\bibinfo {author} {\bibfnamefont {E.~A.}\ \bibnamefont
  {Martinez}} \emph {et~al.},\ }\bibfield  {title} {\enquote {\bibinfo {title}
  {{Real-time dynamics of lattice gauge theories with a few-qubit quantum
  computer}},}\ }\href {\doibase 10.1038/nature18318} {\bibfield  {journal}
  {\bibinfo  {journal} {Nature}\ }\textbf {\bibinfo {volume} {534}},\ \bibinfo
  {pages} {516--519} (\bibinfo {year} {2016})},\ \Eprint
  {http://arxiv.org/abs/1605.04570} {arXiv:1605.04570 [quant-ph]} \BibitemShut
  {NoStop}%
\bibitem [{\citenamefont {Kasper}\ \emph {et~al.}(2017)\citenamefont {Kasper},
  \citenamefont {Hebenstreit}, \citenamefont {Jendrzejewski}, \citenamefont
  {Oberthaler},\ and\ \citenamefont {Berges}}]{Kasper:2016mzj}%
  \BibitemOpen
  \bibfield  {author} {\bibinfo {author} {\bibfnamefont {V.}~\bibnamefont
  {Kasper}}, \bibinfo {author} {\bibfnamefont {F.}~\bibnamefont {Hebenstreit}},
  \bibinfo {author} {\bibfnamefont {F.}~\bibnamefont {Jendrzejewski}}, \bibinfo
  {author} {\bibfnamefont {M.~K.}\ \bibnamefont {Oberthaler}}, \ and\ \bibinfo
  {author} {\bibfnamefont {J.}~\bibnamefont {Berges}},\ }\bibfield  {title}
  {\enquote {\bibinfo {title} {{Implementing quantum electrodynamics with
  ultracold atomic systems}},}\ }\href {\doibase 10.1088/1367-2630/aa54e0}
  {\bibfield  {journal} {\bibinfo  {journal} {New J. Phys.}\ }\textbf {\bibinfo
  {volume} {19}},\ \bibinfo {pages} {023030} (\bibinfo {year} {2017})},\
  \Eprint {http://arxiv.org/abs/1608.03480} {arXiv:1608.03480
  [cond-mat.quant-gas]} \BibitemShut {NoStop}%
\bibitem [{\citenamefont {Kane}\ \emph {et~al.}(2022)\citenamefont {Kane},
  \citenamefont {Grabowska}, \citenamefont {Nachman},\ and\ \citenamefont
  {Bauer}}]{Kane:2022ejm}%
  \BibitemOpen
  \bibfield  {author} {\bibinfo {author} {\bibfnamefont {Christopher}\
  \bibnamefont {Kane}}, \bibinfo {author} {\bibfnamefont {Dorota~M.}\
  \bibnamefont {Grabowska}}, \bibinfo {author} {\bibfnamefont {Benjamin}\
  \bibnamefont {Nachman}}, \ and\ \bibinfo {author} {\bibfnamefont
  {Christian~W.}\ \bibnamefont {Bauer}},\ }\bibfield  {title} {\enquote
  {\bibinfo {title} {{Efficient quantum implementation of 2+1 U(1) lattice
  gauge theories with Gauss law constraints}},}\ }\href@noop {} {\  (\bibinfo
  {year} {2022})},\ \Eprint {http://arxiv.org/abs/2211.10497} {arXiv:2211.10497
  [quant-ph]} \BibitemShut {NoStop}%
\bibitem [{\citenamefont {Mil}\ \emph {et~al.}(2020)\citenamefont {Mil},
  \citenamefont {Zache}, \citenamefont {Hegde}, \citenamefont {Xia},
  \citenamefont {Bhatt}, \citenamefont {Oberthaler}, \citenamefont {Hauke},
  \citenamefont {Berges},\ and\ \citenamefont {Jendrzejewski}}]{Mil:2019pbt}%
  \BibitemOpen
  \bibfield  {author} {\bibinfo {author} {\bibfnamefont {Alexander}\
  \bibnamefont {Mil}}, \bibinfo {author} {\bibfnamefont {Torsten~V.}\
  \bibnamefont {Zache}}, \bibinfo {author} {\bibfnamefont {Apoorva}\
  \bibnamefont {Hegde}}, \bibinfo {author} {\bibfnamefont {Andy}\ \bibnamefont
  {Xia}}, \bibinfo {author} {\bibfnamefont {Rohit~P.}\ \bibnamefont {Bhatt}},
  \bibinfo {author} {\bibfnamefont {Markus~K.}\ \bibnamefont {Oberthaler}},
  \bibinfo {author} {\bibfnamefont {Philipp}\ \bibnamefont {Hauke}}, \bibinfo
  {author} {\bibfnamefont {J\"urgen}\ \bibnamefont {Berges}}, \ and\ \bibinfo
  {author} {\bibfnamefont {Fred}\ \bibnamefont {Jendrzejewski}},\ }\bibfield
  {title} {\enquote {\bibinfo {title} {{A scalable realization of local U(1)
  gauge invariance in cold atomic mixtures}},}\ }\href {\doibase
  10.1126/science.aaz5312} {\bibfield  {journal} {\bibinfo  {journal}
  {Science}\ }\textbf {\bibinfo {volume} {367}},\ \bibinfo {pages} {1128--1130}
  (\bibinfo {year} {2020})},\ \Eprint {http://arxiv.org/abs/1909.07641}
  {arXiv:1909.07641 [cond-mat.quant-gas]} \BibitemShut {NoStop}%
\bibitem [{\citenamefont {Muschik}\ \emph {et~al.}(2017)\citenamefont
  {Muschik}, \citenamefont {Heyl}, \citenamefont {Martinez}, \citenamefont
  {Monz}, \citenamefont {Schindler}, \citenamefont {Vogell}, \citenamefont
  {Dalmonte}, \citenamefont {Hauke}, \citenamefont {Blatt},\ and\ \citenamefont
  {Zoller}}]{Muschik:2016tws}%
  \BibitemOpen
  \bibfield  {author} {\bibinfo {author} {\bibfnamefont {Christine}\
  \bibnamefont {Muschik}}, \bibinfo {author} {\bibfnamefont {Markus}\
  \bibnamefont {Heyl}}, \bibinfo {author} {\bibfnamefont {Esteban}\
  \bibnamefont {Martinez}}, \bibinfo {author} {\bibfnamefont {Thomas}\
  \bibnamefont {Monz}}, \bibinfo {author} {\bibfnamefont {Philipp}\
  \bibnamefont {Schindler}}, \bibinfo {author} {\bibfnamefont {Berit}\
  \bibnamefont {Vogell}}, \bibinfo {author} {\bibfnamefont {Marcello}\
  \bibnamefont {Dalmonte}}, \bibinfo {author} {\bibfnamefont {Philipp}\
  \bibnamefont {Hauke}}, \bibinfo {author} {\bibfnamefont {Rainer}\
  \bibnamefont {Blatt}}, \ and\ \bibinfo {author} {\bibfnamefont {Peter}\
  \bibnamefont {Zoller}},\ }\bibfield  {title} {\enquote {\bibinfo {title}
  {{U(1) Wilson lattice gauge theories in digital quantum simulators}},}\
  }\href {\doibase 10.1088/1367-2630/aa89ab} {\bibfield  {journal} {\bibinfo
  {journal} {New J. Phys.}\ }\textbf {\bibinfo {volume} {19}},\ \bibinfo
  {pages} {103020} (\bibinfo {year} {2017})},\ \Eprint
  {http://arxiv.org/abs/1612.08653} {arXiv:1612.08653 [quant-ph]} \BibitemShut
  {NoStop}%
\bibitem [{\citenamefont {Meth}\ \emph {et~al.}(2023)\citenamefont {Meth} \emph
  {et~al.}}]{Meth:2023wzd}%
  \BibitemOpen
  \bibfield  {author} {\bibinfo {author} {\bibfnamefont {Michael}\ \bibnamefont
  {Meth}} \emph {et~al.},\ }\bibfield  {title} {\enquote {\bibinfo {title}
  {{Simulating 2D lattice gauge theories on a qudit quantum computer}},}\
  }\href@noop {} {\  (\bibinfo {year} {2023})},\ \Eprint
  {http://arxiv.org/abs/2310.12110} {arXiv:2310.12110 [quant-ph]} \BibitemShut
  {NoStop}%
\bibitem [{\citenamefont {Zhang}\ \emph {et~al.}(2023)\citenamefont {Zhang},
  \citenamefont {Ferguson}, \citenamefont {K\"uhn}, \citenamefont {Haase},
  \citenamefont {Wilson}, \citenamefont {Jansen},\ and\ \citenamefont
  {Muschik}}]{Zhang:2021bjq}%
  \BibitemOpen
  \bibfield  {author} {\bibinfo {author} {\bibfnamefont {Jinglei}\ \bibnamefont
  {Zhang}}, \bibinfo {author} {\bibfnamefont {Ryan}\ \bibnamefont {Ferguson}},
  \bibinfo {author} {\bibfnamefont {Stefan}\ \bibnamefont {K\"uhn}}, \bibinfo
  {author} {\bibfnamefont {Jan~F.}\ \bibnamefont {Haase}}, \bibinfo {author}
  {\bibfnamefont {C.~M.}\ \bibnamefont {Wilson}}, \bibinfo {author}
  {\bibfnamefont {Karl}\ \bibnamefont {Jansen}}, \ and\ \bibinfo {author}
  {\bibfnamefont {Christine~A.}\ \bibnamefont {Muschik}},\ }\bibfield  {title}
  {\enquote {\bibinfo {title} {{Simulating gauge theories with variational
  quantum eigensolvers in superconducting microwave cavities}},}\ }\href
  {\doibase 10.22331/q-2023-10-23-1148} {\bibfield  {journal} {\bibinfo
  {journal} {Quantum}\ }\textbf {\bibinfo {volume} {7}},\ \bibinfo {pages}
  {1148} (\bibinfo {year} {2023})},\ \Eprint {http://arxiv.org/abs/2108.08248}
  {arXiv:2108.08248 [quant-ph]} \BibitemShut {NoStop}%
\bibitem [{\citenamefont {Paulson}\ \emph {et~al.}(2021)\citenamefont {Paulson}
  \emph {et~al.}}]{Paulson:2020zjd}%
  \BibitemOpen
  \bibfield  {author} {\bibinfo {author} {\bibfnamefont {Danny}\ \bibnamefont
  {Paulson}} \emph {et~al.},\ }\bibfield  {title} {\enquote {\bibinfo {title}
  {{Simulating 2D Effects in Lattice Gauge Theories on a Quantum Computer}},}\
  }\href {\doibase 10.1103/PRXQuantum.2.030334} {\bibfield  {journal} {\bibinfo
   {journal} {PRX Quantum}\ }\textbf {\bibinfo {volume} {2}},\ \bibinfo {pages}
  {030334} (\bibinfo {year} {2021})},\ \Eprint
  {http://arxiv.org/abs/2008.09252} {arXiv:2008.09252 [quant-ph]} \BibitemShut
  {NoStop}%
\bibitem [{\citenamefont {Guo}\ \emph {et~al.}(2024)\citenamefont {Guo},
  \citenamefont {Angelides}, \citenamefont {Jansen},\ and\ \citenamefont
  {K\"uhn}}]{Guo:2024tnb}%
  \BibitemOpen
  \bibfield  {author} {\bibinfo {author} {\bibfnamefont {Yibin}\ \bibnamefont
  {Guo}}, \bibinfo {author} {\bibfnamefont {Takis}\ \bibnamefont {Angelides}},
  \bibinfo {author} {\bibfnamefont {Karl}\ \bibnamefont {Jansen}}, \ and\
  \bibinfo {author} {\bibfnamefont {Stefan}\ \bibnamefont {K\"uhn}},\
  }\bibfield  {title} {\enquote {\bibinfo {title} {{Concurrent VQE for
  Simulating Excited States of the Schwinger Model}},}\ }\href@noop {} {\
  (\bibinfo {year} {2024})},\ \Eprint {http://arxiv.org/abs/2407.15629}
  {arXiv:2407.15629 [quant-ph]} \BibitemShut {NoStop}%
\bibitem [{\citenamefont {Crippa}\ \emph {et~al.}(2024)\citenamefont {Crippa},
  \citenamefont {Romiti}, \citenamefont {Funcke}, \citenamefont {Jansen},
  \citenamefont {K\"uhn}, \citenamefont {Stornati},\ and\ \citenamefont
  {Urbach}}]{Crippa:2024cqr}%
  \BibitemOpen
  \bibfield  {author} {\bibinfo {author} {\bibfnamefont {Arianna}\ \bibnamefont
  {Crippa}}, \bibinfo {author} {\bibfnamefont {Simone}\ \bibnamefont {Romiti}},
  \bibinfo {author} {\bibfnamefont {Lena}\ \bibnamefont {Funcke}}, \bibinfo
  {author} {\bibfnamefont {Karl}\ \bibnamefont {Jansen}}, \bibinfo {author}
  {\bibfnamefont {Stefan}\ \bibnamefont {K\"uhn}}, \bibinfo {author}
  {\bibfnamefont {Paolo}\ \bibnamefont {Stornati}}, \ and\ \bibinfo {author}
  {\bibfnamefont {Carsten}\ \bibnamefont {Urbach}},\ }\bibfield  {title}
  {\enquote {\bibinfo {title} {{Towards determining the (2+1)-dimensional
  Quantum Electrodynamics running coupling with Monte Carlo and quantum
  computing methods}},}\ }\href@noop {} {\  (\bibinfo {year} {2024})},\ \Eprint
  {http://arxiv.org/abs/2404.17545} {arXiv:2404.17545 [hep-lat]} \BibitemShut
  {NoStop}%
\bibitem [{\citenamefont {Angelides}\ \emph {et~al.}(2025)\citenamefont
  {Angelides}, \citenamefont {Naredi}, \citenamefont {Crippa}, \citenamefont
  {Jansen}, \citenamefont {K\"uhn}, \citenamefont {Tavernelli},\ and\
  \citenamefont {Wang}}]{Angelides:2023noe}%
  \BibitemOpen
  \bibfield  {author} {\bibinfo {author} {\bibfnamefont {Takis}\ \bibnamefont
  {Angelides}}, \bibinfo {author} {\bibfnamefont {Pranay}\ \bibnamefont
  {Naredi}}, \bibinfo {author} {\bibfnamefont {Arianna}\ \bibnamefont
  {Crippa}}, \bibinfo {author} {\bibfnamefont {Karl}\ \bibnamefont {Jansen}},
  \bibinfo {author} {\bibfnamefont {Stefan}\ \bibnamefont {K\"uhn}}, \bibinfo
  {author} {\bibfnamefont {Ivano}\ \bibnamefont {Tavernelli}}, \ and\ \bibinfo
  {author} {\bibfnamefont {Derek~S.}\ \bibnamefont {Wang}},\ }\bibfield
  {title} {\enquote {\bibinfo {title} {{First-order phase transition of the
  Schwinger model with a quantum computer}},}\ }\href {\doibase
  10.1038/s41534-024-00950-6} {\bibfield  {journal} {\bibinfo  {journal} {npj
  Quantum Inf.}\ }\textbf {\bibinfo {volume} {11}},\ \bibinfo {pages} {6}
  (\bibinfo {year} {2025})},\ \Eprint {http://arxiv.org/abs/2312.12831}
  {arXiv:2312.12831 [hep-lat]} \BibitemShut {NoStop}%
\bibitem [{\citenamefont {Chakraborty}\ \emph {et~al.}(2022)\citenamefont
  {Chakraborty}, \citenamefont {Honda}, \citenamefont {Izubuchi}, \citenamefont
  {Kikuchi},\ and\ \citenamefont {Tomiya}}]{Chakraborty:2020uhf}%
  \BibitemOpen
  \bibfield  {author} {\bibinfo {author} {\bibfnamefont {Bipasha}\ \bibnamefont
  {Chakraborty}}, \bibinfo {author} {\bibfnamefont {Masazumi}\ \bibnamefont
  {Honda}}, \bibinfo {author} {\bibfnamefont {Taku}\ \bibnamefont {Izubuchi}},
  \bibinfo {author} {\bibfnamefont {Yuta}\ \bibnamefont {Kikuchi}}, \ and\
  \bibinfo {author} {\bibfnamefont {Akio}\ \bibnamefont {Tomiya}},\ }\bibfield
  {title} {\enquote {\bibinfo {title} {{Classically emulated digital quantum
  simulation of the Schwinger model with a topological term via adiabatic state
  preparation}},}\ }\href {\doibase 10.1103/PhysRevD.105.094503} {\bibfield
  {journal} {\bibinfo  {journal} {Phys. Rev. D}\ }\textbf {\bibinfo {volume}
  {105}},\ \bibinfo {pages} {094503} (\bibinfo {year} {2022})},\ \Eprint
  {http://arxiv.org/abs/2001.00485} {arXiv:2001.00485 [hep-lat]} \BibitemShut
  {NoStop}%
\bibitem [{\citenamefont {Halimeh}\ \emph {et~al.}(2025)\citenamefont
  {Halimeh}, \citenamefont {Aidelsburger}, \citenamefont {Grusdt},
  \citenamefont {Hauke},\ and\ \citenamefont {Yang}}]{Halimeh:2023lid}%
  \BibitemOpen
  \bibfield  {author} {\bibinfo {author} {\bibfnamefont {Jad~C.}\ \bibnamefont
  {Halimeh}}, \bibinfo {author} {\bibfnamefont {Monika}\ \bibnamefont
  {Aidelsburger}}, \bibinfo {author} {\bibfnamefont {Fabian}\ \bibnamefont
  {Grusdt}}, \bibinfo {author} {\bibfnamefont {Philipp}\ \bibnamefont {Hauke}},
  \ and\ \bibinfo {author} {\bibfnamefont {Bing}\ \bibnamefont {Yang}},\
  }\bibfield  {title} {\enquote {\bibinfo {title} {{Cold-atom quantum
  simulators of gauge theories}},}\ }\href {\doibase
  10.1038/s41567-024-02721-8} {\bibfield  {journal} {\bibinfo  {journal}
  {Nature Phys.}\ }\textbf {\bibinfo {volume} {21}},\ \bibinfo {pages} {25--36}
  (\bibinfo {year} {2025})},\ \Eprint {http://arxiv.org/abs/2310.12201}
  {arXiv:2310.12201 [cond-mat.quant-gas]} \BibitemShut {NoStop}%
\bibitem [{\citenamefont {de~Jong}\ \emph {et~al.}(2022)\citenamefont
  {de~Jong}, \citenamefont {Lee}, \citenamefont {Mulligan}, \citenamefont
  {P\l{}osko\'n}, \citenamefont {Ringer},\ and\ \citenamefont
  {Yao}}]{deJong:2021wsd}%
  \BibitemOpen
  \bibfield  {author} {\bibinfo {author} {\bibfnamefont {Wibe~A.}\ \bibnamefont
  {de~Jong}}, \bibinfo {author} {\bibfnamefont {Kyle}\ \bibnamefont {Lee}},
  \bibinfo {author} {\bibfnamefont {James}\ \bibnamefont {Mulligan}}, \bibinfo
  {author} {\bibfnamefont {Mateusz}\ \bibnamefont {P\l{}osko\'n}}, \bibinfo
  {author} {\bibfnamefont {Felix}\ \bibnamefont {Ringer}}, \ and\ \bibinfo
  {author} {\bibfnamefont {Xiaojun}\ \bibnamefont {Yao}},\ }\bibfield  {title}
  {\enquote {\bibinfo {title} {{Quantum simulation of nonequilibrium dynamics
  and thermalization in the Schwinger model}},}\ }\href {\doibase
  10.1103/PhysRevD.106.054508} {\bibfield  {journal} {\bibinfo  {journal}
  {Phys. Rev. D}\ }\textbf {\bibinfo {volume} {106}},\ \bibinfo {pages}
  {054508} (\bibinfo {year} {2022})},\ \Eprint
  {http://arxiv.org/abs/2106.08394} {arXiv:2106.08394 [quant-ph]} \BibitemShut
  {NoStop}%
\bibitem [{\citenamefont {Lamm}\ \emph {et~al.}(2019)\citenamefont {Lamm},
  \citenamefont {Lawrence},\ and\ \citenamefont {Yamauchi}}]{Lamm:2019bik}%
  \BibitemOpen
  \bibfield  {author} {\bibinfo {author} {\bibfnamefont {Henry}\ \bibnamefont
  {Lamm}}, \bibinfo {author} {\bibfnamefont {Scott}\ \bibnamefont {Lawrence}},
  \ and\ \bibinfo {author} {\bibfnamefont {Yukari}\ \bibnamefont {Yamauchi}}
  (\bibinfo {collaboration} {NuQS}),\ }\bibfield  {title} {\enquote {\bibinfo
  {title} {{General Methods for Digital Quantum Simulation of Gauge
  Theories}},}\ }\href {\doibase 10.1103/PhysRevD.100.034518} {\bibfield
  {journal} {\bibinfo  {journal} {Phys. Rev. D}\ }\textbf {\bibinfo {volume}
  {100}},\ \bibinfo {pages} {034518} (\bibinfo {year} {2019})},\ \Eprint
  {http://arxiv.org/abs/1903.08807} {arXiv:1903.08807 [hep-lat]} \BibitemShut
  {NoStop}%
\bibitem [{\citenamefont {Shaw}\ \emph {et~al.}(2020)\citenamefont {Shaw},
  \citenamefont {Lougovski}, \citenamefont {Stryker},\ and\ \citenamefont
  {Wiebe}}]{Shaw:2020udc}%
  \BibitemOpen
  \bibfield  {author} {\bibinfo {author} {\bibfnamefont {Alexander~F.}\
  \bibnamefont {Shaw}}, \bibinfo {author} {\bibfnamefont {Pavel}\ \bibnamefont
  {Lougovski}}, \bibinfo {author} {\bibfnamefont {Jesse~R.}\ \bibnamefont
  {Stryker}}, \ and\ \bibinfo {author} {\bibfnamefont {Nathan}\ \bibnamefont
  {Wiebe}},\ }\bibfield  {title} {\enquote {\bibinfo {title} {{Quantum
  Algorithms for Simulating the Lattice Schwinger Model}},}\ }\href {\doibase
  10.22331/q-2020-08-10-306} {\bibfield  {journal} {\bibinfo  {journal}
  {Quantum}\ }\textbf {\bibinfo {volume} {4}},\ \bibinfo {pages} {306}
  (\bibinfo {year} {2020})},\ \Eprint {http://arxiv.org/abs/2002.11146}
  {arXiv:2002.11146 [quant-ph]} \BibitemShut {NoStop}%
\bibitem [{\citenamefont {Kan}\ and\ \citenamefont {Nam}(2021)}]{Kan:2021xfc}%
  \BibitemOpen
  \bibfield  {author} {\bibinfo {author} {\bibfnamefont {Angus}\ \bibnamefont
  {Kan}}\ and\ \bibinfo {author} {\bibfnamefont {Yunseong}\ \bibnamefont
  {Nam}},\ }\bibfield  {title} {\enquote {\bibinfo {title} {{Lattice Quantum
  Chromodynamics and Electrodynamics on a Universal Quantum Computer}},}\
  }\href@noop {} {\  (\bibinfo {year} {2021})},\ \Eprint
  {http://arxiv.org/abs/2107.12769} {arXiv:2107.12769 [quant-ph]} \BibitemShut
  {NoStop}%
\bibitem [{\citenamefont {Cohen}\ \emph {et~al.}(2021)\citenamefont {Cohen},
  \citenamefont {Lamm}, \citenamefont {Lawrence},\ and\ \citenamefont
  {Yamauchi}}]{Cohen:2021imf}%
  \BibitemOpen
  \bibfield  {author} {\bibinfo {author} {\bibfnamefont {Thomas~D.}\
  \bibnamefont {Cohen}}, \bibinfo {author} {\bibfnamefont {Henry}\ \bibnamefont
  {Lamm}}, \bibinfo {author} {\bibfnamefont {Scott}\ \bibnamefont {Lawrence}},
  \ and\ \bibinfo {author} {\bibfnamefont {Yukari}\ \bibnamefont {Yamauchi}}
  (\bibinfo {collaboration} {NuQS}),\ }\bibfield  {title} {\enquote {\bibinfo
  {title} {{Quantum algorithms for transport coefficients in gauge
  theories}},}\ }\href {\doibase 10.1103/PhysRevD.104.094514} {\bibfield
  {journal} {\bibinfo  {journal} {Phys. Rev. D}\ }\textbf {\bibinfo {volume}
  {104}},\ \bibinfo {pages} {094514} (\bibinfo {year} {2021})},\ \Eprint
  {http://arxiv.org/abs/2104.02024} {arXiv:2104.02024 [hep-lat]} \BibitemShut
  {NoStop}%
\bibitem [{\citenamefont {Murairi}\ \emph {et~al.}(2022)\citenamefont
  {Murairi}, \citenamefont {Cervia}, \citenamefont {Kumar}, \citenamefont
  {Bedaque},\ and\ \citenamefont {Alexandru}}]{Murairi:2022zdg}%
  \BibitemOpen
  \bibfield  {author} {\bibinfo {author} {\bibfnamefont {Edison~M.}\
  \bibnamefont {Murairi}}, \bibinfo {author} {\bibfnamefont {Michael~J.}\
  \bibnamefont {Cervia}}, \bibinfo {author} {\bibfnamefont {Hersh}\
  \bibnamefont {Kumar}}, \bibinfo {author} {\bibfnamefont {Paulo~F.}\
  \bibnamefont {Bedaque}}, \ and\ \bibinfo {author} {\bibfnamefont {Andrei}\
  \bibnamefont {Alexandru}},\ }\bibfield  {title} {\enquote {\bibinfo {title}
  {{How many quantum gates do gauge theories require?}}}\ }\href {\doibase
  10.1103/PhysRevD.106.094504} {\bibfield  {journal} {\bibinfo  {journal}
  {Phys. Rev. D}\ }\textbf {\bibinfo {volume} {106}},\ \bibinfo {pages}
  {094504} (\bibinfo {year} {2022})},\ \Eprint
  {http://arxiv.org/abs/2208.11789} {arXiv:2208.11789 [hep-lat]} \BibitemShut
  {NoStop}%
\bibitem [{\citenamefont {Charles}\ \emph {et~al.}(2024)\citenamefont
  {Charles}, \citenamefont {Gustafson}, \citenamefont {Hardt}, \citenamefont
  {Herren}, \citenamefont {Hogan}, \citenamefont {Lamm}, \citenamefont
  {Starecheski}, \citenamefont {Van~de Water},\ and\ \citenamefont
  {Wagman}}]{Charles:2023zbl}%
  \BibitemOpen
  \bibfield  {author} {\bibinfo {author} {\bibfnamefont {Clement}\ \bibnamefont
  {Charles}}, \bibinfo {author} {\bibfnamefont {Erik~J.}\ \bibnamefont
  {Gustafson}}, \bibinfo {author} {\bibfnamefont {Elizabeth}\ \bibnamefont
  {Hardt}}, \bibinfo {author} {\bibfnamefont {Florian}\ \bibnamefont {Herren}},
  \bibinfo {author} {\bibfnamefont {Norman}\ \bibnamefont {Hogan}}, \bibinfo
  {author} {\bibfnamefont {Henry}\ \bibnamefont {Lamm}}, \bibinfo {author}
  {\bibfnamefont {Sara}\ \bibnamefont {Starecheski}}, \bibinfo {author}
  {\bibfnamefont {Ruth~S.}\ \bibnamefont {Van~de Water}}, \ and\ \bibinfo
  {author} {\bibfnamefont {Michael~L.}\ \bibnamefont {Wagman}},\ }\bibfield
  {title} {\enquote {\bibinfo {title} {{Simulating Z2 lattice gauge theory on a
  quantum computer}},}\ }\href {\doibase 10.1103/PhysRevE.109.015307}
  {\bibfield  {journal} {\bibinfo  {journal} {Phys. Rev. E}\ }\textbf {\bibinfo
  {volume} {109}},\ \bibinfo {pages} {015307} (\bibinfo {year} {2024})},\
  \Eprint {http://arxiv.org/abs/2305.02361} {arXiv:2305.02361 [hep-lat]}
  \BibitemShut {NoStop}%
\bibitem [{\citenamefont {Mildenberger}\ \emph {et~al.}(2025)\citenamefont
  {Mildenberger}, \citenamefont {Mruczkiewicz}, \citenamefont {Halimeh},
  \citenamefont {Jiang},\ and\ \citenamefont {Hauke}}]{Mildenberger:2022jqr}%
  \BibitemOpen
  \bibfield  {author} {\bibinfo {author} {\bibfnamefont {Julius}\ \bibnamefont
  {Mildenberger}}, \bibinfo {author} {\bibfnamefont {Wojciech}\ \bibnamefont
  {Mruczkiewicz}}, \bibinfo {author} {\bibfnamefont {Jad~C.}\ \bibnamefont
  {Halimeh}}, \bibinfo {author} {\bibfnamefont {Zhang}\ \bibnamefont {Jiang}},
  \ and\ \bibinfo {author} {\bibfnamefont {Philipp}\ \bibnamefont {Hauke}},\
  }\bibfield  {title} {\enquote {\bibinfo {title} {{Confinement in a
  ${{\mathbb{Z}}}_{2}$ lattice gauge theory on a quantum computer}},}\ }\href
  {\doibase 10.1038/s41567-024-02723-6} {\bibfield  {journal} {\bibinfo
  {journal} {Nature Phys.}\ }\textbf {\bibinfo {volume} {21}},\ \bibinfo
  {pages} {312--317} (\bibinfo {year} {2025})},\ \Eprint
  {http://arxiv.org/abs/2203.08905} {arXiv:2203.08905 [quant-ph]} \BibitemShut
  {NoStop}%
\bibitem [{\citenamefont {Halimeh}\ \emph {et~al.}(2022)\citenamefont
  {Halimeh}, \citenamefont {McCulloch}, \citenamefont {Yang},\ and\
  \citenamefont {Hauke}}]{Halimeh:2022pkw}%
  \BibitemOpen
  \bibfield  {author} {\bibinfo {author} {\bibfnamefont {Jad~C.}\ \bibnamefont
  {Halimeh}}, \bibinfo {author} {\bibfnamefont {Ian~P.}\ \bibnamefont
  {McCulloch}}, \bibinfo {author} {\bibfnamefont {Bing}\ \bibnamefont {Yang}},
  \ and\ \bibinfo {author} {\bibfnamefont {Philipp}\ \bibnamefont {Hauke}},\
  }\bibfield  {title} {\enquote {\bibinfo {title} {{Tuning the Topological
  \ensuremath{\theta}-Angle in Cold-Atom Quantum Simulators of Gauge
  Theories}},}\ }\href {\doibase 10.1103/PRXQuantum.3.040316} {\bibfield
  {journal} {\bibinfo  {journal} {PRX Quantum}\ }\textbf {\bibinfo {volume}
  {3}},\ \bibinfo {pages} {040316} (\bibinfo {year} {2022})},\ \Eprint
  {http://arxiv.org/abs/2204.06570} {arXiv:2204.06570 [cond-mat.quant-gas]}
  \BibitemShut {NoStop}%
\bibitem [{\citenamefont {Zhang}\ \emph {et~al.}(2025)\citenamefont {Zhang}
  \emph {et~al.}}]{Zhang:2023hzr}%
  \BibitemOpen
  \bibfield  {author} {\bibinfo {author} {\bibfnamefont {Wei-Yong}\
  \bibnamefont {Zhang}} \emph {et~al.},\ }\bibfield  {title} {\enquote
  {\bibinfo {title} {{Observation of microscopic confinement dynamics by a
  tunable topological \ensuremath{\theta}-angle}},}\ }\href {\doibase
  10.1038/s41567-024-02702-x} {\bibfield  {journal} {\bibinfo  {journal}
  {Nature Phys.}\ }\textbf {\bibinfo {volume} {21}},\ \bibinfo {pages}
  {155--160} (\bibinfo {year} {2025})},\ \Eprint
  {http://arxiv.org/abs/2306.11794} {arXiv:2306.11794 [cond-mat.quant-gas]}
  \BibitemShut {NoStop}%
\bibitem [{\citenamefont {Kane}\ \emph {et~al.}(2024)\citenamefont {Kane},
  \citenamefont {Gomes},\ and\ \citenamefont {Kreshchuk}}]{Kane:2023jdo}%
  \BibitemOpen
  \bibfield  {author} {\bibinfo {author} {\bibfnamefont {Christopher~F.}\
  \bibnamefont {Kane}}, \bibinfo {author} {\bibfnamefont {Niladri}\
  \bibnamefont {Gomes}}, \ and\ \bibinfo {author} {\bibfnamefont {Michael}\
  \bibnamefont {Kreshchuk}},\ }\bibfield  {title} {\enquote {\bibinfo {title}
  {{Nearly optimal state preparation for quantum simulations of lattice gauge
  theories}},}\ }\href {\doibase 10.1103/PhysRevA.110.012455} {\bibfield
  {journal} {\bibinfo  {journal} {Phys. Rev. A}\ }\textbf {\bibinfo {volume}
  {110}},\ \bibinfo {pages} {012455} (\bibinfo {year} {2024})},\ \Eprint
  {http://arxiv.org/abs/2310.13757} {arXiv:2310.13757 [quant-ph]} \BibitemShut
  {NoStop}%
\bibitem [{\citenamefont {Yang}\ \emph {et~al.}(2020)\citenamefont {Yang},
  \citenamefont {Sun}, \citenamefont {Ott}, \citenamefont {Wang}, \citenamefont
  {Zache}, \citenamefont {Halimeh}, \citenamefont {Yuan}, \citenamefont
  {Hauke},\ and\ \citenamefont {Pan}}]{Yang:2020yer}%
  \BibitemOpen
  \bibfield  {author} {\bibinfo {author} {\bibfnamefont {Bing}\ \bibnamefont
  {Yang}}, \bibinfo {author} {\bibfnamefont {Hui}\ \bibnamefont {Sun}},
  \bibinfo {author} {\bibfnamefont {Robert}\ \bibnamefont {Ott}}, \bibinfo
  {author} {\bibfnamefont {Han-Yi}\ \bibnamefont {Wang}}, \bibinfo {author}
  {\bibfnamefont {Torsten~V.}\ \bibnamefont {Zache}}, \bibinfo {author}
  {\bibfnamefont {Jad~C.}\ \bibnamefont {Halimeh}}, \bibinfo {author}
  {\bibfnamefont {Zhen-Sheng}\ \bibnamefont {Yuan}}, \bibinfo {author}
  {\bibfnamefont {Philipp}\ \bibnamefont {Hauke}}, \ and\ \bibinfo {author}
  {\bibfnamefont {Jian-Wei}\ \bibnamefont {Pan}},\ }\bibfield  {title}
  {\enquote {\bibinfo {title} {{Observation of gauge invariance in a 71-site
  Bose\textendash{}Hubbard quantum simulator}},}\ }\href {\doibase
  10.1038/s41586-020-2910-8} {\bibfield  {journal} {\bibinfo  {journal}
  {Nature}\ }\textbf {\bibinfo {volume} {587}},\ \bibinfo {pages} {392--396}
  (\bibinfo {year} {2020})},\ \Eprint {http://arxiv.org/abs/2003.08945}
  {arXiv:2003.08945 [cond-mat.quant-gas]} \BibitemShut {NoStop}%
\bibitem [{\citenamefont {Zohar}\ \emph
  {et~al.}(2013{\natexlab{a}})\citenamefont {Zohar}, \citenamefont {Cirac},\
  and\ \citenamefont {Reznik}}]{Zohar:2012ts}%
  \BibitemOpen
  \bibfield  {author} {\bibinfo {author} {\bibfnamefont {Erez}\ \bibnamefont
  {Zohar}}, \bibinfo {author} {\bibfnamefont {J.~Ignacio}\ \bibnamefont
  {Cirac}}, \ and\ \bibinfo {author} {\bibfnamefont {Benni}\ \bibnamefont
  {Reznik}},\ }\bibfield  {title} {\enquote {\bibinfo {title} {{Simulating
  (2+1)-Dimensional Lattice QED with Dynamical Matter Using Ultracold
  Atoms}},}\ }\href {\doibase 10.1103/PhysRevLett.110.055302} {\bibfield
  {journal} {\bibinfo  {journal} {Phys. Rev. Lett.}\ }\textbf {\bibinfo
  {volume} {110}},\ \bibinfo {pages} {055302} (\bibinfo {year}
  {2013}{\natexlab{a}})},\ \Eprint {http://arxiv.org/abs/1208.4299}
  {arXiv:1208.4299 [quant-ph]} \BibitemShut {NoStop}%
\bibitem [{\citenamefont {Zohar}\ \emph
  {et~al.}(2013{\natexlab{b}})\citenamefont {Zohar}, \citenamefont {Cirac},\
  and\ \citenamefont {Reznik}}]{Zohar:2012xf}%
  \BibitemOpen
  \bibfield  {author} {\bibinfo {author} {\bibfnamefont {Erez}\ \bibnamefont
  {Zohar}}, \bibinfo {author} {\bibfnamefont {J.~Ignacio}\ \bibnamefont
  {Cirac}}, \ and\ \bibinfo {author} {\bibfnamefont {Benni}\ \bibnamefont
  {Reznik}},\ }\bibfield  {title} {\enquote {\bibinfo {title} {{Cold-Atom
  Quantum Simulator for SU(2) Yang-Mills Lattice Gauge Theory}},}\ }\href
  {\doibase 10.1103/PhysRevLett.110.125304} {\bibfield  {journal} {\bibinfo
  {journal} {Phys. Rev. Lett.}\ }\textbf {\bibinfo {volume} {110}},\ \bibinfo
  {pages} {125304} (\bibinfo {year} {2013}{\natexlab{b}})},\ \Eprint
  {http://arxiv.org/abs/1211.2241} {arXiv:1211.2241 [quant-ph]} \BibitemShut
  {NoStop}%
\bibitem [{\citenamefont {Hauke}\ \emph {et~al.}(2013)\citenamefont {Hauke},
  \citenamefont {Marcos}, \citenamefont {Dalmonte},\ and\ \citenamefont
  {Zoller}}]{Hauke:2013jga}%
  \BibitemOpen
  \bibfield  {author} {\bibinfo {author} {\bibfnamefont {Philipp}\ \bibnamefont
  {Hauke}}, \bibinfo {author} {\bibfnamefont {David}\ \bibnamefont {Marcos}},
  \bibinfo {author} {\bibfnamefont {Marcello}\ \bibnamefont {Dalmonte}}, \ and\
  \bibinfo {author} {\bibfnamefont {Peter}\ \bibnamefont {Zoller}},\ }\bibfield
   {title} {\enquote {\bibinfo {title} {{Quantum simulation of a lattice
  Schwinger model in a chain of trapped ions}},}\ }\href {\doibase
  10.1103/PhysRevX.3.041018} {\bibfield  {journal} {\bibinfo  {journal} {Phys.
  Rev. X}\ }\textbf {\bibinfo {volume} {3}},\ \bibinfo {pages} {041018}
  (\bibinfo {year} {2013})},\ \Eprint {http://arxiv.org/abs/1306.2162}
  {arXiv:1306.2162 [cond-mat.quant-gas]} \BibitemShut {NoStop}%
\bibitem [{\citenamefont {Wiese}(2013)}]{Wiese:2013uua}%
  \BibitemOpen
  \bibfield  {author} {\bibinfo {author} {\bibfnamefont {Uwe-Jens}\
  \bibnamefont {Wiese}},\ }\bibfield  {title} {\enquote {\bibinfo {title}
  {{Ultracold Quantum Gases and Lattice Systems: Quantum Simulation of Lattice
  Gauge Theories}},}\ }\href {\doibase 10.1002/andp.201300104} {\bibfield
  {journal} {\bibinfo  {journal} {Annalen Phys.}\ }\textbf {\bibinfo {volume}
  {525}},\ \bibinfo {pages} {777--796} (\bibinfo {year} {2013})},\ \Eprint
  {http://arxiv.org/abs/1305.1602} {arXiv:1305.1602 [quant-ph]} \BibitemShut
  {NoStop}%
\bibitem [{\citenamefont {Marcos}\ \emph {et~al.}(2013)\citenamefont {Marcos},
  \citenamefont {Rabl}, \citenamefont {Rico},\ and\ \citenamefont
  {Zoller}}]{Marcos:2013aya}%
  \BibitemOpen
  \bibfield  {author} {\bibinfo {author} {\bibfnamefont {D.}~\bibnamefont
  {Marcos}}, \bibinfo {author} {\bibfnamefont {P.}~\bibnamefont {Rabl}},
  \bibinfo {author} {\bibfnamefont {E.}~\bibnamefont {Rico}}, \ and\ \bibinfo
  {author} {\bibfnamefont {P.}~\bibnamefont {Zoller}},\ }\bibfield  {title}
  {\enquote {\bibinfo {title} {{Superconducting Circuits for Quantum Simulation
  of Dynamical Gauge Fields}},}\ }\href {\doibase
  10.1103/PhysRevLett.111.110504} {\bibfield  {journal} {\bibinfo  {journal}
  {Phys. Rev. Lett.}\ }\textbf {\bibinfo {volume} {111}},\ \bibinfo {pages}
  {110504} (\bibinfo {year} {2013})},\ \Eprint {http://arxiv.org/abs/1306.1674}
  {arXiv:1306.1674 [cond-mat.mes-hall]} \BibitemShut {NoStop}%
\bibitem [{\citenamefont {Zohar}\ \emph {et~al.}(2016)\citenamefont {Zohar},
  \citenamefont {Cirac},\ and\ \citenamefont {Reznik}}]{Zohar:2015hwa}%
  \BibitemOpen
  \bibfield  {author} {\bibinfo {author} {\bibfnamefont {Erez}\ \bibnamefont
  {Zohar}}, \bibinfo {author} {\bibfnamefont {J.~Ignacio}\ \bibnamefont
  {Cirac}}, \ and\ \bibinfo {author} {\bibfnamefont {Benni}\ \bibnamefont
  {Reznik}},\ }\bibfield  {title} {\enquote {\bibinfo {title} {{Quantum
  Simulations of Lattice Gauge Theories using Ultracold Atoms in Optical
  Lattices}},}\ }\href {\doibase 10.1088/0034-4885/79/1/014401} {\bibfield
  {journal} {\bibinfo  {journal} {Rept. Prog. Phys.}\ }\textbf {\bibinfo
  {volume} {79}},\ \bibinfo {pages} {014401} (\bibinfo {year} {2016})},\
  \Eprint {http://arxiv.org/abs/1503.02312} {arXiv:1503.02312 [quant-ph]}
  \BibitemShut {NoStop}%
\bibitem [{\citenamefont {Yang}\ \emph {et~al.}(2016)\citenamefont {Yang},
  \citenamefont {Giri}, \citenamefont {Johanning}, \citenamefont {Wunderlich},
  \citenamefont {Zoller},\ and\ \citenamefont {Hauke}}]{Yang:2016hjn}%
  \BibitemOpen
  \bibfield  {author} {\bibinfo {author} {\bibfnamefont {Dayou}\ \bibnamefont
  {Yang}}, \bibinfo {author} {\bibfnamefont {Gouri~Shankar}\ \bibnamefont
  {Giri}}, \bibinfo {author} {\bibfnamefont {Michael}\ \bibnamefont
  {Johanning}}, \bibinfo {author} {\bibfnamefont {Christof}\ \bibnamefont
  {Wunderlich}}, \bibinfo {author} {\bibfnamefont {Peter}\ \bibnamefont
  {Zoller}}, \ and\ \bibinfo {author} {\bibfnamefont {Philipp}\ \bibnamefont
  {Hauke}},\ }\bibfield  {title} {\enquote {\bibinfo {title} {{Analog quantum
  simulation of (1+1)-dimensional lattice QED with trapped ions}},}\ }\href
  {\doibase 10.1103/PhysRevA.94.052321} {\bibfield  {journal} {\bibinfo
  {journal} {Phys. Rev. A}\ }\textbf {\bibinfo {volume} {94}},\ \bibinfo
  {pages} {052321} (\bibinfo {year} {2016})},\ \Eprint
  {http://arxiv.org/abs/1604.03124} {arXiv:1604.03124 [quant-ph]} \BibitemShut
  {NoStop}%
\bibitem [{\citenamefont {Bender}\ \emph {et~al.}(2018)\citenamefont {Bender},
  \citenamefont {Zohar}, \citenamefont {Farace},\ and\ \citenamefont
  {Cirac}}]{Bender:2018rdp}%
  \BibitemOpen
  \bibfield  {author} {\bibinfo {author} {\bibfnamefont {Julian}\ \bibnamefont
  {Bender}}, \bibinfo {author} {\bibfnamefont {Erez}\ \bibnamefont {Zohar}},
  \bibinfo {author} {\bibfnamefont {Alessandro}\ \bibnamefont {Farace}}, \ and\
  \bibinfo {author} {\bibfnamefont {J.~Ignacio}\ \bibnamefont {Cirac}},\
  }\bibfield  {title} {\enquote {\bibinfo {title} {{Digital quantum simulation
  of lattice gauge theories in three spatial dimensions}},}\ }\href {\doibase
  10.1088/1367-2630/aadb71} {\bibfield  {journal} {\bibinfo  {journal} {New J.
  Phys.}\ }\textbf {\bibinfo {volume} {20}},\ \bibinfo {pages} {093001}
  (\bibinfo {year} {2018})},\ \Eprint {http://arxiv.org/abs/1804.02082}
  {arXiv:1804.02082 [quant-ph]} \BibitemShut {NoStop}%
\bibitem [{\citenamefont {Luo}\ \emph {et~al.}(2020)\citenamefont {Luo},
  \citenamefont {Shen}, \citenamefont {Highman}, \citenamefont {Clark},
  \citenamefont {DeMarco}, \citenamefont {El-Khadra},\ and\ \citenamefont
  {Gadway}}]{Luo:2019vmi}%
  \BibitemOpen
  \bibfield  {author} {\bibinfo {author} {\bibfnamefont {Di}~\bibnamefont
  {Luo}}, \bibinfo {author} {\bibfnamefont {Jiayu}\ \bibnamefont {Shen}},
  \bibinfo {author} {\bibfnamefont {Michael}\ \bibnamefont {Highman}}, \bibinfo
  {author} {\bibfnamefont {Bryan~K.}\ \bibnamefont {Clark}}, \bibinfo {author}
  {\bibfnamefont {Brian}\ \bibnamefont {DeMarco}}, \bibinfo {author}
  {\bibfnamefont {Aida~X.}\ \bibnamefont {El-Khadra}}, \ and\ \bibinfo {author}
  {\bibfnamefont {Bryce}\ \bibnamefont {Gadway}},\ }\bibfield  {title}
  {\enquote {\bibinfo {title} {{Framework for simulating gauge theories with
  dipolar spin systems}},}\ }\href {\doibase 10.1103/PhysRevA.102.032617}
  {\bibfield  {journal} {\bibinfo  {journal} {Phys. Rev. A}\ }\textbf {\bibinfo
  {volume} {102}},\ \bibinfo {pages} {032617} (\bibinfo {year} {2020})},\
  \Eprint {http://arxiv.org/abs/1912.11488} {arXiv:1912.11488 [quant-ph]}
  \BibitemShut {NoStop}%
\bibitem [{\citenamefont {Notarnicola}\ \emph {et~al.}(2020)\citenamefont
  {Notarnicola}, \citenamefont {Collura},\ and\ \citenamefont
  {Montangero}}]{Notarnicola:2019wzb}%
  \BibitemOpen
  \bibfield  {author} {\bibinfo {author} {\bibfnamefont {Simone}\ \bibnamefont
  {Notarnicola}}, \bibinfo {author} {\bibfnamefont {Mario}\ \bibnamefont
  {Collura}}, \ and\ \bibinfo {author} {\bibfnamefont {Simone}\ \bibnamefont
  {Montangero}},\ }\bibfield  {title} {\enquote {\bibinfo {title}
  {{Real-time-dynamics quantum simulation of (1+1)-dimensional lattice QED with
  Rydberg atoms}},}\ }\href {\doibase 10.1103/PhysRevResearch.2.013288}
  {\bibfield  {journal} {\bibinfo  {journal} {Phys. Rev. Res.}\ }\textbf
  {\bibinfo {volume} {2}},\ \bibinfo {pages} {013288} (\bibinfo {year}
  {2020})},\ \Eprint {http://arxiv.org/abs/1907.12579} {arXiv:1907.12579
  [cond-mat.quant-gas]} \BibitemShut {NoStop}%
\bibitem [{\citenamefont {Surace}\ \emph {et~al.}(2020)\citenamefont {Surace},
  \citenamefont {Mazza}, \citenamefont {Giudici}, \citenamefont {Lerose},
  \citenamefont {Gambassi},\ and\ \citenamefont {Dalmonte}}]{Surace:2019dtp}%
  \BibitemOpen
  \bibfield  {author} {\bibinfo {author} {\bibfnamefont {Federica~M.}\
  \bibnamefont {Surace}}, \bibinfo {author} {\bibfnamefont {Paolo~P.}\
  \bibnamefont {Mazza}}, \bibinfo {author} {\bibfnamefont {Giuliano}\
  \bibnamefont {Giudici}}, \bibinfo {author} {\bibfnamefont {Alessio}\
  \bibnamefont {Lerose}}, \bibinfo {author} {\bibfnamefont {Andrea}\
  \bibnamefont {Gambassi}}, \ and\ \bibinfo {author} {\bibfnamefont {Marcello}\
  \bibnamefont {Dalmonte}},\ }\bibfield  {title} {\enquote {\bibinfo {title}
  {{Lattice gauge theories and string dynamics in Rydberg atom quantum
  simulators}},}\ }\href {\doibase 10.1103/PhysRevX.10.021041} {\bibfield
  {journal} {\bibinfo  {journal} {Phys. Rev. X}\ }\textbf {\bibinfo {volume}
  {10}},\ \bibinfo {pages} {021041} (\bibinfo {year} {2020})},\ \Eprint
  {http://arxiv.org/abs/1902.09551} {arXiv:1902.09551 [cond-mat.quant-gas]}
  \BibitemShut {NoStop}%
\bibitem [{\citenamefont {Surace}\ and\ \citenamefont
  {Lerose}(2021)}]{Surace:2020ycc}%
  \BibitemOpen
  \bibfield  {author} {\bibinfo {author} {\bibfnamefont {Federica~Maria}\
  \bibnamefont {Surace}}\ and\ \bibinfo {author} {\bibfnamefont {Alessio}\
  \bibnamefont {Lerose}},\ }\bibfield  {title} {\enquote {\bibinfo {title}
  {{Scattering of mesons in quantum simulators}},}\ }\href {\doibase
  10.1088/1367-2630/abfc40} {\bibfield  {journal} {\bibinfo  {journal} {New J.
  Phys.}\ }\textbf {\bibinfo {volume} {23}},\ \bibinfo {pages} {062001}
  (\bibinfo {year} {2021})},\ \Eprint {http://arxiv.org/abs/2011.10583}
  {arXiv:2011.10583 [cond-mat.quant-gas]} \BibitemShut {NoStop}%
\bibitem [{\citenamefont {Kasper}\ \emph {et~al.}(2020)\citenamefont {Kasper},
  \citenamefont {Juzeliunas}, \citenamefont {Lewenstein}, \citenamefont
  {Jendrzejewski},\ and\ \citenamefont {Zohar}}]{Kasper:2020akk}%
  \BibitemOpen
  \bibfield  {author} {\bibinfo {author} {\bibfnamefont {Valentin}\
  \bibnamefont {Kasper}}, \bibinfo {author} {\bibfnamefont {Gediminas}\
  \bibnamefont {Juzeliunas}}, \bibinfo {author} {\bibfnamefont {Maciej}\
  \bibnamefont {Lewenstein}}, \bibinfo {author} {\bibfnamefont {Fred}\
  \bibnamefont {Jendrzejewski}}, \ and\ \bibinfo {author} {\bibfnamefont
  {Erez}\ \bibnamefont {Zohar}},\ }\bibfield  {title} {\enquote {\bibinfo
  {title} {{From the Jaynes\textendash{}Cummings model to non-abelian gauge
  theories: a guided tour for the quantum engineer}},}\ }\href {\doibase
  10.1088/1367-2630/abb961} {\bibfield  {journal} {\bibinfo  {journal} {New J.
  Phys.}\ }\textbf {\bibinfo {volume} {22}},\ \bibinfo {pages} {103027}
  (\bibinfo {year} {2020})},\ \Eprint {http://arxiv.org/abs/2006.01258}
  {arXiv:2006.01258 [quant-ph]} \BibitemShut {NoStop}%
\bibitem [{\citenamefont {Aidelsburger}\ \emph {et~al.}(2021)\citenamefont
  {Aidelsburger} \emph {et~al.}}]{Aidelsburger:2021mia}%
  \BibitemOpen
  \bibfield  {author} {\bibinfo {author} {\bibfnamefont {Monika}\ \bibnamefont
  {Aidelsburger}} \emph {et~al.},\ }\bibfield  {title} {\enquote {\bibinfo
  {title} {{Cold atoms meet lattice gauge theory}},}\ }\href {\doibase
  10.1098/rsta.2021.0064} {\bibfield  {journal} {\bibinfo  {journal} {Phil.
  Trans. Roy. Soc. Lond. A}\ }\textbf {\bibinfo {volume} {380}},\ \bibinfo
  {pages} {20210064} (\bibinfo {year} {2021})},\ \Eprint
  {http://arxiv.org/abs/2106.03063} {arXiv:2106.03063 [cond-mat.quant-gas]}
  \BibitemShut {NoStop}%
\bibitem [{\citenamefont {Surace}\ \emph {et~al.}(2024)\citenamefont {Surace},
  \citenamefont {Fromholz}, \citenamefont {Scazza},\ and\ \citenamefont
  {Dalmonte}}]{Surace:2023qwo}%
  \BibitemOpen
  \bibfield  {author} {\bibinfo {author} {\bibfnamefont {Federica~Maria}\
  \bibnamefont {Surace}}, \bibinfo {author} {\bibfnamefont {Pierre}\
  \bibnamefont {Fromholz}}, \bibinfo {author} {\bibfnamefont {Francesco}\
  \bibnamefont {Scazza}}, \ and\ \bibinfo {author} {\bibfnamefont {Marcello}\
  \bibnamefont {Dalmonte}},\ }\bibfield  {title} {\enquote {\bibinfo {title}
  {{Scalable, ab initio protocol for quantum simulating SU($N$)$\times$U(1)
  Lattice Gauge Theories}},}\ }\href {\doibase 10.22331/q-2024-05-23-1359}
  {\bibfield  {journal} {\bibinfo  {journal} {Quantum}\ }\textbf {\bibinfo
  {volume} {8}},\ \bibinfo {pages} {1359} (\bibinfo {year} {2024})},\ \Eprint
  {http://arxiv.org/abs/2310.08643} {arXiv:2310.08643 [cond-mat.quant-gas]}
  \BibitemShut {NoStop}%
\bibitem [{\citenamefont {Zache}\ \emph
  {et~al.}(2023{\natexlab{a}})\citenamefont {Zache}, \citenamefont
  {Gonz\'alez-Cuadra},\ and\ \citenamefont {Zoller}}]{Zache:2023cfj}%
  \BibitemOpen
  \bibfield  {author} {\bibinfo {author} {\bibfnamefont {Torsten~V.}\
  \bibnamefont {Zache}}, \bibinfo {author} {\bibfnamefont {Daniel}\
  \bibnamefont {Gonz\'alez-Cuadra}}, \ and\ \bibinfo {author} {\bibfnamefont
  {Peter}\ \bibnamefont {Zoller}},\ }\bibfield  {title} {\enquote {\bibinfo
  {title} {{Fermion-qudit quantum processors for simulating lattice gauge
  theories with matter}},}\ }\href {\doibase 10.22331/q-2023-10-16-1140}
  {\bibfield  {journal} {\bibinfo  {journal} {Quantum}\ }\textbf {\bibinfo
  {volume} {7}},\ \bibinfo {pages} {1140} (\bibinfo {year}
  {2023}{\natexlab{a}})},\ \Eprint {http://arxiv.org/abs/2303.08683}
  {arXiv:2303.08683 [quant-ph]} \BibitemShut {NoStop}%
\bibitem [{\citenamefont {Popov}\ \emph {et~al.}(2024)\citenamefont {Popov},
  \citenamefont {Meth}, \citenamefont {Lewenstein}, \citenamefont {Hauke},
  \citenamefont {Ringbauer}, \citenamefont {Zohar},\ and\ \citenamefont
  {Kasper}}]{Popov:2023xft}%
  \BibitemOpen
  \bibfield  {author} {\bibinfo {author} {\bibfnamefont {Pavel~P.}\
  \bibnamefont {Popov}}, \bibinfo {author} {\bibfnamefont {Michael}\
  \bibnamefont {Meth}}, \bibinfo {author} {\bibfnamefont {Maciej}\ \bibnamefont
  {Lewenstein}}, \bibinfo {author} {\bibfnamefont {Philipp}\ \bibnamefont
  {Hauke}}, \bibinfo {author} {\bibfnamefont {Martin}\ \bibnamefont
  {Ringbauer}}, \bibinfo {author} {\bibfnamefont {Erez}\ \bibnamefont {Zohar}},
  \ and\ \bibinfo {author} {\bibfnamefont {Valentin}\ \bibnamefont {Kasper}},\
  }\bibfield  {title} {\enquote {\bibinfo {title} {{Variational quantum
  simulation of U(1) lattice gauge theories with qudit systems}},}\ }\href
  {\doibase 10.1103/PhysRevResearch.6.013202} {\bibfield  {journal} {\bibinfo
  {journal} {Phys. Rev. Res.}\ }\textbf {\bibinfo {volume} {6}},\ \bibinfo
  {pages} {013202} (\bibinfo {year} {2024})},\ \Eprint
  {http://arxiv.org/abs/2307.15173} {arXiv:2307.15173 [quant-ph]} \BibitemShut
  {NoStop}%
\bibitem [{\citenamefont {Gonz\'alez-Cuadra}\ \emph {et~al.}(2022)\citenamefont
  {Gonz\'alez-Cuadra}, \citenamefont {Zache}, \citenamefont {Carrasco},
  \citenamefont {Kraus},\ and\ \citenamefont
  {Zoller}}]{Gonzalez-Cuadra:2022hxt}%
  \BibitemOpen
  \bibfield  {author} {\bibinfo {author} {\bibfnamefont {Daniel}\ \bibnamefont
  {Gonz\'alez-Cuadra}}, \bibinfo {author} {\bibfnamefont {Torsten~V.}\
  \bibnamefont {Zache}}, \bibinfo {author} {\bibfnamefont {Jose}\ \bibnamefont
  {Carrasco}}, \bibinfo {author} {\bibfnamefont {Barbara}\ \bibnamefont
  {Kraus}}, \ and\ \bibinfo {author} {\bibfnamefont {Peter}\ \bibnamefont
  {Zoller}},\ }\bibfield  {title} {\enquote {\bibinfo {title} {{Hardware
  Efficient Quantum Simulation of Non-Abelian Gauge Theories with Qudits on
  Rydberg Platforms}},}\ }\href {\doibase 10.1103/PhysRevLett.129.160501}
  {\bibfield  {journal} {\bibinfo  {journal} {Phys. Rev. Lett.}\ }\textbf
  {\bibinfo {volume} {129}},\ \bibinfo {pages} {160501} (\bibinfo {year}
  {2022})},\ \Eprint {http://arxiv.org/abs/2203.15541} {arXiv:2203.15541
  [quant-ph]} \BibitemShut {NoStop}%
\bibitem [{\citenamefont {Illa}\ \emph {et~al.}(2024)\citenamefont {Illa},
  \citenamefont {Robin},\ and\ \citenamefont {Savage}}]{Illa:2024kmf}%
  \BibitemOpen
  \bibfield  {author} {\bibinfo {author} {\bibfnamefont {Marc}\ \bibnamefont
  {Illa}}, \bibinfo {author} {\bibfnamefont {Caroline E.~P.}\ \bibnamefont
  {Robin}}, \ and\ \bibinfo {author} {\bibfnamefont {Martin~J.}\ \bibnamefont
  {Savage}},\ }\bibfield  {title} {\enquote {\bibinfo {title} {{Qu8its for
  quantum simulations of lattice quantum chromodynamics}},}\ }\href {\doibase
  10.1103/PhysRevD.110.014507} {\bibfield  {journal} {\bibinfo  {journal}
  {Phys. Rev. D}\ }\textbf {\bibinfo {volume} {110}},\ \bibinfo {pages}
  {014507} (\bibinfo {year} {2024})},\ \Eprint
  {http://arxiv.org/abs/2403.14537} {arXiv:2403.14537 [quant-ph]} \BibitemShut
  {NoStop}%
\bibitem [{\citenamefont {Zache}\ \emph {et~al.}(2018)\citenamefont {Zache},
  \citenamefont {Hebenstreit}, \citenamefont {Jendrzejewski}, \citenamefont
  {Oberthaler}, \citenamefont {Berges},\ and\ \citenamefont
  {Hauke}}]{Zache:2018jbt}%
  \BibitemOpen
  \bibfield  {author} {\bibinfo {author} {\bibfnamefont {T.~V.}\ \bibnamefont
  {Zache}}, \bibinfo {author} {\bibfnamefont {F.}~\bibnamefont {Hebenstreit}},
  \bibinfo {author} {\bibfnamefont {F.}~\bibnamefont {Jendrzejewski}}, \bibinfo
  {author} {\bibfnamefont {M.~K.}\ \bibnamefont {Oberthaler}}, \bibinfo
  {author} {\bibfnamefont {J.}~\bibnamefont {Berges}}, \ and\ \bibinfo {author}
  {\bibfnamefont {P.}~\bibnamefont {Hauke}},\ }\bibfield  {title} {\enquote
  {\bibinfo {title} {{Quantum simulation of lattice gauge theories using Wilson
  fermions}},}\ }\href {\doibase 10.1088/2058-9565/aac33b} {\bibfield
  {journal} {\bibinfo  {journal} {Quantum Sci. Technol.}\ }\textbf {\bibinfo
  {volume} {3}},\ \bibinfo {pages} {034010} (\bibinfo {year} {2018})},\ \Eprint
  {http://arxiv.org/abs/1802.06704} {arXiv:1802.06704 [cond-mat.quant-gas]}
  \BibitemShut {NoStop}%
\bibitem [{\citenamefont {Armon}\ \emph {et~al.}(2021)\citenamefont {Armon},
  \citenamefont {Ashkenazi}, \citenamefont {Garc\'\i{}a-Moreno}, \citenamefont
  {Gonz\'alez-Tudela},\ and\ \citenamefont {Zohar}}]{Armon:2021uqr}%
  \BibitemOpen
  \bibfield  {author} {\bibinfo {author} {\bibfnamefont {Tsafrir}\ \bibnamefont
  {Armon}}, \bibinfo {author} {\bibfnamefont {Shachar}\ \bibnamefont
  {Ashkenazi}}, \bibinfo {author} {\bibfnamefont {Gerardo}\ \bibnamefont
  {Garc\'\i{}a-Moreno}}, \bibinfo {author} {\bibfnamefont {Alejandro}\
  \bibnamefont {Gonz\'alez-Tudela}}, \ and\ \bibinfo {author} {\bibfnamefont
  {Erez}\ \bibnamefont {Zohar}},\ }\bibfield  {title} {\enquote {\bibinfo
  {title} {{Photon-Mediated Stroboscopic Quantum Simulation of a Z2 Lattice
  Gauge Theory}},}\ }\href {\doibase 10.1103/PhysRevLett.127.250501} {\bibfield
   {journal} {\bibinfo  {journal} {Phys. Rev. Lett.}\ }\textbf {\bibinfo
  {volume} {127}},\ \bibinfo {pages} {250501} (\bibinfo {year} {2021})},\
  \Eprint {http://arxiv.org/abs/2107.13024} {arXiv:2107.13024 [quant-ph]}
  \BibitemShut {NoStop}%
\bibitem [{\citenamefont {Byrnes}\ and\ \citenamefont
  {Yamamoto}(2006)}]{Byrnes:2005qx}%
  \BibitemOpen
  \bibfield  {author} {\bibinfo {author} {\bibfnamefont {Tim}\ \bibnamefont
  {Byrnes}}\ and\ \bibinfo {author} {\bibfnamefont {Yoshihisa}\ \bibnamefont
  {Yamamoto}},\ }\bibfield  {title} {\enquote {\bibinfo {title} {{Simulating
  lattice gauge theories on a quantum computer}},}\ }\href {\doibase
  10.1103/PhysRevA.73.022328} {\bibfield  {journal} {\bibinfo  {journal} {Phys.
  Rev. A}\ }\textbf {\bibinfo {volume} {73}},\ \bibinfo {pages} {022328}
  (\bibinfo {year} {2006})},\ \Eprint {http://arxiv.org/abs/quant-ph/0510027}
  {arXiv:quant-ph/0510027} \BibitemShut {NoStop}%
\bibitem [{\citenamefont {Gonz\'alez-Cuadra}\ \emph {et~al.}(2017)\citenamefont
  {Gonz\'alez-Cuadra}, \citenamefont {Zohar},\ and\ \citenamefont
  {Cirac}}]{Gonzalez-Cuadra:2017lvz}%
  \BibitemOpen
  \bibfield  {author} {\bibinfo {author} {\bibfnamefont {Daniel}\ \bibnamefont
  {Gonz\'alez-Cuadra}}, \bibinfo {author} {\bibfnamefont {Erez}\ \bibnamefont
  {Zohar}}, \ and\ \bibinfo {author} {\bibfnamefont {J.~Ignacio}\ \bibnamefont
  {Cirac}},\ }\bibfield  {title} {\enquote {\bibinfo {title} {{Quantum
  Simulation of the Abelian-Higgs Lattice Gauge Theory with Ultracold
  Atoms}},}\ }\href {\doibase 10.1088/1367-2630/aa6f37} {\bibfield  {journal}
  {\bibinfo  {journal} {New J. Phys.}\ }\textbf {\bibinfo {volume} {19}},\
  \bibinfo {pages} {063038} (\bibinfo {year} {2017})},\ \Eprint
  {http://arxiv.org/abs/1702.05492} {arXiv:1702.05492 [quant-ph]} \BibitemShut
  {NoStop}%
\bibitem [{\citenamefont {Tagliacozzo}\ \emph
  {et~al.}(2013{\natexlab{a}})\citenamefont {Tagliacozzo}, \citenamefont
  {Celi}, \citenamefont {Orland},\ and\ \citenamefont
  {Lewenstein}}]{Tagliacozzo:2012df}%
  \BibitemOpen
  \bibfield  {author} {\bibinfo {author} {\bibfnamefont {L.}~\bibnamefont
  {Tagliacozzo}}, \bibinfo {author} {\bibfnamefont {A.}~\bibnamefont {Celi}},
  \bibinfo {author} {\bibfnamefont {P.}~\bibnamefont {Orland}}, \ and\ \bibinfo
  {author} {\bibfnamefont {M.}~\bibnamefont {Lewenstein}},\ }\bibfield  {title}
  {\enquote {\bibinfo {title} {{Simulations of non-Abelian gauge theories with
  optical lattices}},}\ }\href {\doibase 10.1038/ncomms3615} {\bibfield
  {journal} {\bibinfo  {journal} {Nature Commun.}\ }\textbf {\bibinfo {volume}
  {4}},\ \bibinfo {pages} {2615} (\bibinfo {year} {2013}{\natexlab{a}})},\
  \Eprint {http://arxiv.org/abs/1211.2704} {arXiv:1211.2704
  [cond-mat.quant-gas]} \BibitemShut {NoStop}%
\bibitem [{\citenamefont {Tagliacozzo}\ \emph
  {et~al.}(2013{\natexlab{b}})\citenamefont {Tagliacozzo}, \citenamefont
  {Celi}, \citenamefont {Zamora},\ and\ \citenamefont
  {Lewenstein}}]{Tagliacozzo:2012vg}%
  \BibitemOpen
  \bibfield  {author} {\bibinfo {author} {\bibfnamefont {L.}~\bibnamefont
  {Tagliacozzo}}, \bibinfo {author} {\bibfnamefont {A.}~\bibnamefont {Celi}},
  \bibinfo {author} {\bibfnamefont {A.}~\bibnamefont {Zamora}}, \ and\ \bibinfo
  {author} {\bibfnamefont {M.}~\bibnamefont {Lewenstein}},\ }\bibfield  {title}
  {\enquote {\bibinfo {title} {{Optical Abelian Lattice Gauge Theories}},}\
  }\href {\doibase 10.1016/j.aop.2012.11.009} {\bibfield  {journal} {\bibinfo
  {journal} {Annals Phys.}\ }\textbf {\bibinfo {volume} {330}},\ \bibinfo
  {pages} {160--191} (\bibinfo {year} {2013}{\natexlab{b}})},\ \Eprint
  {http://arxiv.org/abs/1205.0496} {arXiv:1205.0496 [cond-mat.quant-gas]}
  \BibitemShut {NoStop}%
\bibitem [{\citenamefont {Zhou}\ \emph {et~al.}(2022)\citenamefont {Zhou},
  \citenamefont {Su}, \citenamefont {Halimeh}, \citenamefont {Ott},
  \citenamefont {Sun}, \citenamefont {Hauke}, \citenamefont {Yang},
  \citenamefont {Yuan}, \citenamefont {Berges},\ and\ \citenamefont
  {Pan}}]{Zhou:2021kdl}%
  \BibitemOpen
  \bibfield  {author} {\bibinfo {author} {\bibfnamefont {Zhao-Yu}\ \bibnamefont
  {Zhou}}, \bibinfo {author} {\bibfnamefont {Guo-Xian}\ \bibnamefont {Su}},
  \bibinfo {author} {\bibfnamefont {Jad~C.}\ \bibnamefont {Halimeh}}, \bibinfo
  {author} {\bibfnamefont {Robert}\ \bibnamefont {Ott}}, \bibinfo {author}
  {\bibfnamefont {Hui}\ \bibnamefont {Sun}}, \bibinfo {author} {\bibfnamefont
  {Philipp}\ \bibnamefont {Hauke}}, \bibinfo {author} {\bibfnamefont {Bing}\
  \bibnamefont {Yang}}, \bibinfo {author} {\bibfnamefont {Zhen-Sheng}\
  \bibnamefont {Yuan}}, \bibinfo {author} {\bibfnamefont {J\"urgen}\
  \bibnamefont {Berges}}, \ and\ \bibinfo {author} {\bibfnamefont {Jian-Wei}\
  \bibnamefont {Pan}},\ }\bibfield  {title} {\enquote {\bibinfo {title}
  {{Thermalization dynamics of a gauge theory on a quantum simulator}},}\
  }\href {\doibase 10.1126/science.abl6277} {\bibfield  {journal} {\bibinfo
  {journal} {Science}\ }\textbf {\bibinfo {volume} {377}},\ \bibinfo {pages}
  {abl6277} (\bibinfo {year} {2022})},\ \Eprint
  {http://arxiv.org/abs/2107.13563} {arXiv:2107.13563 [cond-mat.quant-gas]}
  \BibitemShut {NoStop}%
\bibitem [{\citenamefont {Yamamoto}\ and\ \citenamefont
  {Doi}(2024)}]{Yamamoto:2022jnn}%
  \BibitemOpen
  \bibfield  {author} {\bibinfo {author} {\bibfnamefont {Arata}\ \bibnamefont
  {Yamamoto}}\ and\ \bibinfo {author} {\bibfnamefont {Takumi}\ \bibnamefont
  {Doi}},\ }\bibfield  {title} {\enquote {\bibinfo {title} {{Toward Nuclear
  Physics from Lattice QCD on Quantum Computers}},}\ }\href {\doibase
  10.1093/ptep/ptae019} {\bibfield  {journal} {\bibinfo  {journal} {PTEP}\
  }\textbf {\bibinfo {volume} {2024}},\ \bibinfo {pages} {033D02} (\bibinfo
  {year} {2024})},\ \Eprint {http://arxiv.org/abs/2211.14550} {arXiv:2211.14550
  [hep-lat]} \BibitemShut {NoStop}%
\bibitem [{\citenamefont {Zohar}\ and\ \citenamefont
  {Reznik}(2011)}]{Zohar:2011cw}%
  \BibitemOpen
  \bibfield  {author} {\bibinfo {author} {\bibfnamefont {Erez}\ \bibnamefont
  {Zohar}}\ and\ \bibinfo {author} {\bibfnamefont {Benni}\ \bibnamefont
  {Reznik}},\ }\bibfield  {title} {\enquote {\bibinfo {title} {{Confinement and
  lattice QED electric flux-tubes simulated with ultracold atoms}},}\ }\href
  {\doibase 10.1103/PhysRevLett.107.275301} {\bibfield  {journal} {\bibinfo
  {journal} {Phys. Rev. Lett.}\ }\textbf {\bibinfo {volume} {107}},\ \bibinfo
  {pages} {275301} (\bibinfo {year} {2011})},\ \Eprint
  {http://arxiv.org/abs/1108.1562} {arXiv:1108.1562 [quant-ph]} \BibitemShut
  {NoStop}%
\bibitem [{\citenamefont {Zohar}\ \emph {et~al.}(2017)\citenamefont {Zohar},
  \citenamefont {Farace}, \citenamefont {Reznik},\ and\ \citenamefont
  {Cirac}}]{Zohar:2016iic}%
  \BibitemOpen
  \bibfield  {author} {\bibinfo {author} {\bibfnamefont {Erez}\ \bibnamefont
  {Zohar}}, \bibinfo {author} {\bibfnamefont {Alessandro}\ \bibnamefont
  {Farace}}, \bibinfo {author} {\bibfnamefont {Benni}\ \bibnamefont {Reznik}},
  \ and\ \bibinfo {author} {\bibfnamefont {J.~Ignacio}\ \bibnamefont {Cirac}},\
  }\bibfield  {title} {\enquote {\bibinfo {title} {{Digital lattice gauge
  theories}},}\ }\href {\doibase 10.1103/PhysRevA.95.023604} {\bibfield
  {journal} {\bibinfo  {journal} {Phys. Rev. A}\ }\textbf {\bibinfo {volume}
  {95}},\ \bibinfo {pages} {023604} (\bibinfo {year} {2017})},\ \Eprint
  {http://arxiv.org/abs/1607.08121} {arXiv:1607.08121 [quant-ph]} \BibitemShut
  {NoStop}%
\bibitem [{\citenamefont {Alexandru}\ \emph {et~al.}(2022)\citenamefont
  {Alexandru}, \citenamefont {Bedaque}, \citenamefont {Brett},\ and\
  \citenamefont {Lamm}}]{Alexandru:2021jpm}%
  \BibitemOpen
  \bibfield  {author} {\bibinfo {author} {\bibfnamefont {Andrei}\ \bibnamefont
  {Alexandru}}, \bibinfo {author} {\bibfnamefont {Paulo~F.}\ \bibnamefont
  {Bedaque}}, \bibinfo {author} {\bibfnamefont {Ruair\'\i{}}\ \bibnamefont
  {Brett}}, \ and\ \bibinfo {author} {\bibfnamefont {Henry}\ \bibnamefont
  {Lamm}},\ }\bibfield  {title} {\enquote {\bibinfo {title} {{Spectrum of
  digitized QCD: Glueballs in a S(1080) gauge theory}},}\ }\href {\doibase
  10.1103/PhysRevD.105.114508} {\bibfield  {journal} {\bibinfo  {journal}
  {Phys. Rev. D}\ }\textbf {\bibinfo {volume} {105}},\ \bibinfo {pages}
  {114508} (\bibinfo {year} {2022})},\ \Eprint
  {http://arxiv.org/abs/2112.08482} {arXiv:2112.08482 [hep-lat]} \BibitemShut
  {NoStop}%
\bibitem [{\citenamefont {Lamm}\ \emph {et~al.}(2024)\citenamefont {Lamm},
  \citenamefont {Li}, \citenamefont {Shu}, \citenamefont {Wang},\ and\
  \citenamefont {Xu}}]{Lamm:2024jnl}%
  \BibitemOpen
  \bibfield  {author} {\bibinfo {author} {\bibfnamefont {Henry}\ \bibnamefont
  {Lamm}}, \bibinfo {author} {\bibfnamefont {Ying-Ying}\ \bibnamefont {Li}},
  \bibinfo {author} {\bibfnamefont {Jing}\ \bibnamefont {Shu}}, \bibinfo
  {author} {\bibfnamefont {Yi-Lin}\ \bibnamefont {Wang}}, \ and\ \bibinfo
  {author} {\bibfnamefont {Bin}\ \bibnamefont {Xu}},\ }\bibfield  {title}
  {\enquote {\bibinfo {title} {{Block encodings of discrete subgroups on a
  quantum computer}},}\ }\href {\doibase 10.1103/PhysRevD.110.054505}
  {\bibfield  {journal} {\bibinfo  {journal} {Phys. Rev. D}\ }\textbf {\bibinfo
  {volume} {110}},\ \bibinfo {pages} {054505} (\bibinfo {year} {2024})},\
  \Eprint {http://arxiv.org/abs/2405.12890} {arXiv:2405.12890 [hep-lat]}
  \BibitemShut {NoStop}%
\bibitem [{\citenamefont {Lamm}\ \emph
  {et~al.}(2020{\natexlab{a}})\citenamefont {Lamm}, \citenamefont {Lawrence},\
  and\ \citenamefont {Yamauchi}}]{Lamm:2019uyc}%
  \BibitemOpen
  \bibfield  {author} {\bibinfo {author} {\bibfnamefont {Henry}\ \bibnamefont
  {Lamm}}, \bibinfo {author} {\bibfnamefont {Scott}\ \bibnamefont {Lawrence}},
  \ and\ \bibinfo {author} {\bibfnamefont {Yukari}\ \bibnamefont {Yamauchi}}
  (\bibinfo {collaboration} {NuQS}),\ }\bibfield  {title} {\enquote {\bibinfo
  {title} {{Parton physics on a quantum computer}},}\ }\href {\doibase
  10.1103/PhysRevResearch.2.013272} {\bibfield  {journal} {\bibinfo  {journal}
  {Phys. Rev. Res.}\ }\textbf {\bibinfo {volume} {2}},\ \bibinfo {pages}
  {013272} (\bibinfo {year} {2020}{\natexlab{a}})},\ \Eprint
  {http://arxiv.org/abs/1908.10439} {arXiv:1908.10439 [hep-lat]} \BibitemShut
  {NoStop}%
\bibitem [{\citenamefont {Riechert}\ \emph {et~al.}(2022)\citenamefont
  {Riechert}, \citenamefont {Halimeh}, \citenamefont {Kasper}, \citenamefont
  {Bretheau}, \citenamefont {Zohar}, \citenamefont {Hauke},\ and\ \citenamefont
  {Jendrzejewski}}]{Riechert:2021ink}%
  \BibitemOpen
  \bibfield  {author} {\bibinfo {author} {\bibfnamefont {Hannes}\ \bibnamefont
  {Riechert}}, \bibinfo {author} {\bibfnamefont {Jad~C.}\ \bibnamefont
  {Halimeh}}, \bibinfo {author} {\bibfnamefont {Valentin}\ \bibnamefont
  {Kasper}}, \bibinfo {author} {\bibfnamefont {Landry}\ \bibnamefont
  {Bretheau}}, \bibinfo {author} {\bibfnamefont {Erez}\ \bibnamefont {Zohar}},
  \bibinfo {author} {\bibfnamefont {Philipp}\ \bibnamefont {Hauke}}, \ and\
  \bibinfo {author} {\bibfnamefont {Fred}\ \bibnamefont {Jendrzejewski}},\
  }\bibfield  {title} {\enquote {\bibinfo {title} {{Engineering a U(1) lattice
  gauge theory in classical electric circuits}},}\ }\href {\doibase
  10.1103/PhysRevB.105.205141} {\bibfield  {journal} {\bibinfo  {journal}
  {Phys. Rev. B}\ }\textbf {\bibinfo {volume} {105}},\ \bibinfo {pages}
  {205141} (\bibinfo {year} {2022})},\ \Eprint
  {http://arxiv.org/abs/2108.01086} {arXiv:2108.01086 [cond-mat.mes-hall]}
  \BibitemShut {NoStop}%
\bibitem [{\citenamefont {Semeghini}\ \emph {et~al.}(2021)\citenamefont
  {Semeghini} \emph {et~al.}}]{Semeghini:2021wls}%
  \BibitemOpen
  \bibfield  {author} {\bibinfo {author} {\bibfnamefont {Giulia}\ \bibnamefont
  {Semeghini}} \emph {et~al.},\ }\bibfield  {title} {\enquote {\bibinfo {title}
  {{Probing topological spin liquids on a programmable quantum simulator}},}\
  }\href {\doibase 10.1126/science.abi8794} {\bibfield  {journal} {\bibinfo
  {journal} {Science}\ }\textbf {\bibinfo {volume} {374}},\ \bibinfo {pages}
  {abi8794} (\bibinfo {year} {2021})},\ \Eprint
  {http://arxiv.org/abs/2104.04119} {arXiv:2104.04119 [quant-ph]} \BibitemShut
  {NoStop}%
\bibitem [{\citenamefont {Kokail}\ \emph {et~al.}(2019)\citenamefont {Kokail}
  \emph {et~al.}}]{Kokail:2018eiw}%
  \BibitemOpen
  \bibfield  {author} {\bibinfo {author} {\bibfnamefont {Christian}\
  \bibnamefont {Kokail}} \emph {et~al.},\ }\bibfield  {title} {\enquote
  {\bibinfo {title} {{Self-verifying variational quantum simulation of lattice
  models}},}\ }\href {\doibase 10.1038/s41586-019-1177-4} {\bibfield  {journal}
  {\bibinfo  {journal} {Nature}\ }\textbf {\bibinfo {volume} {569}},\ \bibinfo
  {pages} {355--360} (\bibinfo {year} {2019})},\ \Eprint
  {http://arxiv.org/abs/1810.03421} {arXiv:1810.03421 [quant-ph]} \BibitemShut
  {NoStop}%
\bibitem [{\citenamefont {Gonzalez-Cuadra}\ \emph {et~al.}(2024)\citenamefont
  {Gonzalez-Cuadra} \emph {et~al.}}]{Gonzalez-Cuadra:2024xul}%
  \BibitemOpen
  \bibfield  {author} {\bibinfo {author} {\bibfnamefont {Daniel}\ \bibnamefont
  {Gonzalez-Cuadra}} \emph {et~al.},\ }\bibfield  {title} {\enquote {\bibinfo
  {title} {{Observation of string breaking on a (2 + 1)D Rydberg quantum
  simulator}},}\ }\href@noop {} {\  (\bibinfo {year} {2024})},\ \Eprint
  {http://arxiv.org/abs/2410.16558} {arXiv:2410.16558 [quant-ph]} \BibitemShut
  {NoStop}%
\bibitem [{\citenamefont {De}\ \emph {et~al.}(2024)\citenamefont {De} \emph
  {et~al.}}]{De:2024smi}%
  \BibitemOpen
  \bibfield  {author} {\bibinfo {author} {\bibfnamefont {Arinjoy}\ \bibnamefont
  {De}} \emph {et~al.},\ }\bibfield  {title} {\enquote {\bibinfo {title}
  {{Observation of string-breaking dynamics in a quantum simulator}},}\
  }\href@noop {} {\  (\bibinfo {year} {2024})},\ \Eprint
  {http://arxiv.org/abs/2410.13815} {arXiv:2410.13815 [quant-ph]} \BibitemShut
  {NoStop}%
\bibitem [{\citenamefont {Pinto~Barros}\ \emph {et~al.}(2019)\citenamefont
  {Pinto~Barros}, \citenamefont {Dalmonte},\ and\ \citenamefont
  {Trombettoni}}]{PintoBarros:2018bzz}%
  \BibitemOpen
  \bibfield  {author} {\bibinfo {author} {\bibfnamefont {Joao~C.}\ \bibnamefont
  {Pinto~Barros}}, \bibinfo {author} {\bibfnamefont {Marcello}\ \bibnamefont
  {Dalmonte}}, \ and\ \bibinfo {author} {\bibfnamefont {Andrea}\ \bibnamefont
  {Trombettoni}},\ }\bibfield  {title} {\enquote {\bibinfo {title} {{String
  tension and robustness of confinement properties in the Schwinger-Thirring
  model}},}\ }\href {\doibase 10.1103/PhysRevD.100.036009} {\bibfield
  {journal} {\bibinfo  {journal} {Phys. Rev. D}\ }\textbf {\bibinfo {volume}
  {100}},\ \bibinfo {pages} {036009} (\bibinfo {year} {2019})},\ \Eprint
  {http://arxiv.org/abs/1808.00444} {arXiv:1808.00444 [hep-th]} \BibitemShut
  {NoStop}%
\bibitem [{\citenamefont {Fontana}\ \emph {et~al.}(2023)\citenamefont
  {Fontana}, \citenamefont {Barros},\ and\ \citenamefont
  {Trombettoni}}]{Fontana:2022dil}%
  \BibitemOpen
  \bibfield  {author} {\bibinfo {author} {\bibfnamefont {Pierpaolo}\
  \bibnamefont {Fontana}}, \bibinfo {author} {\bibfnamefont {Joao C.~Pinto}\
  \bibnamefont {Barros}}, \ and\ \bibinfo {author} {\bibfnamefont {Andrea}\
  \bibnamefont {Trombettoni}},\ }\bibfield  {title} {\enquote {\bibinfo {title}
  {{Quantum simulator of link models using spinor dipolar ultracold atoms}},}\
  }\href {\doibase 10.1103/PhysRevA.107.043312} {\bibfield  {journal} {\bibinfo
   {journal} {Phys. Rev. A}\ }\textbf {\bibinfo {volume} {107}},\ \bibinfo
  {pages} {043312} (\bibinfo {year} {2023})},\ \Eprint
  {http://arxiv.org/abs/2210.14836} {arXiv:2210.14836 [cond-mat.quant-gas]}
  \BibitemShut {NoStop}%
\bibitem [{\citenamefont {Su}\ \emph {et~al.}(2024)\citenamefont {Su},
  \citenamefont {Osborne},\ and\ \citenamefont {Halimeh}}]{Su:2024uuc}%
  \BibitemOpen
  \bibfield  {author} {\bibinfo {author} {\bibfnamefont {Guo-Xian}\
  \bibnamefont {Su}}, \bibinfo {author} {\bibfnamefont {Jesse~J.}\ \bibnamefont
  {Osborne}}, \ and\ \bibinfo {author} {\bibfnamefont {Jad~C.}\ \bibnamefont
  {Halimeh}},\ }\bibfield  {title} {\enquote {\bibinfo {title} {{Cold-Atom
  Particle Collider}},}\ }\href {\doibase 10.1103/PRXQuantum.5.040310}
  {\bibfield  {journal} {\bibinfo  {journal} {PRX Quantum}\ }\textbf {\bibinfo
  {volume} {5}},\ \bibinfo {pages} {040310} (\bibinfo {year} {2024})},\ \Eprint
  {http://arxiv.org/abs/2401.05489} {arXiv:2401.05489 [cond-mat.quant-gas]}
  \BibitemShut {NoStop}%
\bibitem [{\citenamefont {Bender}\ and\ \citenamefont
  {Zohar}(2020)}]{Bender:2020ztu}%
  \BibitemOpen
  \bibfield  {author} {\bibinfo {author} {\bibfnamefont {Julian}\ \bibnamefont
  {Bender}}\ and\ \bibinfo {author} {\bibfnamefont {Erez}\ \bibnamefont
  {Zohar}},\ }\bibfield  {title} {\enquote {\bibinfo {title} {{Gauge
  redundancy-free formulation of compact QED with dynamical matter for quantum
  and classical computations}},}\ }\href {\doibase 10.1103/PhysRevD.102.114517}
  {\bibfield  {journal} {\bibinfo  {journal} {Phys. Rev. D}\ }\textbf {\bibinfo
  {volume} {102}},\ \bibinfo {pages} {114517} (\bibinfo {year} {2020})},\
  \Eprint {http://arxiv.org/abs/2008.01349} {arXiv:2008.01349 [quant-ph]}
  \BibitemShut {NoStop}%
\bibitem [{\citenamefont {Fontana}\ \emph {et~al.}(2022)\citenamefont
  {Fontana}, \citenamefont {Barros},\ and\ \citenamefont
  {Trombettoni}}]{Fontana:2020xzp}%
  \BibitemOpen
  \bibfield  {author} {\bibinfo {author} {\bibfnamefont {Pierpaolo}\
  \bibnamefont {Fontana}}, \bibinfo {author} {\bibfnamefont {Joao C.~Pinto}\
  \bibnamefont {Barros}}, \ and\ \bibinfo {author} {\bibfnamefont {Andrea}\
  \bibnamefont {Trombettoni}},\ }\bibfield  {title} {\enquote {\bibinfo {title}
  {{Reformulation of gauge theories in terms of gauge invariant fields}},}\
  }\href {\doibase 10.1016/j.aop.2021.168683} {\bibfield  {journal} {\bibinfo
  {journal} {Annals Phys.}\ }\textbf {\bibinfo {volume} {436}},\ \bibinfo
  {pages} {168683} (\bibinfo {year} {2022})},\ \Eprint
  {http://arxiv.org/abs/2008.12973} {arXiv:2008.12973 [quant-ph]} \BibitemShut
  {NoStop}%
\bibitem [{\citenamefont {Schweizer}\ \emph {et~al.}(2019)\citenamefont
  {Schweizer}, \citenamefont {Grusdt}, \citenamefont {Berngruber},
  \citenamefont {Barbiero}, \citenamefont {Demler}, \citenamefont {Goldman},
  \citenamefont {Bloch},\ and\ \citenamefont {Aidelsburger}}]{Schweizer_2019}%
  \BibitemOpen
  \bibfield  {author} {\bibinfo {author} {\bibfnamefont {Christian}\
  \bibnamefont {Schweizer}}, \bibinfo {author} {\bibfnamefont {Fabian}\
  \bibnamefont {Grusdt}}, \bibinfo {author} {\bibfnamefont {Moritz}\
  \bibnamefont {Berngruber}}, \bibinfo {author} {\bibfnamefont {Luca}\
  \bibnamefont {Barbiero}}, \bibinfo {author} {\bibfnamefont {Eugene}\
  \bibnamefont {Demler}}, \bibinfo {author} {\bibfnamefont {Nathan}\
  \bibnamefont {Goldman}}, \bibinfo {author} {\bibfnamefont {Immanuel}\
  \bibnamefont {Bloch}}, \ and\ \bibinfo {author} {\bibfnamefont {Monika}\
  \bibnamefont {Aidelsburger}},\ }\bibfield  {title} {\enquote {\bibinfo
  {title} {Floquet approach to z2 lattice gauge theories with ultracold atoms
  in optical lattices},}\ }\href {\doibase 10.1038/s41567-019-0649-7}
  {\bibfield  {journal} {\bibinfo  {journal} {Nature Physics}\ }\textbf
  {\bibinfo {volume} {15}},\ \bibinfo {pages} {1168–1173} (\bibinfo {year}
  {2019})}\BibitemShut {NoStop}%
\bibitem [{\citenamefont {Kalinowski}\ \emph {et~al.}(2023)\citenamefont
  {Kalinowski}, \citenamefont {Maskara},\ and\ \citenamefont
  {Lukin}}]{Kalinowski:2022hze}%
  \BibitemOpen
  \bibfield  {author} {\bibinfo {author} {\bibfnamefont {Marcin}\ \bibnamefont
  {Kalinowski}}, \bibinfo {author} {\bibfnamefont {Nishad}\ \bibnamefont
  {Maskara}}, \ and\ \bibinfo {author} {\bibfnamefont {Mikhail~D.}\
  \bibnamefont {Lukin}},\ }\bibfield  {title} {\enquote {\bibinfo {title}
  {{Non-Abelian Floquet Spin Liquids in a Digital Rydberg Simulator}},}\ }\href
  {\doibase 10.1103/PhysRevX.13.031008} {\bibfield  {journal} {\bibinfo
  {journal} {Phys. Rev. X}\ }\textbf {\bibinfo {volume} {13}},\ \bibinfo
  {pages} {031008} (\bibinfo {year} {2023})},\ \Eprint
  {http://arxiv.org/abs/2211.00017} {arXiv:2211.00017 [quant-ph]} \BibitemShut
  {NoStop}%
\bibitem [{\citenamefont {Feynman}(1982)}]{Feynman:1981tf}%
  \BibitemOpen
  \bibfield  {author} {\bibinfo {author} {\bibfnamefont {Richard~P.}\
  \bibnamefont {Feynman}},\ }\bibfield  {title} {\enquote {\bibinfo {title}
  {{Simulating physics with computers}},}\ }\href {\doibase 10.1007/BF02650179}
  {\bibfield  {journal} {\bibinfo  {journal} {Int. J. Theor. Phys.}\ }\textbf
  {\bibinfo {volume} {21}},\ \bibinfo {pages} {467--488} (\bibinfo {year}
  {1982})}\BibitemShut {NoStop}%
\bibitem [{\citenamefont {Lloyd}(1996)}]{Lloyd:1996uni}%
  \BibitemOpen
  \bibfield  {author} {\bibinfo {author} {\bibfnamefont {Seth}\ \bibnamefont
  {Lloyd}},\ }\bibfield  {title} {\enquote {\bibinfo {title} {Universal quantum
  simulators},}\ }\href {\doibase 10.1126/science.273.5278.1073} {\bibfield
  {journal} {\bibinfo  {journal} {Science}\ }\textbf {\bibinfo {volume}
  {273}},\ \bibinfo {pages} {1073--1078} (\bibinfo {year} {1996})},\ \Eprint
  {http://arxiv.org/abs/https://www.science.org/doi/pdf/10.1126/science.273.5278.1073}
  {https://www.science.org/doi/pdf/10.1126/science.273.5278.1073} \BibitemShut
  {NoStop}%
\bibitem [{\citenamefont {Preskill}(2021)}]{Preskill:2021apy}%
  \BibitemOpen
  \bibfield  {author} {\bibinfo {author} {\bibfnamefont {John}\ \bibnamefont
  {Preskill}},\ }\bibfield  {title} {\enquote {\bibinfo {title} {{Quantum
  computing 40 years later}},}\ }\href@noop {} {\  (\bibinfo {year} {2021})},\
  \Eprint {http://arxiv.org/abs/2106.10522} {arXiv:2106.10522 [quant-ph]}
  \BibitemShut {NoStop}%
\bibitem [{\citenamefont {Carena}\ \emph {et~al.}(2022)\citenamefont {Carena},
  \citenamefont {Lamm}, \citenamefont {Li},\ and\ \citenamefont
  {Liu}}]{Carena:2022kpg}%
  \BibitemOpen
  \bibfield  {author} {\bibinfo {author} {\bibfnamefont {Marcela}\ \bibnamefont
  {Carena}}, \bibinfo {author} {\bibfnamefont {Henry}\ \bibnamefont {Lamm}},
  \bibinfo {author} {\bibfnamefont {Ying-Ying}\ \bibnamefont {Li}}, \ and\
  \bibinfo {author} {\bibfnamefont {Wanqiang}\ \bibnamefont {Liu}},\ }\bibfield
   {title} {\enquote {\bibinfo {title} {{Improved Hamiltonians for Quantum
  Simulations of Gauge Theories}},}\ }\href {\doibase
  10.1103/PhysRevLett.129.051601} {\bibfield  {journal} {\bibinfo  {journal}
  {Phys. Rev. Lett.}\ }\textbf {\bibinfo {volume} {129}},\ \bibinfo {pages}
  {051601} (\bibinfo {year} {2022})},\ \Eprint
  {http://arxiv.org/abs/2203.02823} {arXiv:2203.02823 [hep-lat]} \BibitemShut
  {NoStop}%
\bibitem [{\citenamefont {Ba\~nuls}\ \emph {et~al.}(2017)\citenamefont
  {Ba\~nuls}, \citenamefont {Cichy}, \citenamefont {Cirac}, \citenamefont
  {Jansen},\ and\ \citenamefont {K\"uhn}}]{Banuls:2017ena}%
  \BibitemOpen
  \bibfield  {author} {\bibinfo {author} {\bibfnamefont {Mari~Carmen}\
  \bibnamefont {Ba\~nuls}}, \bibinfo {author} {\bibfnamefont {Krzysztof}\
  \bibnamefont {Cichy}}, \bibinfo {author} {\bibfnamefont {J.~Ignacio}\
  \bibnamefont {Cirac}}, \bibinfo {author} {\bibfnamefont {Karl}\ \bibnamefont
  {Jansen}}, \ and\ \bibinfo {author} {\bibfnamefont {Stefan}\ \bibnamefont
  {K\"uhn}},\ }\bibfield  {title} {\enquote {\bibinfo {title} {{Efficient basis
  formulation for 1+1 dimensional SU(2) lattice gauge theory: Spectral
  calculations with matrix product states}},}\ }\href {\doibase
  10.1103/PhysRevX.7.041046} {\bibfield  {journal} {\bibinfo  {journal} {Phys.
  Rev. X}\ }\textbf {\bibinfo {volume} {7}},\ \bibinfo {pages} {041046}
  (\bibinfo {year} {2017})},\ \Eprint {http://arxiv.org/abs/1707.06434}
  {arXiv:1707.06434 [hep-lat]} \BibitemShut {NoStop}%
\bibitem [{\citenamefont {Raychowdhury}\ and\ \citenamefont
  {Stryker}(2020{\natexlab{a}})}]{Raychowdhury:2019iki}%
  \BibitemOpen
  \bibfield  {author} {\bibinfo {author} {\bibfnamefont {Indrakshi}\
  \bibnamefont {Raychowdhury}}\ and\ \bibinfo {author} {\bibfnamefont
  {Jesse~R.}\ \bibnamefont {Stryker}},\ }\bibfield  {title} {\enquote {\bibinfo
  {title} {{Loop, string, and hadron dynamics in SU(2) Hamiltonian lattice
  gauge theories}},}\ }\href {\doibase 10.1103/PhysRevD.101.114502} {\bibfield
  {journal} {\bibinfo  {journal} {Phys. Rev. D}\ }\textbf {\bibinfo {volume}
  {101}},\ \bibinfo {pages} {114502} (\bibinfo {year} {2020}{\natexlab{a}})},\
  \Eprint {http://arxiv.org/abs/1912.06133} {arXiv:1912.06133 [hep-lat]}
  \BibitemShut {NoStop}%
\bibitem [{\citenamefont {Kadam}\ \emph {et~al.}(2023)\citenamefont {Kadam},
  \citenamefont {Raychowdhury},\ and\ \citenamefont {Stryker}}]{Kadam:2022ipf}%
  \BibitemOpen
  \bibfield  {author} {\bibinfo {author} {\bibfnamefont {Saurabh~V.}\
  \bibnamefont {Kadam}}, \bibinfo {author} {\bibfnamefont {Indrakshi}\
  \bibnamefont {Raychowdhury}}, \ and\ \bibinfo {author} {\bibfnamefont
  {Jesse~R.}\ \bibnamefont {Stryker}},\ }\bibfield  {title} {\enquote {\bibinfo
  {title} {{Loop-string-hadron formulation of an SU(3) gauge theory with
  dynamical quarks}},}\ }\href {\doibase 10.1103/PhysRevD.107.094513}
  {\bibfield  {journal} {\bibinfo  {journal} {Phys. Rev. D}\ }\textbf {\bibinfo
  {volume} {107}},\ \bibinfo {pages} {094513} (\bibinfo {year} {2023})},\
  \Eprint {http://arxiv.org/abs/2212.04490} {arXiv:2212.04490 [hep-lat]}
  \BibitemShut {NoStop}%
\bibitem [{\citenamefont {Pardo}\ \emph {et~al.}(2023)\citenamefont {Pardo},
  \citenamefont {Greenberg}, \citenamefont {Fortinsky}, \citenamefont {Katz},\
  and\ \citenamefont {Zohar}}]{Pardo:2022hrp}%
  \BibitemOpen
  \bibfield  {author} {\bibinfo {author} {\bibfnamefont {Guy}\ \bibnamefont
  {Pardo}}, \bibinfo {author} {\bibfnamefont {Tomer}\ \bibnamefont
  {Greenberg}}, \bibinfo {author} {\bibfnamefont {Aryeh}\ \bibnamefont
  {Fortinsky}}, \bibinfo {author} {\bibfnamefont {Nadav}\ \bibnamefont {Katz}},
  \ and\ \bibinfo {author} {\bibfnamefont {Erez}\ \bibnamefont {Zohar}},\
  }\bibfield  {title} {\enquote {\bibinfo {title} {{Resource-efficient quantum
  simulation of lattice gauge theories in arbitrary dimensions: Solving for
  Gauss's law and fermion elimination}},}\ }\href {\doibase
  10.1103/PhysRevResearch.5.023077} {\bibfield  {journal} {\bibinfo  {journal}
  {Phys. Rev. Res.}\ }\textbf {\bibinfo {volume} {5}},\ \bibinfo {pages}
  {023077} (\bibinfo {year} {2023})},\ \Eprint
  {http://arxiv.org/abs/2206.00685} {arXiv:2206.00685 [quant-ph]} \BibitemShut
  {NoStop}%
\bibitem [{\citenamefont {Chandrasekharan}\ and\ \citenamefont
  {Wiese}(1997)}]{Chandrasekharan:1996ih}%
  \BibitemOpen
  \bibfield  {author} {\bibinfo {author} {\bibfnamefont {S.}~\bibnamefont
  {Chandrasekharan}}\ and\ \bibinfo {author} {\bibfnamefont {U.~J.}\
  \bibnamefont {Wiese}},\ }\bibfield  {title} {\enquote {\bibinfo {title}
  {{Quantum link models: A Discrete approach to gauge theories}},}\ }\href
  {\doibase 10.1016/S0550-3213(97)00006-0} {\bibfield  {journal} {\bibinfo
  {journal} {Nucl. Phys. B}\ }\textbf {\bibinfo {volume} {492}},\ \bibinfo
  {pages} {455--474} (\bibinfo {year} {1997})},\ \Eprint
  {http://arxiv.org/abs/hep-lat/9609042} {arXiv:hep-lat/9609042} \BibitemShut
  {NoStop}%
\bibitem [{\citenamefont {Bauer}\ and\ \citenamefont
  {Grabowska}(2023)}]{Bauer:2021gek}%
  \BibitemOpen
  \bibfield  {author} {\bibinfo {author} {\bibfnamefont {Christian~W.}\
  \bibnamefont {Bauer}}\ and\ \bibinfo {author} {\bibfnamefont {Dorota~M.}\
  \bibnamefont {Grabowska}},\ }\bibfield  {title} {\enquote {\bibinfo {title}
  {{Efficient representation for simulating U(1) gauge theories on digital
  quantum computers at all values of the coupling}},}\ }\href {\doibase
  10.1103/PhysRevD.107.L031503} {\bibfield  {journal} {\bibinfo  {journal}
  {Phys. Rev. D}\ }\textbf {\bibinfo {volume} {107}},\ \bibinfo {pages}
  {L031503} (\bibinfo {year} {2023})},\ \Eprint
  {http://arxiv.org/abs/2111.08015} {arXiv:2111.08015 [hep-ph]} \BibitemShut
  {NoStop}%
\bibitem [{\citenamefont {Zohar}\ and\ \citenamefont
  {Burrello}(2015)}]{Zohar:2014qma}%
  \BibitemOpen
  \bibfield  {author} {\bibinfo {author} {\bibfnamefont {Erez}\ \bibnamefont
  {Zohar}}\ and\ \bibinfo {author} {\bibfnamefont {Michele}\ \bibnamefont
  {Burrello}},\ }\bibfield  {title} {\enquote {\bibinfo {title} {{Formulation
  of lattice gauge theories for quantum simulations}},}\ }\href {\doibase
  10.1103/PhysRevD.91.054506} {\bibfield  {journal} {\bibinfo  {journal} {Phys.
  Rev. D}\ }\textbf {\bibinfo {volume} {91}},\ \bibinfo {pages} {054506}
  (\bibinfo {year} {2015})},\ \Eprint {http://arxiv.org/abs/1409.3085}
  {arXiv:1409.3085 [quant-ph]} \BibitemShut {NoStop}%
\bibitem [{\citenamefont {Zohar}\ and\ \citenamefont
  {Cirac}(2018)}]{Zohar:2018cwb}%
  \BibitemOpen
  \bibfield  {author} {\bibinfo {author} {\bibfnamefont {Erez}\ \bibnamefont
  {Zohar}}\ and\ \bibinfo {author} {\bibfnamefont {J.~Ignacio}\ \bibnamefont
  {Cirac}},\ }\bibfield  {title} {\enquote {\bibinfo {title} {{Eliminating
  fermionic matter fields in lattice gauge theories}},}\ }\href {\doibase
  10.1103/PhysRevB.98.075119} {\bibfield  {journal} {\bibinfo  {journal} {Phys.
  Rev. B}\ }\textbf {\bibinfo {volume} {98}},\ \bibinfo {pages} {075119}
  (\bibinfo {year} {2018})},\ \Eprint {http://arxiv.org/abs/1805.05347}
  {arXiv:1805.05347 [quant-ph]} \BibitemShut {NoStop}%
\bibitem [{\citenamefont {Zohar}\ and\ \citenamefont
  {Cirac}(2019)}]{Zohar:2019ygc}%
  \BibitemOpen
  \bibfield  {author} {\bibinfo {author} {\bibfnamefont {Erez}\ \bibnamefont
  {Zohar}}\ and\ \bibinfo {author} {\bibfnamefont {J.~Ignacio}\ \bibnamefont
  {Cirac}},\ }\bibfield  {title} {\enquote {\bibinfo {title} {{Removing
  Staggered Fermionic Matter in $U(N)$ and $SU(N)$ Lattice Gauge Theories}},}\
  }\href {\doibase 10.1103/PhysRevD.99.114511} {\bibfield  {journal} {\bibinfo
  {journal} {Phys. Rev. D}\ }\textbf {\bibinfo {volume} {99}},\ \bibinfo
  {pages} {114511} (\bibinfo {year} {2019})},\ \Eprint
  {http://arxiv.org/abs/1905.00652} {arXiv:1905.00652 [quant-ph]} \BibitemShut
  {NoStop}%
\bibitem [{\citenamefont {Alexandru}\ \emph {et~al.}(2024)\citenamefont
  {Alexandru}, \citenamefont {Bedaque}, \citenamefont {Carosso}, \citenamefont
  {Cervia}, \citenamefont {Murairi},\ and\ \citenamefont
  {Sheng}}]{Alexandru:2023qzd}%
  \BibitemOpen
  \bibfield  {author} {\bibinfo {author} {\bibfnamefont {Andrei}\ \bibnamefont
  {Alexandru}}, \bibinfo {author} {\bibfnamefont {Paulo~F.}\ \bibnamefont
  {Bedaque}}, \bibinfo {author} {\bibfnamefont {Andrea}\ \bibnamefont
  {Carosso}}, \bibinfo {author} {\bibfnamefont {Michael~J.}\ \bibnamefont
  {Cervia}}, \bibinfo {author} {\bibfnamefont {Edison~M.}\ \bibnamefont
  {Murairi}}, \ and\ \bibinfo {author} {\bibfnamefont {Andy}\ \bibnamefont
  {Sheng}},\ }\bibfield  {title} {\enquote {\bibinfo {title} {{Fuzzy gauge
  theory for quantum computers}},}\ }\href {\doibase
  10.1103/PhysRevD.109.094502} {\bibfield  {journal} {\bibinfo  {journal}
  {Phys. Rev. D}\ }\textbf {\bibinfo {volume} {109}},\ \bibinfo {pages}
  {094502} (\bibinfo {year} {2024})},\ \Eprint
  {http://arxiv.org/abs/2308.05253} {arXiv:2308.05253 [hep-lat]} \BibitemShut
  {NoStop}%
\bibitem [{\citenamefont {Alexandru}\ \emph {et~al.}(2019)\citenamefont
  {Alexandru}, \citenamefont {Bedaque}, \citenamefont {Harmalkar},
  \citenamefont {Lamm}, \citenamefont {Lawrence},\ and\ \citenamefont
  {Warrington}}]{Alexandru:2019nsa}%
  \BibitemOpen
  \bibfield  {author} {\bibinfo {author} {\bibfnamefont {Andrei}\ \bibnamefont
  {Alexandru}}, \bibinfo {author} {\bibfnamefont {Paulo~F.}\ \bibnamefont
  {Bedaque}}, \bibinfo {author} {\bibfnamefont {Siddhartha}\ \bibnamefont
  {Harmalkar}}, \bibinfo {author} {\bibfnamefont {Henry}\ \bibnamefont {Lamm}},
  \bibinfo {author} {\bibfnamefont {Scott}\ \bibnamefont {Lawrence}}, \ and\
  \bibinfo {author} {\bibfnamefont {Neill~C.}\ \bibnamefont {Warrington}}
  (\bibinfo {collaboration} {NuQS}),\ }\bibfield  {title} {\enquote {\bibinfo
  {title} {{Gluon Field Digitization for Quantum Computers}},}\ }\href
  {\doibase 10.1103/PhysRevD.100.114501} {\bibfield  {journal} {\bibinfo
  {journal} {Phys. Rev. D}\ }\textbf {\bibinfo {volume} {100}},\ \bibinfo
  {pages} {114501} (\bibinfo {year} {2019})},\ \Eprint
  {http://arxiv.org/abs/1906.11213} {arXiv:1906.11213 [hep-lat]} \BibitemShut
  {NoStop}%
\bibitem [{\citenamefont {Ji}\ \emph {et~al.}(2023)\citenamefont {Ji},
  \citenamefont {Lamm},\ and\ \citenamefont {Zhu}}]{Ji:2022qvr}%
  \BibitemOpen
  \bibfield  {author} {\bibinfo {author} {\bibfnamefont {Yao}\ \bibnamefont
  {Ji}}, \bibinfo {author} {\bibfnamefont {Henry}\ \bibnamefont {Lamm}}, \ and\
  \bibinfo {author} {\bibfnamefont {Shuchen}\ \bibnamefont {Zhu}} (\bibinfo
  {collaboration} {NuQS}),\ }\bibfield  {title} {\enquote {\bibinfo {title}
  {{Gluon digitization via character expansion for quantum computers}},}\
  }\href {\doibase 10.1103/PhysRevD.107.114503} {\bibfield  {journal} {\bibinfo
   {journal} {Phys. Rev. D}\ }\textbf {\bibinfo {volume} {107}},\ \bibinfo
  {pages} {114503} (\bibinfo {year} {2023})},\ \Eprint
  {http://arxiv.org/abs/2203.02330} {arXiv:2203.02330 [hep-lat]} \BibitemShut
  {NoStop}%
\bibitem [{\citenamefont {Ji}\ \emph {et~al.}(2020)\citenamefont {Ji},
  \citenamefont {Lamm},\ and\ \citenamefont {Zhu}}]{Ji:2020kjk}%
  \BibitemOpen
  \bibfield  {author} {\bibinfo {author} {\bibfnamefont {Yao}\ \bibnamefont
  {Ji}}, \bibinfo {author} {\bibfnamefont {Henry}\ \bibnamefont {Lamm}}, \ and\
  \bibinfo {author} {\bibfnamefont {Shuchen}\ \bibnamefont {Zhu}} (\bibinfo
  {collaboration} {NuQS}),\ }\bibfield  {title} {\enquote {\bibinfo {title}
  {{Gluon Field Digitization via Group Space Decimation for Quantum
  Computers}},}\ }\href {\doibase 10.1103/PhysRevD.102.114513} {\bibfield
  {journal} {\bibinfo  {journal} {Phys. Rev. D}\ }\textbf {\bibinfo {volume}
  {102}},\ \bibinfo {pages} {114513} (\bibinfo {year} {2020})},\ \Eprint
  {http://arxiv.org/abs/2005.14221} {arXiv:2005.14221 [hep-lat]} \BibitemShut
  {NoStop}%
\bibitem [{\citenamefont {Kavaki}\ and\ \citenamefont
  {Lewis}(2024)}]{Kavaki:2024ijd}%
  \BibitemOpen
  \bibfield  {author} {\bibinfo {author} {\bibfnamefont {Ali H.~Z.}\
  \bibnamefont {Kavaki}}\ and\ \bibinfo {author} {\bibfnamefont {Randy}\
  \bibnamefont {Lewis}},\ }\bibfield  {title} {\enquote {\bibinfo {title}
  {{From square plaquettes to triamond lattices for SU(2) gauge theory}},}\
  }\href {\doibase 10.1038/s42005-024-01697-4} {\bibfield  {journal} {\bibinfo
  {journal} {Commun. Phys.}\ }\textbf {\bibinfo {volume} {7}},\ \bibinfo
  {pages} {208} (\bibinfo {year} {2024})},\ \Eprint
  {http://arxiv.org/abs/2401.14570} {arXiv:2401.14570 [hep-lat]} \BibitemShut
  {NoStop}%
\bibitem [{\citenamefont {D'Andrea}\ \emph {et~al.}(2024)\citenamefont
  {D'Andrea}, \citenamefont {Bauer}, \citenamefont {Grabowska},\ and\
  \citenamefont {Freytsis}}]{DAndrea:2023qnr}%
  \BibitemOpen
  \bibfield  {author} {\bibinfo {author} {\bibfnamefont {Irian}\ \bibnamefont
  {D'Andrea}}, \bibinfo {author} {\bibfnamefont {Christian~W.}\ \bibnamefont
  {Bauer}}, \bibinfo {author} {\bibfnamefont {Dorota~M.}\ \bibnamefont
  {Grabowska}}, \ and\ \bibinfo {author} {\bibfnamefont {Marat}\ \bibnamefont
  {Freytsis}},\ }\bibfield  {title} {\enquote {\bibinfo {title} {{New basis for
  Hamiltonian SU(2) simulations}},}\ }\href {\doibase
  10.1103/PhysRevD.109.074501} {\bibfield  {journal} {\bibinfo  {journal}
  {Phys. Rev. D}\ }\textbf {\bibinfo {volume} {109}},\ \bibinfo {pages}
  {074501} (\bibinfo {year} {2024})},\ \Eprint
  {http://arxiv.org/abs/2307.11829} {arXiv:2307.11829 [hep-ph]} \BibitemShut
  {NoStop}%
\bibitem [{\citenamefont {Romiti}\ and\ \citenamefont
  {Urbach}(2024)}]{Romiti:2023hbd}%
  \BibitemOpen
  \bibfield  {author} {\bibinfo {author} {\bibfnamefont {Simone}\ \bibnamefont
  {Romiti}}\ and\ \bibinfo {author} {\bibfnamefont {Carsten}\ \bibnamefont
  {Urbach}},\ }\bibfield  {title} {\enquote {\bibinfo {title} {{Digitizing
  lattice gauge theories in the magnetic basis: reducing the breaking of the
  fundamental commutation relations}},}\ }\href {\doibase
  10.1140/epjc/s10052-024-13037-5} {\bibfield  {journal} {\bibinfo  {journal}
  {Eur. Phys. J. C}\ }\textbf {\bibinfo {volume} {84}},\ \bibinfo {pages} {708}
  (\bibinfo {year} {2024})},\ \Eprint {http://arxiv.org/abs/2311.11928}
  {arXiv:2311.11928 [hep-lat]} \BibitemShut {NoStop}%
\bibitem [{\citenamefont {Zache}\ \emph
  {et~al.}(2023{\natexlab{b}})\citenamefont {Zache}, \citenamefont
  {Gonz\'alez-Cuadra},\ and\ \citenamefont {Zoller}}]{Zache:2023dko}%
  \BibitemOpen
  \bibfield  {author} {\bibinfo {author} {\bibfnamefont {Torsten~V.}\
  \bibnamefont {Zache}}, \bibinfo {author} {\bibfnamefont {Daniel}\
  \bibnamefont {Gonz\'alez-Cuadra}}, \ and\ \bibinfo {author} {\bibfnamefont
  {Peter}\ \bibnamefont {Zoller}},\ }\bibfield  {title} {\enquote {\bibinfo
  {title} {{Quantum and Classical Spin-Network Algorithms for q-Deformed
  Kogut-Susskind Gauge Theories}},}\ }\href {\doibase
  10.1103/PhysRevLett.131.171902} {\bibfield  {journal} {\bibinfo  {journal}
  {Phys. Rev. Lett.}\ }\textbf {\bibinfo {volume} {131}},\ \bibinfo {pages}
  {171902} (\bibinfo {year} {2023}{\natexlab{b}})},\ \Eprint
  {http://arxiv.org/abs/2304.02527} {arXiv:2304.02527 [quant-ph]} \BibitemShut
  {NoStop}%
\bibitem [{\citenamefont {Wiese}(2021)}]{Wiese:2021djl}%
  \BibitemOpen
  \bibfield  {author} {\bibinfo {author} {\bibfnamefont {Uwe-Jens}\
  \bibnamefont {Wiese}},\ }\bibfield  {title} {\enquote {\bibinfo {title}
  {{From quantum link models to D-theory: a resource efficient framework for
  the quantum simulation and computation of gauge theories}},}\ }\href
  {\doibase 10.1098/rsta.2021.0068} {\bibfield  {journal} {\bibinfo  {journal}
  {Phil. Trans. A. Math. Phys. Eng. Sci.}\ }\textbf {\bibinfo {volume} {380}},\
  \bibinfo {pages} {20210068} (\bibinfo {year} {2021})},\ \Eprint
  {http://arxiv.org/abs/2107.09335} {arXiv:2107.09335 [hep-lat]} \BibitemShut
  {NoStop}%
\bibitem [{\citenamefont {Kaplan}\ and\ \citenamefont
  {Stryker}(2020)}]{Kaplan:2018vnj}%
  \BibitemOpen
  \bibfield  {author} {\bibinfo {author} {\bibfnamefont {David~B.}\
  \bibnamefont {Kaplan}}\ and\ \bibinfo {author} {\bibfnamefont {Jesse~R.}\
  \bibnamefont {Stryker}},\ }\bibfield  {title} {\enquote {\bibinfo {title}
  {{Gauss\textquoteright{}s law, duality, and the Hamiltonian formulation of
  U(1) lattice gauge theory}},}\ }\href {\doibase 10.1103/PhysRevD.102.094515}
  {\bibfield  {journal} {\bibinfo  {journal} {Phys. Rev. D}\ }\textbf {\bibinfo
  {volume} {102}},\ \bibinfo {pages} {094515} (\bibinfo {year} {2020})},\
  \Eprint {http://arxiv.org/abs/1806.08797} {arXiv:1806.08797 [hep-lat]}
  \BibitemShut {NoStop}%
\bibitem [{\citenamefont {Mathur}\ and\ \citenamefont
  {Sreeraj}(2015)}]{Mathur:2015wba}%
  \BibitemOpen
  \bibfield  {author} {\bibinfo {author} {\bibfnamefont {Manu}\ \bibnamefont
  {Mathur}}\ and\ \bibinfo {author} {\bibfnamefont {T.~P.}\ \bibnamefont
  {Sreeraj}},\ }\bibfield  {title} {\enquote {\bibinfo {title} {{Canonical
  Transformations and Loop Formulation of SU(N) Lattice Gauge Theories}},}\
  }\href {\doibase 10.1103/PhysRevD.92.125018} {\bibfield  {journal} {\bibinfo
  {journal} {Phys. Rev. D}\ }\textbf {\bibinfo {volume} {92}},\ \bibinfo
  {pages} {125018} (\bibinfo {year} {2015})},\ \Eprint
  {http://arxiv.org/abs/1509.04033} {arXiv:1509.04033 [hep-lat]} \BibitemShut
  {NoStop}%
\bibitem [{\citenamefont {Mathur}\ and\ \citenamefont
  {Rathor}(2023)}]{Mathur:2023lky}%
  \BibitemOpen
  \bibfield  {author} {\bibinfo {author} {\bibfnamefont {Manu}\ \bibnamefont
  {Mathur}}\ and\ \bibinfo {author} {\bibfnamefont {Atul}\ \bibnamefont
  {Rathor}},\ }\bibfield  {title} {\enquote {\bibinfo {title} {{Disorder
  operators and magnetic vortices in SU(N) lattice gauge theory}},}\ }\href
  {\doibase 10.1103/PhysRevD.108.114507} {\bibfield  {journal} {\bibinfo
  {journal} {Phys. Rev. D}\ }\textbf {\bibinfo {volume} {108}},\ \bibinfo
  {pages} {114507} (\bibinfo {year} {2023})},\ \Eprint
  {http://arxiv.org/abs/2307.06278} {arXiv:2307.06278 [hep-lat]} \BibitemShut
  {NoStop}%
\bibitem [{\citenamefont {Fontana}\ \emph {et~al.}(2024)\citenamefont
  {Fontana}, \citenamefont {Riaza},\ and\ \citenamefont
  {Celi}}]{Fontana:2024rux}%
  \BibitemOpen
  \bibfield  {author} {\bibinfo {author} {\bibfnamefont {Pierpaolo}\
  \bibnamefont {Fontana}}, \bibinfo {author} {\bibfnamefont {Marc~Miranda}\
  \bibnamefont {Riaza}}, \ and\ \bibinfo {author} {\bibfnamefont {Alessio}\
  \bibnamefont {Celi}},\ }\bibfield  {title} {\enquote {\bibinfo {title} {{An
  efficient finite-resource formulation of non-Abelian lattice gauge theories
  beyond one dimension}},}\ }\href@noop {} {\  (\bibinfo {year} {2024})},\
  \Eprint {http://arxiv.org/abs/2409.04441} {arXiv:2409.04441 [quant-ph]}
  \BibitemShut {NoStop}%
\bibitem [{\citenamefont {Grabowska}\ \emph {et~al.}(2024)\citenamefont
  {Grabowska}, \citenamefont {Kane},\ and\ \citenamefont
  {Bauer}}]{Grabowska:2024emw}%
  \BibitemOpen
  \bibfield  {author} {\bibinfo {author} {\bibfnamefont {Dorota~M.}\
  \bibnamefont {Grabowska}}, \bibinfo {author} {\bibfnamefont {Christopher~F.}\
  \bibnamefont {Kane}}, \ and\ \bibinfo {author} {\bibfnamefont {Christian~W.}\
  \bibnamefont {Bauer}},\ }\bibfield  {title} {\enquote {\bibinfo {title} {{A
  Fully Gauge-Fixed SU(2) Hamiltonian for Quantum Simulations}},}\ }\href@noop
  {} {\  (\bibinfo {year} {2024})},\ \Eprint {http://arxiv.org/abs/2409.10610}
  {arXiv:2409.10610 [quant-ph]} \BibitemShut {NoStop}%
\bibitem [{\citenamefont {Hartung}\ \emph {et~al.}(2022)\citenamefont
  {Hartung}, \citenamefont {Jakobs}, \citenamefont {Jansen}, \citenamefont
  {Ostmeyer},\ and\ \citenamefont {Urbach}}]{Hartung:2022hoz}%
  \BibitemOpen
  \bibfield  {author} {\bibinfo {author} {\bibfnamefont {Tobias}\ \bibnamefont
  {Hartung}}, \bibinfo {author} {\bibfnamefont {Timo}\ \bibnamefont {Jakobs}},
  \bibinfo {author} {\bibfnamefont {Karl}\ \bibnamefont {Jansen}}, \bibinfo
  {author} {\bibfnamefont {Johann}\ \bibnamefont {Ostmeyer}}, \ and\ \bibinfo
  {author} {\bibfnamefont {Carsten}\ \bibnamefont {Urbach}},\ }\bibfield
  {title} {\enquote {\bibinfo {title} {{Digitising SU(2) gauge fields and the
  freezing transition}},}\ }\href {\doibase 10.1140/epjc/s10052-022-10192-5}
  {\bibfield  {journal} {\bibinfo  {journal} {Eur. Phys. J. C}\ }\textbf
  {\bibinfo {volume} {82}},\ \bibinfo {pages} {237} (\bibinfo {year} {2022})},\
  \Eprint {http://arxiv.org/abs/2201.09625} {arXiv:2201.09625 [hep-lat]}
  \BibitemShut {NoStop}%
\bibitem [{\citenamefont {M\"uller}\ and\ \citenamefont
  {Yao}(2023)}]{Muller:2023nnk}%
  \BibitemOpen
  \bibfield  {author} {\bibinfo {author} {\bibfnamefont {Berndt}\ \bibnamefont
  {M\"uller}}\ and\ \bibinfo {author} {\bibfnamefont {Xiaojun}\ \bibnamefont
  {Yao}},\ }\bibfield  {title} {\enquote {\bibinfo {title} {{Simple Hamiltonian
  for quantum simulation of strongly coupled (2+1)D SU(2) lattice gauge theory
  on a honeycomb lattice}},}\ }\href {\doibase 10.1103/PhysRevD.108.094505}
  {\bibfield  {journal} {\bibinfo  {journal} {Phys. Rev. D}\ }\textbf {\bibinfo
  {volume} {108}},\ \bibinfo {pages} {094505} (\bibinfo {year} {2023})},\
  \Eprint {http://arxiv.org/abs/2307.00045} {arXiv:2307.00045 [quant-ph]}
  \BibitemShut {NoStop}%
\bibitem [{\citenamefont {Farrell}\ \emph
  {et~al.}(2024{\natexlab{b}})\citenamefont {Farrell}, \citenamefont {Illa},
  \citenamefont {Ciavarella},\ and\ \citenamefont {Savage}}]{Farrell:2023fgd}%
  \BibitemOpen
  \bibfield  {author} {\bibinfo {author} {\bibfnamefont {Roland~C.}\
  \bibnamefont {Farrell}}, \bibinfo {author} {\bibfnamefont {Marc}\
  \bibnamefont {Illa}}, \bibinfo {author} {\bibfnamefont {Anthony~N.}\
  \bibnamefont {Ciavarella}}, \ and\ \bibinfo {author} {\bibfnamefont
  {Martin~J.}\ \bibnamefont {Savage}},\ }\bibfield  {title} {\enquote {\bibinfo
  {title} {{Scalable Circuits for Preparing Ground States on Digital Quantum
  Computers: The Schwinger Model Vacuum on 100 Qubits}},}\ }\href {\doibase
  10.1103/PRXQuantum.5.020315} {\bibfield  {journal} {\bibinfo  {journal} {PRX
  Quantum}\ }\textbf {\bibinfo {volume} {5}},\ \bibinfo {pages} {020315}
  (\bibinfo {year} {2024}{\natexlab{b}})},\ \Eprint
  {http://arxiv.org/abs/2308.04481} {arXiv:2308.04481 [quant-ph]} \BibitemShut
  {NoStop}%
\bibitem [{\citenamefont {Stryker}(2021)}]{Stryker:2021asy}%
  \BibitemOpen
  \bibfield  {author} {\bibinfo {author} {\bibfnamefont {Jesse~R.}\
  \bibnamefont {Stryker}},\ }\bibfield  {title} {\enquote {\bibinfo {title}
  {{Shearing approach to gauge invariant Trotterization}},}\ }\href@noop {} {\
  (\bibinfo {year} {2021})},\ \Eprint {http://arxiv.org/abs/2105.11548}
  {arXiv:2105.11548 [hep-lat]} \BibitemShut {NoStop}%
\bibitem [{\citenamefont {Gustafson}\ and\ \citenamefont
  {Lamm}(2023)}]{Gustafson:2023swx}%
  \BibitemOpen
  \bibfield  {author} {\bibinfo {author} {\bibfnamefont {Erik~J.}\ \bibnamefont
  {Gustafson}}\ and\ \bibinfo {author} {\bibfnamefont {Henry}\ \bibnamefont
  {Lamm}},\ }\bibfield  {title} {\enquote {\bibinfo {title} {{Robustness of
  Gauge Digitization to Quantum Noise}},}\ }\href@noop {} {\  (\bibinfo {year}
  {2023})},\ \Eprint {http://arxiv.org/abs/2301.10207} {arXiv:2301.10207
  [hep-lat]} \BibitemShut {NoStop}%
\bibitem [{\citenamefont {Lamm}\ \emph
  {et~al.}(2020{\natexlab{b}})\citenamefont {Lamm}, \citenamefont {Lawrence},\
  and\ \citenamefont {Yamauchi}}]{Lamm:2020jwv}%
  \BibitemOpen
  \bibfield  {author} {\bibinfo {author} {\bibfnamefont {Henry}\ \bibnamefont
  {Lamm}}, \bibinfo {author} {\bibfnamefont {Scott}\ \bibnamefont {Lawrence}},
  \ and\ \bibinfo {author} {\bibfnamefont {Yukari}\ \bibnamefont {Yamauchi}}
  (\bibinfo {collaboration} {NuQS}),\ }\bibfield  {title} {\enquote {\bibinfo
  {title} {{Suppressing Coherent Gauge Drift in Quantum Simulations}},}\
  }\href@noop {} {\  (\bibinfo {year} {2020}{\natexlab{b}})},\ \Eprint
  {http://arxiv.org/abs/2005.12688} {arXiv:2005.12688 [quant-ph]} \BibitemShut
  {NoStop}%
\bibitem [{\citenamefont {Raychowdhury}\ and\ \citenamefont
  {Stryker}(2020{\natexlab{b}})}]{Raychowdhury:2018osk}%
  \BibitemOpen
  \bibfield  {author} {\bibinfo {author} {\bibfnamefont {Indrakshi}\
  \bibnamefont {Raychowdhury}}\ and\ \bibinfo {author} {\bibfnamefont
  {Jesse~R.}\ \bibnamefont {Stryker}},\ }\bibfield  {title} {\enquote {\bibinfo
  {title} {{Solving Gauss's Law on Digital Quantum Computers with
  Loop-String-Hadron Digitization}},}\ }\href {\doibase
  10.1103/PhysRevResearch.2.033039} {\bibfield  {journal} {\bibinfo  {journal}
  {Phys. Rev. Res.}\ }\textbf {\bibinfo {volume} {2}},\ \bibinfo {pages}
  {033039} (\bibinfo {year} {2020}{\natexlab{b}})},\ \Eprint
  {http://arxiv.org/abs/1812.07554} {arXiv:1812.07554 [hep-lat]} \BibitemShut
  {NoStop}%
\bibitem [{\citenamefont {Mathew}\ and\ \citenamefont
  {Raychowdhury}(2022)}]{Mathew:2022nep}%
  \BibitemOpen
  \bibfield  {author} {\bibinfo {author} {\bibfnamefont {Emil}\ \bibnamefont
  {Mathew}}\ and\ \bibinfo {author} {\bibfnamefont {Indrakshi}\ \bibnamefont
  {Raychowdhury}},\ }\bibfield  {title} {\enquote {\bibinfo {title}
  {{Protecting local and global symmetries in simulating (1+1)D non-Abelian
  gauge theories}},}\ }\href {\doibase 10.1103/PhysRevD.106.054510} {\bibfield
  {journal} {\bibinfo  {journal} {Phys. Rev. D}\ }\textbf {\bibinfo {volume}
  {106}},\ \bibinfo {pages} {054510} (\bibinfo {year} {2022})},\ \Eprint
  {http://arxiv.org/abs/2206.07444} {arXiv:2206.07444 [hep-lat]} \BibitemShut
  {NoStop}%
\bibitem [{\citenamefont {Mathew}\ and\ \citenamefont
  {Raychowdhury}(2024)}]{Mathew:2024bed}%
  \BibitemOpen
  \bibfield  {author} {\bibinfo {author} {\bibfnamefont {Emil}\ \bibnamefont
  {Mathew}}\ and\ \bibinfo {author} {\bibfnamefont {Indrakshi}\ \bibnamefont
  {Raychowdhury}},\ }\bibfield  {title} {\enquote {\bibinfo {title}
  {{Protecting gauge symmetries in the the dynamics of SU(3) lattice gauge
  theories}},}\ }\href {\doibase 10.48550/arXiv.2404.12158} {\  (\bibinfo
  {year} {2024}),\ 10.48550/arXiv.2404.12158},\ \Eprint
  {http://arxiv.org/abs/2404.12158} {arXiv:2404.12158 [hep-lat]} \BibitemShut
  {NoStop}%
\bibitem [{\citenamefont {Grabowska}\ \emph {et~al.}(2022)\citenamefont
  {Grabowska}, \citenamefont {Kane}, \citenamefont {Nachman},\ and\
  \citenamefont {Bauer}}]{Grabowska:2022uos}%
  \BibitemOpen
  \bibfield  {author} {\bibinfo {author} {\bibfnamefont {Dorota~M.}\
  \bibnamefont {Grabowska}}, \bibinfo {author} {\bibfnamefont {Christopher}\
  \bibnamefont {Kane}}, \bibinfo {author} {\bibfnamefont {Benjamin}\
  \bibnamefont {Nachman}}, \ and\ \bibinfo {author} {\bibfnamefont
  {Christian~W.}\ \bibnamefont {Bauer}},\ }\bibfield  {title} {\enquote
  {\bibinfo {title} {{Overcoming exponential scaling with system size in
  Trotter-Suzuki implementations of constrained Hamiltonians: 2+1 U(1) lattice
  gauge theories}},}\ }\href@noop {} {\  (\bibinfo {year} {2022})},\ \Eprint
  {http://arxiv.org/abs/2208.03333} {arXiv:2208.03333 [quant-ph]} \BibitemShut
  {NoStop}%
\bibitem [{\citenamefont {Kogut}\ and\ \citenamefont
  {Susskind}(1975)}]{Kogut:1974ag}%
  \BibitemOpen
  \bibfield  {author} {\bibinfo {author} {\bibfnamefont {John~B.}\ \bibnamefont
  {Kogut}}\ and\ \bibinfo {author} {\bibfnamefont {Leonard}\ \bibnamefont
  {Susskind}},\ }\bibfield  {title} {\enquote {\bibinfo {title} {{Hamiltonian
  Formulation of Wilson's Lattice Gauge Theories}},}\ }\href {\doibase
  10.1103/PhysRevD.11.395} {\bibfield  {journal} {\bibinfo  {journal} {Phys.
  Rev. D}\ }\textbf {\bibinfo {volume} {11}},\ \bibinfo {pages} {395--408}
  (\bibinfo {year} {1975})}\BibitemShut {NoStop}%
\bibitem [{\citenamefont {Davoudi}\ \emph {et~al.}(2021)\citenamefont
  {Davoudi}, \citenamefont {Raychowdhury},\ and\ \citenamefont
  {Shaw}}]{Davoudi:2020yln}%
  \BibitemOpen
  \bibfield  {author} {\bibinfo {author} {\bibfnamefont {Zohreh}\ \bibnamefont
  {Davoudi}}, \bibinfo {author} {\bibfnamefont {Indrakshi}\ \bibnamefont
  {Raychowdhury}}, \ and\ \bibinfo {author} {\bibfnamefont {Andrew}\
  \bibnamefont {Shaw}},\ }\bibfield  {title} {\enquote {\bibinfo {title}
  {{Search for efficient formulations for Hamiltonian simulation of non-Abelian
  lattice gauge theories}},}\ }\href {\doibase 10.1103/PhysRevD.104.074505}
  {\bibfield  {journal} {\bibinfo  {journal} {Phys. Rev. D}\ }\textbf {\bibinfo
  {volume} {104}},\ \bibinfo {pages} {074505} (\bibinfo {year} {2021})},\
  \Eprint {http://arxiv.org/abs/2009.11802} {arXiv:2009.11802 [hep-lat]}
  \BibitemShut {NoStop}%
\bibitem [{\citenamefont {Mandelstam}(1962)}]{Mandelstam:1962us}%
  \BibitemOpen
  \bibfield  {author} {\bibinfo {author} {\bibfnamefont {Stanley}\ \bibnamefont
  {Mandelstam}},\ }\bibfield  {title} {\enquote {\bibinfo {title}
  {{Quantization of the gravitational field}},}\ }\href {\doibase
  10.1016/0003-4916(62)90233-6} {\bibfield  {journal} {\bibinfo  {journal}
  {Annals Phys.}\ }\textbf {\bibinfo {volume} {19}},\ \bibinfo {pages} {25--66}
  (\bibinfo {year} {1962})}\BibitemShut {NoStop}%
\bibitem [{\citenamefont {Mandelstam}(1968)}]{Mandelstam:1968hz}%
  \BibitemOpen
  \bibfield  {author} {\bibinfo {author} {\bibfnamefont {Stanley}\ \bibnamefont
  {Mandelstam}},\ }\bibfield  {title} {\enquote {\bibinfo {title} {{Feynman
  rules for electromagnetic and Yang-Mills fields from the gauge independent
  field theoretic formalism}},}\ }\href {\doibase 10.1103/PhysRev.175.1580}
  {\bibfield  {journal} {\bibinfo  {journal} {Phys. Rev.}\ }\textbf {\bibinfo
  {volume} {175}},\ \bibinfo {pages} {1580--1623} (\bibinfo {year}
  {1968})}\BibitemShut {NoStop}%
\bibitem [{\citenamefont {Mandelstam}(1979)}]{Mandelstam:1978ed}%
  \BibitemOpen
  \bibfield  {author} {\bibinfo {author} {\bibfnamefont {S.}~\bibnamefont
  {Mandelstam}},\ }\bibfield  {title} {\enquote {\bibinfo {title} {{Charge -
  Monopole Duality and the Phases of Nonabelian Gauge Theories}},}\ }\href
  {\doibase 10.1103/PhysRevD.19.2391} {\bibfield  {journal} {\bibinfo
  {journal} {Phys. Rev. D}\ }\textbf {\bibinfo {volume} {19}},\ \bibinfo
  {pages} {2391} (\bibinfo {year} {1979})}\BibitemShut {NoStop}%
\bibitem [{\citenamefont {Mathur}(2007)}]{Mathur:2007nu}%
  \BibitemOpen
  \bibfield  {author} {\bibinfo {author} {\bibfnamefont {Manu}\ \bibnamefont
  {Mathur}},\ }\bibfield  {title} {\enquote {\bibinfo {title} {{Loop Approach
  to Lattice Gauge Theories}},}\ }\href {\doibase
  10.1016/j.nuclphysb.2007.04.031} {\bibfield  {journal} {\bibinfo  {journal}
  {Nucl. Phys. B}\ }\textbf {\bibinfo {volume} {779}},\ \bibinfo {pages}
  {32--62} (\bibinfo {year} {2007})},\ \Eprint
  {http://arxiv.org/abs/hep-lat/0702007} {arXiv:hep-lat/0702007} \BibitemShut
  {NoStop}%
\bibitem [{\citenamefont {Anishetty}\ \emph {et~al.}(2010)\citenamefont
  {Anishetty}, \citenamefont {Mathur},\ and\ \citenamefont
  {Raychowdhury}}]{Anishetty:2009nh}%
  \BibitemOpen
  \bibfield  {author} {\bibinfo {author} {\bibfnamefont {Ramesh}\ \bibnamefont
  {Anishetty}}, \bibinfo {author} {\bibfnamefont {Manu}\ \bibnamefont
  {Mathur}}, \ and\ \bibinfo {author} {\bibfnamefont {Indrakshi}\ \bibnamefont
  {Raychowdhury}},\ }\bibfield  {title} {\enquote {\bibinfo {title}
  {{Prepotential formulation of SU(3) lattice gauge theory}},}\ }\href
  {\doibase 10.1088/1751-8113/43/3/035403} {\bibfield  {journal} {\bibinfo
  {journal} {J. Phys. A}\ }\textbf {\bibinfo {volume} {43}},\ \bibinfo {pages}
  {035403} (\bibinfo {year} {2010})},\ \Eprint {http://arxiv.org/abs/0909.2394}
  {arXiv:0909.2394 [hep-lat]} \BibitemShut {NoStop}%
\bibitem [{\citenamefont {Mathur}(2005)}]{Mathur:2004kr}%
  \BibitemOpen
  \bibfield  {author} {\bibinfo {author} {\bibfnamefont {Manu}\ \bibnamefont
  {Mathur}},\ }\bibfield  {title} {\enquote {\bibinfo {title} {{Harmonic
  oscillator prepotentials in SU(2) lattice gauge theory}},}\ }\href {\doibase
  10.1088/0305-4470/38/46/008} {\bibfield  {journal} {\bibinfo  {journal} {J.
  Phys. A}\ }\textbf {\bibinfo {volume} {38}},\ \bibinfo {pages} {10015--10026}
  (\bibinfo {year} {2005})},\ \Eprint {http://arxiv.org/abs/hep-lat/0403029}
  {arXiv:hep-lat/0403029} \BibitemShut {NoStop}%
\bibitem [{\citenamefont {Mathur}\ \emph {et~al.}(2010)\citenamefont {Mathur},
  \citenamefont {Raychowdhury},\ and\ \citenamefont
  {Anishetty}}]{Mathur:2010wc}%
  \BibitemOpen
  \bibfield  {author} {\bibinfo {author} {\bibfnamefont {Manu}\ \bibnamefont
  {Mathur}}, \bibinfo {author} {\bibfnamefont {Indrakshi}\ \bibnamefont
  {Raychowdhury}}, \ and\ \bibinfo {author} {\bibfnamefont {Ramesh}\
  \bibnamefont {Anishetty}},\ }\bibfield  {title} {\enquote {\bibinfo {title}
  {{SU(N) Irreducible Schwinger Bosons}},}\ }\href {\doibase 10.1063/1.3464267}
  {\bibfield  {journal} {\bibinfo  {journal} {J. Math. Phys.}\ }\textbf
  {\bibinfo {volume} {51}},\ \bibinfo {pages} {093504} (\bibinfo {year}
  {2010})},\ \Eprint {http://arxiv.org/abs/1003.5487} {arXiv:1003.5487
  [math-ph]} \BibitemShut {NoStop}%
\bibitem [{\citenamefont {Mathur}\ and\ \citenamefont
  {Sen}(2001)}]{Mathur:2000sv}%
  \BibitemOpen
  \bibfield  {author} {\bibinfo {author} {\bibfnamefont {Manu}\ \bibnamefont
  {Mathur}}\ and\ \bibinfo {author} {\bibfnamefont {Diptiman}\ \bibnamefont
  {Sen}},\ }\bibfield  {title} {\enquote {\bibinfo {title} {{Coherent states
  for SU(3)}},}\ }\href {\doibase 10.1063/1.1385563} {\bibfield  {journal}
  {\bibinfo  {journal} {J. Math. Phys.}\ }\textbf {\bibinfo {volume} {42}},\
  \bibinfo {pages} {4181--4196} (\bibinfo {year} {2001})},\ \Eprint
  {http://arxiv.org/abs/quant-ph/0012099} {arXiv:quant-ph/0012099} \BibitemShut
  {NoStop}%
\bibitem [{\citenamefont {Raychowdhury}(2013)}]{Raychowdhury:2013rwa}%
  \BibitemOpen
  \bibfield  {author} {\bibinfo {author} {\bibfnamefont {Indrakshi}\
  \bibnamefont {Raychowdhury}},\ }\emph {\bibinfo {title} {{Prepotential
  Formulation of Lattice Gauge Theories}}},\ \href
  {https://www.bose.res.in/linked-objects/academic-programmes/PhD%20Thesies/2013/Indrakshi%20Raychowdhury_thesis.pdf}
  {Ph.D. thesis},\ \bibinfo  {school} {Calcutta U.} (\bibinfo {year}
  {2013})\BibitemShut {NoStop}%
\bibitem [{\citenamefont {Anishetty}\ and\ \citenamefont
  {Sreeraj}(2019)}]{Anishetty:2019xge}%
  \BibitemOpen
  \bibfield  {author} {\bibinfo {author} {\bibfnamefont {Ramesh}\ \bibnamefont
  {Anishetty}}\ and\ \bibinfo {author} {\bibfnamefont {T.~P.}\ \bibnamefont
  {Sreeraj}},\ }\bibfield  {title} {\enquote {\bibinfo {title} {{Addition of
  SU(3) generators and its Singlet Hilbert space}},}\ }\href {\doibase
  10.1063/1.5096613} {\bibfield  {journal} {\bibinfo  {journal} {J. Math.
  Phys.}\ }\textbf {\bibinfo {volume} {60}},\ \bibinfo {pages} {061701}
  (\bibinfo {year} {2019})},\ \Eprint {http://arxiv.org/abs/1903.07956}
  {arXiv:1903.07956 [math-ph]} \BibitemShut {NoStop}%
\bibitem [{\citenamefont {Raychowdhury}(2019)}]{Raychowdhury:2018tfj}%
  \BibitemOpen
  \bibfield  {author} {\bibinfo {author} {\bibfnamefont {Indrakshi}\
  \bibnamefont {Raychowdhury}},\ }\bibfield  {title} {\enquote {\bibinfo
  {title} {{Low energy spectrum of SU(2) lattice gauge theory}: {An alternate
  proposal via loop formulation}},}\ }\href {\doibase
  10.1140/epjc/s10052-019-6753-0} {\bibfield  {journal} {\bibinfo  {journal}
  {Eur. Phys. J. C}\ }\textbf {\bibinfo {volume} {79}},\ \bibinfo {pages} {235}
  (\bibinfo {year} {2019})},\ \Eprint {http://arxiv.org/abs/1804.01304}
  {arXiv:1804.01304 [hep-lat]} \BibitemShut {NoStop}%
\bibitem [{\citenamefont {Anishetty}\ \emph {et~al.}(2009)\citenamefont
  {Anishetty}, \citenamefont {Mathur},\ and\ \citenamefont
  {Raychowdhury}}]{Anishetty:2009ai}%
  \BibitemOpen
  \bibfield  {author} {\bibinfo {author} {\bibfnamefont {Ramesh}\ \bibnamefont
  {Anishetty}}, \bibinfo {author} {\bibfnamefont {Manu}\ \bibnamefont
  {Mathur}}, \ and\ \bibinfo {author} {\bibfnamefont {Indrakshi}\ \bibnamefont
  {Raychowdhury}},\ }\bibfield  {title} {\enquote {\bibinfo {title}
  {{Irreducible SU(3) Schhwinger Bosons}},}\ }\href {\doibase
  10.1063/1.3122666} {\bibfield  {journal} {\bibinfo  {journal} {J. Math.
  Phys.}\ }\textbf {\bibinfo {volume} {50}},\ \bibinfo {pages} {053503}
  (\bibinfo {year} {2009})},\ \Eprint {http://arxiv.org/abs/0901.0644}
  {arXiv:0901.0644 [math-ph]} \BibitemShut {NoStop}%
\bibitem [{\citenamefont {Mukunda}\ and\ \citenamefont
  {Pandit}(1965)}]{Mukunda1965TensorMA}%
  \BibitemOpen
  \bibfield  {author} {\bibinfo {author} {\bibfnamefont {N.}~\bibnamefont
  {Mukunda}}\ and\ \bibinfo {author} {\bibfnamefont {L.~K.}\ \bibnamefont
  {Pandit}},\ }\bibfield  {title} {\enquote {\bibinfo {title} {Tensor methods
  and a unified representation theory of su3},}\ }\href
  {https://api.semanticscholar.org/CorpusID:122585563} {\bibfield  {journal}
  {\bibinfo  {journal} {Journal of Mathematical Physics}\ }\textbf {\bibinfo
  {volume} {6}},\ \bibinfo {pages} {746--765} (\bibinfo {year}
  {1965})}\BibitemShut {NoStop}%
\bibitem [{\citenamefont {Chaturvedi}\ and\ \citenamefont
  {Mukunda}(2002)}]{Chaturvedi:2002si}%
  \BibitemOpen
  \bibfield  {author} {\bibinfo {author} {\bibfnamefont {S.}~\bibnamefont
  {Chaturvedi}}\ and\ \bibinfo {author} {\bibfnamefont {N.}~\bibnamefont
  {Mukunda}},\ }\bibfield  {title} {\enquote {\bibinfo {title} {{The Schwinger
  SU(3) construction. 1. Multiplicity problem and relation to induced
  representations}},}\ }\href {\doibase 10.1063/1.1508810} {\bibfield
  {journal} {\bibinfo  {journal} {J. Math. Phys.}\ }\textbf {\bibinfo {volume}
  {43}},\ \bibinfo {pages} {5262--5277} (\bibinfo {year} {2002})},\ \Eprint
  {http://arxiv.org/abs/quant-ph/0204119} {arXiv:quant-ph/0204119} \BibitemShut
  {NoStop}%
\bibitem [{SM()}]{SM}%
  \BibitemOpen
  \href@noop {} {}\bibinfo {note} {The Supplemental Material, containing a
  Mathematica notebook of Schwinger-boson-based calculations and related data,
  can be obtained from the journal publication
  (\href{http://link.aps.org/supplemental/10.1103/PhysRevD.111.074516}{http://link.aps.org/supplemental/10.1103/PhysRevD.111.074516})
  or from GitHub
  (\href{https://github.com/jrstryker/lsh}{https://github.com/jrstryker/lsh})}\BibitemShut
  {NoStop}%
\bibitem [{\citenamefont {Littlewood}\ and\ \citenamefont
  {Richardson}(1934)}]{Littlewood1934GroupCA}%
  \BibitemOpen
  \bibfield  {author} {\bibinfo {author} {\bibfnamefont {Dudley~Ernest}\
  \bibnamefont {Littlewood}}\ and\ \bibinfo {author} {\bibfnamefont
  {Archibald~Read}\ \bibnamefont {Richardson}},\ }\bibfield  {title} {\enquote
  {\bibinfo {title} {Group characters and algebra},}\ }\href
  {https://api.semanticscholar.org/CorpusID:265760757} {\bibfield  {journal}
  {\bibinfo  {journal} {Philosophical Transactions of the Royal Society A}\
  }\textbf {\bibinfo {volume} {233}},\ \bibinfo {pages} {99--141} (\bibinfo
  {year} {1934})}\BibitemShut {NoStop}%
\bibitem [{\citenamefont {Knuth}(1970)}]{knuth1970permutations}%
  \BibitemOpen
  \bibfield  {author} {\bibinfo {author} {\bibfnamefont {Donald}\ \bibnamefont
  {Knuth}},\ }\bibfield  {title} {\enquote {\bibinfo {title} {Permutations,
  matrices, and generalized young tableaux},}\ }\href {\doibase
  10.2140/pjm.1970.34.709} {\bibfield  {journal} {\bibinfo  {journal} {Pacific
  journal of mathematics}\ }\textbf {\bibinfo {volume} {34}},\ \bibinfo {pages}
  {709--727} (\bibinfo {year} {1970})}\BibitemShut {NoStop}%
\bibitem [{\citenamefont {Gopalkrishna~Gadiyar}\ and\ \citenamefont
  {Sharatchandra}(1992)}]{GopalkrishnaGadiyar:1991ah}%
  \BibitemOpen
  \bibfield  {author} {\bibinfo {author} {\bibfnamefont {H.}~\bibnamefont
  {Gopalkrishna~Gadiyar}}\ and\ \bibinfo {author} {\bibfnamefont {H.~S.}\
  \bibnamefont {Sharatchandra}},\ }\bibfield  {title} {\enquote {\bibinfo
  {title} {{The Missing link: Operators for labeling multiplicity in the
  Clebsch-Gordan series}},}\ }\href {\doibase 10.1088/0305-4470/25/3/001}
  {\bibfield  {journal} {\bibinfo  {journal} {J. Phys. A}\ }\textbf {\bibinfo
  {volume} {25}},\ \bibinfo {pages} {L85--L88} (\bibinfo {year}
  {1992})}\BibitemShut {NoStop}%
\bibitem [{\citenamefont {{Cramp{\'e}}}\ \emph {et~al.}(2023)\citenamefont
  {{Cramp{\'e}}}, \citenamefont {{Poulain d'Andecy}},\ and\ \citenamefont
  {{Vinet}}}]{2023CMaPh.400..179C}%
  \BibitemOpen
  \bibfield  {author} {\bibinfo {author} {\bibfnamefont {Nicolas}\ \bibnamefont
  {{Cramp{\'e}}}}, \bibinfo {author} {\bibfnamefont {Lo{\"\i}c}\ \bibnamefont
  {{Poulain d'Andecy}}}, \ and\ \bibinfo {author} {\bibfnamefont {Luc}\
  \bibnamefont {{Vinet}}},\ }\bibfield  {title} {\enquote {\bibinfo {title}
  {{The Missing Label of su$_{3}$ and Its Symmetry}},}\ }\href {\doibase
  10.1007/s00220-022-04596-3} {\bibfield  {journal} {\bibinfo  {journal}
  {Communications in Mathematical Physics}\ }\textbf {\bibinfo {volume}
  {400}},\ \bibinfo {pages} {179--213} (\bibinfo {year} {2023})},\ \Eprint
  {http://arxiv.org/abs/2110.03521} {arXiv:2110.03521 [math.RT]} \BibitemShut
  {NoStop}%
\bibitem [{\citenamefont {{The Sage Developers}}(2022)}]{sagemath}%
  \BibitemOpen
  \bibfield  {author} {\bibinfo {author} {\bibnamefont {{The Sage
  Developers}}},\ }\href {\doibase 10.5281/zenodo.6259615} {\emph {\bibinfo
  {title} {{S}age{M}ath, the {S}age {M}athematics {S}oftware {S}ystem}}}
  (\bibinfo {year} {2022}),\ \bibinfo {note}
  {\url{https://www.sagemath.org}}\BibitemShut {NoStop}%
\end{thebibliography}%
\end{document}